\documentclass[12pt]{article}
\usepackage{amsmath, amssymb, amsthm, enumerate, epsfig, MnSymbol, graphicx, proof, hyperref, fullpage, mathrsfs, tikz, tikz-cd, bbm, titlesec, titletoc, float, framed, mdframed}
\usepackage[T1]{fontenc}

\UseRawInputEncoding

\titleformat{\section}[block]
  { \Large \normalfont\scshape}{\thesection}{1em}{}
  
 \titlespacing{\subsection}{3em}{1em}{1em}{}
  
\titleformat{\subsection}
  {\it \large}{\thesubsection}{1em}{}
  
  \titlecontents*{subsection}[0pt]
{\upshape}{}{}
{, \thecontentspage}[ ][ $\bullet$ ][]

\newmdenv[%
innerlinewidth=2pt, 
]{thick}

\newtheorem{thm}{Theorem}[section]
\newtheorem{prop}[thm]{Proposition}
\newtheorem{lemma}[thm]{Lemma}
\newtheorem{cor}[thm]{Corollary}
\newtheorem{ass}{Assumption}

\theoremstyle{definition}
\newtheorem{defn}[thm]{Definition}
\newtheorem{rem}{Remark}[section]

\newcommand{\sA}{\mathcal{A}}
\newcommand{\sF}{\mathcal{F}}
\newcommand{\sI}{\mathcal{I}}
\newcommand{\sS}{\mathcal{S}}
\newcommand{\sU}{\mathcal{U}}
\newcommand{\sN}{\mathcal{N}}
\newcommand{\sP}{\mathcal{P}}
\newcommand{\sB}{\mathcal{B}}
\newcommand{\sO}{\mathcal{O}}
\newcommand{\sM}{\mathscr{M}}

\newcommand{\sH}{\mathcal{H}}
\newcommand{\Cstar}{\mathcal{C}^*}

\newcommand{\Id}{\mathbbm{1}}
\newcommand{\tr}{\textrm{tr}}

\usetikzlibrary{matrix, arrows, decorations.pathmorphing}
\tikzcdset{every label/.append style = {font = \small}}

\title{}
\author{Jane Panangaden \\ Department of Mathematics and Statistics, McGill University, Montreal}
\date{April 2016}

\begin{document}
\begin{titlepage}
\centering
\vspace*{3 cm}
{\scshape  \LARGE Energy Full Counting Statistics in Return-to-Equilibrium  \par}
\vspace{1 cm}
{\scshape  \large Jane Panangaden  \par}
{\it  \large Department of Mathematics and Statistics \par}
{\it  \large  McGill University, Montreal  \par}
\vspace{2 cm}
{ \large April 2016 \par}
\vspace{1cm}
{ \it A thesis submitted to McGill University in partial fulfillment of the requirements of the degree of Master of Science. \par }
\vspace{1cm}
{\textcopyright Jane Panangaden 2016}

\end{titlepage}
\begin{abstract}
We consider a finite dimensional quantum system $S$ in an arbitrary initial state coupled  to an infinitely extended quantum thermal reservoir $R$ in equilibrium at inverse temperature $\beta$. The coupling is given by a bounded perturbation of the dynamics and the coupling strength is controlled by a parameter $\lambda$. We assume the system $S + R$ has the property of return to equilibrium, which means that after sufficiently long time, the joint system will have reached equilibrium at inverse temperature $\beta$. In this context, we prove a refinement to the first law of thermodynamics, which states that the total energy of the system and reservoir is conserved. Specifically, we define two measures which encode all the information about the fluctuations of the system and reservoir energy when two measurements are made at time $0$ and time $t$. These measures are called the full counting statistics (FCS). We prove weak convergence of the system and reservoir FCS in the double limit $t \rightarrow \infty$ and $\lambda \rightarrow 0$. 
\\
\indent The technical difficulty comes from the fact that the reservoir is infinitely extended. In order to define the reservoir FCS, we write the finite dimensional FCS in terms of a relative modular operator: an object which survives the thermodynamic limit. We then draw on tools from Tomita-Takesaki modular theory to prove the result. 
\end{abstract}
\renewcommand{\abstractname}{R\'esum\'e}
\begin{abstract}
On propose dans ce travail d'\'etudier un syst\`eme quantique $S$ de dimension finie coupl\'e \`a un r\'eservoir de chaleur quantique  $R$ infiniement \'etendu et en \'equilibre \`a temp\'erature $\beta^{-1}$. Le couplage est caus\'e par une perturbation born\'ee de la dynamique et l'intensit\'e du couplage est contr\^ol\'ee par un param\`etre $\lambda$. Supposons que le syst\`eme  $S + R$ manifeste la propri\'et\'e de retour \`a l'\'equilibre, c'est-\`a-dire que, apr\`es longtemps, le syst\`eme combin\'e \'evolue vers un \'etat d'\'equilibre \`a temperature $\beta^{-1}$. Dans ce contexte, on d\'emontre un raffinement du premier principe de la thermodynamique, qui affirme que l'\'en\'ergie totale du syst\'eme est conserv\'ee. Plus pr\'ecis\'ement, on d\'efinit deux mesures qui contiennent les informations compl\`etes des fluctuations de l'\'en\'ergie du syst\`eme et du r\'eservoir lorsque des exp\'eriences sont effectu\'ees aux temps $0$ et $t$. Ces mesures s'appellent les statistiques \guillemotleft{} full counting \guillemotright{} (FCS). On d\'emontre la convergence faible des FCS du syst\`eme et du r\`eservoir dans la limite $t \rightarrow \infty$ et $\lambda \rightarrow 0$. 
\\
\indent La difficult\'e technique se produit du fait que le r\'eservoir est infiniement \'etendu. Pour d\'efinir son FCS, il faut r\'e\'ecrire le FCS de dimension finie en utilisant un op\'erateur modulaire: un objet qui survit la limite thermodynamique. Ensuite, on fait appel aux outils de la th\'eorie modulaire de Tomita-Takesaki pour d\'emontrer le th\'eor\`eme. 
\end{abstract}
\newpage
\begin{center}
{ \scshape  \Large Dedication \par}
\vspace{1 cm}
For Emma, forever ago. 
\end{center}
\vspace{1cm}
\begin{center}
{ \scshape  \Large Acknowledgments}
\end{center}
I would first and foremost like to thank my advisors Prof. Vojkan Jak\v si\'c and Dr. Annalisa Panati. Prof. Jak\v si\'c agreed to take me on as a summer student when I had not yet taken a single course in analysis, and was woefully under-qualified. His faith in my abilities and expert guidance as I ventured into research have been an incredibly valuable gift. Dr. Panati has singlehandedly taught me a huge amount of analysis, starting from the definition of a metric space all the way to modular theory. She has spent so much time teaching me and working with me over the past three years. I am extremely grateful for her investment in my development as a researcher. I would also like to mention Ben Landon, who was a fantastic mentor when I was first starting in this field, the entire mathematical physics group at McGill who have made my working environment enjoyable and friendly, and Renaud Raqu\'epas who fixed all my French grammar mistakes. I would like to thank the mathematics and statistics department for giving me such a stimulating environment in which to learn and study and NSERC for financial support. 
\\
\\
Finally, I have to mention my incredible support network: my parents Laurie and Prakash who have provided me with enriching opportunities throughout my life and supported me unconditionally, my doctor Dr. Waqqas Afif for keeping me functioning and my friends who I love deeply. I particularly need to thank my roommates Adam Bene Watts and Galen Voysey who have become part of my family during my time at McGill, and who were consistently there for me during difficult times. Of course, I can't forget my lovely cat Artemis who sat beside me during the writing of this thesis, usually on my keyboard or on top of my notes. 
\newpage
\tableofcontents
\section{Introduction}
The $0^{\rm th}$-law of thermodynamics asserts that a large, isolated system approaches an equilibrium state characterized by a few 
macroscopic parameters such as temperature. In particular, a small system coupled to a comparatively
large thermal reservoir at some temperature is expected to eventually reach 
its equilibrium state at the reservoir temperature, irrespective of its initial state. 
This specific part of the $0^{\rm th}$-law is often called \emph{return to 
equilibrium}. From a mechanical point of view, return to equilibrium holds if the 
interaction is sufficiently dispersive, which translates into ergodic properties
of the dynamics.
\\
\\
We consider the situation of a small quantum system $S$
coupled to a quantum thermal reservoir $R$. We assume that the joint system $S+R$ has the
property of return to equilibrium. The precise mathematical formulation  of this 
property is given by Assumption (\ref{mixing}). 
Although notoriously difficult to prove, return to equilibrium has been 
established for several physically relevant models (spin-boson model, 
spin-fermion model, electronic black box model, locally interacting fermionic systems)  
\cite{AMa,AJPP1,AJPP2,BFS,BM,dRK,DJ,FMU,FMSU,JOP1,JOP2,JP1,JP3}. 
\begin{center}
\begin{tikzpicture}
\draw (-1,0) circle (0.5);
\filldraw[color=black, fill=gray, very thick] (2.5,0) ellipse (2 and 1.5);
\node (none) at (-1,0) {$S$};
\node[color=white] (none) at (2.5,0.25) {$R$}; 
\node[color=white] (none) at (2.5, -0.25) {$\textrm{temp:} 1/ \beta$};
\draw (-0.5, 0) -- (0.5,0);
\node[above] (none) at (0,0) {$\lambda$};
\end{tikzpicture}
 \end{center}
To study the interaction of the system $S$ with physically relevant reservoirs, it is mathematically convenient to idealize the reservoir as infinitely extended. Such an idealization is common in statistical mechanics. From a mathematical point of view, this idealization is required in order for our model to exhibit the return to equilibrium behaviour we are interested in studying. To describe the reservoir, we therefore use the operator algebra formalism for quantum mechanics. 
\\
\\
We consider the exchange of energy between the system and the reservoir as the joint system relaxes into equilibrium. Let $\lambda$ be a parameter which controls the coupling strength between the system and reservoir. If two measurements are made of the energy of $S$ and $R$ at time $0$ and time $t$, let $\Delta Q_S(\lambda, t)$ be the increase in expected value of the energy of the system and  $\Delta Q_R(\lambda, t)$ be the decrease in expected value of the energy of the reservoir. The $1^{\rm st}$ law of thermodynamics is a conservation of energy law for thermal systems. In this setting it can be formulated as:
\[ \Delta Q_S := \lim_{\lambda \rightarrow 0} \lim_{t \rightarrow \infty} \Delta Q_S(\lambda, t) =\lim_{\lambda \rightarrow 0} \lim_{t \rightarrow \infty} \Delta Q_R(\lambda, t) =: \Delta Q_R. \]
Note that it is crucial that we take the limit $\lambda \rightarrow 0$ after taking the large-time limit, as we are making measurements of the system energy and reservoir energy, which do not sum to the total joint system energy as long as $\lambda \neq 0$. The above statement can be thought of as the $1^{\rm st}$ law in terms of averages, since it is the change in \textit{expected value} of the measurements which converges. This is well-understood and can be established by using Araki's perturbation theory. We review this result. 
\\
\\
Our goal is to prove a stronger version of the $1^{\rm st}$ law for return to equilibrium systems. Since we are working with quantum systems, each measurement is probabilistic, and the measurement order matters. We define probability measures encoding the full statistics of the fluctuations in energy of the system and reservoir when two measurements are made at time $0$ and $t$. These measures are called the energy full counting statistics (FCS). Full counting statistics in various contexts have been studied since their introduction in \cite{LL}. It is straightforward to define the FCS for the system, $\mathbb{P}_{S, \lambda, t}$. However, because the reservoir is idealized as infinitely extended, the definition of the reservoir FCS, $\mathbb{P}_{R, \lambda, t}$, is more subtle. It requires a reformulation using Tomita-Takesaki modular theory. We write $\mathbb{P}_{R, \lambda, t}$ in the case of a confined reservoir in terms of a relative modular operator. Physically relevant infinitely extended reservoirs can be obtained as a thermodynamic limit of a sequence of confined reservoirs. The relative modular operator survives this limit, giving us a way to define the FCS even in the infinitely extended case. This also gives the relative modular operator a natural physical interpretation. The objects $\Delta Q_{S/R}(\lambda, t)$ are shown to be the first moments of the measures $\mathbb{P}_{S/R, \lambda, t}$. Hence, the full counting statistics contain the information about the average changes in energy, as well as much more information. 
\\
\\
The main result is the following. Under a dynamical assumption which guarantees return-to-equilibrium behaviour (Assumption \ref{mixing}) and a regularity assumption on the coupling (Assumption \ref{regularity}), the weak limits 
\[\mathbb{P}_{S} :=  \lim_{\lambda \rightarrow 0}\lim_{t \rightarrow \infty} \mathbb{P}_{S, \lambda, t} \;\;\;\;\;\;\;\;\;\; \mathbb{P}_{R} :=  \lim_{\lambda \rightarrow 0}\lim_{t \rightarrow \infty} \mathbb{P}_{R, \lambda, t} \]
exist and 
\[ \mathbb{P}_{S} = \mathbb{P}_R. \]
Since $\Delta Q_{S/R}$ are the first moments of $\mathbb{P}_{S/R}$, this is a significant strengthening of energy conservation in this setting. This result has already been announced in \cite{JPPP}. We provide a comprehensive review of the background and tools required, and more details of the proofs. Furthermore, we show that under an additional technical assumption the proof simplifies considerably and we have convergence of all the moments of $\mathbb{P}_{S/R, \lambda, t}$. 
\\
\\
The thesis is structured as follows. Sections (2)-(5) are dedicated to introducing the operator algebra formalism for quantum mechanics and providing an overview of some basic but essential tools. In Section (6), the main results of Araki's theory of perturbation of the equilibrium states are reviewed. This is the tool which was used to understand the $1^{\rm st}$ law in terms of averages, and we will also draw heavily on these results for our refinement. In Section (7) the modular theory which allows us to define the reservoir FCS is introduced. In Section (8) a proof of the $1^{\rm st}$ law in terms of averages is given, the full counting statistics are defined and this definition is motivated. Section (9) is devoted to the statement and proof of the main result. 
\section{Operator Algebra Formalism}
\subsection{Quantum Mechanics for Mathematicians} 
In quantum mechanics, a physical system is represented by a triple $(\sH, H, \psi)$ where $\sH$ is a Hilbert space, $H$ is a self-adjoint operator on $\sH$ and $\psi$ is a norm-one element of $\sH$. $H$ is called the Hamiltonian and describes the time evolution of the system via the Schr\"{o}dinger equation and $\psi$ is a given initial (pure) state. When $\sH$ is finite dimensional we can axiomatize the theory as follows \cite{Ha}. 

\begin{enumerate}[(i)]
\item In the absence of measurement, the state of the system evolves according to the Schr\"{o}dinger equation. 
\begin{align*}
i \frac{\partial}{\partial t} \psi (t) = H \psi (t) 
\end{align*}
This has general solution
\begin{align*}
\psi (t) = e^{-itH} \psi.
\end{align*}
Note that for each $t \in \mathbb{R}$, $e^{-itH}$ is a unitary operator on $\sH$, and the map $t \mapsto e^{-itH} $ is a continuous one-parameter unitary group. 
\item A quantity which can be measured is represented by a self-adjoint element $A \in \sB(\sH)$. By the spectral theorem, $A$ can be decomposed as 
\begin{align*}
A = \sum_{i} \lambda_i P_{\lambda_i}
\end{align*}
where $\lambda_i \in \mathbb{R}$ are the eigenvalues of $A$ and $P_{\lambda_i}$ are the projections onto the corresponding eigenspaces. The possible outcomes of the measurement are $ \{ \lambda_i \}$. If the system is in state $\psi$ the probability of measuring $\lambda_i$ is given by $\langle \psi, P_{\lambda_i} \psi \rangle$. It follows that the expected value of the measurement is $\langle \psi, A \psi \rangle$. 
\item After the measurement, if a value of $\lambda_i$ is observed, the new state of the system is given by 
\begin{align*}
\frac{P_{\lambda_i} \psi}{|| P_{\lambda_i} \psi ||}.
\end{align*}
\end{enumerate}
When studying quantum statistical mechanics we wish to work with an ensemble of pure states. Remaining in the finite dimensional case, we consider a finite set of pure states, $\{ \psi_1, ..., \psi_k \}$ with probabilities $p_1, ..., p_k$. To this we associate an element of $\sB(\sH)$ defined by 
\begin{align}
\label{density}
\rho = \sum_{i=1}^{k} p_i | \psi_i \rangle \langle \psi_i |.
\end{align}
Let $A$ be an observable with spectral decomposition $A = \sum \lambda_i P_{\lambda_i}$, and $\phi_1, ...,\phi_n$ be an orthonormal basis for $\sH$ of eigenvectors of $A$. The eigenspace of $\lambda$ is spanned by a subset of these eigenvectors, say $\{ \phi_1$, ..., $\phi_m \}$ . 
\begin{align*}
\tr{(\rho P_{\lambda})} & =\sum_{i=1}^{n}  \langle \phi_i ,  \rho P_{\lambda}  \phi_i \rangle 
= \sum_{i=1}^{n} \langle \phi_i ,  \rho \sum_{j=1}^{m}  \phi_j \langle \phi_j , \phi_i \rangle \rangle
= \sum_{j=1}^{m} \sum_{i=1}^n \langle \phi_i ,  \rho \phi_j \rangle \delta_{ij} 
= \sum_{j=1}^{m} \langle \phi_j ,\rho  \phi_j \rangle \\
 &= \sum_{j=1}^{m} \langle \phi_j , \sum_{i=1}^{k} p_i  \psi_i  \langle \psi_i  , \phi_j \rangle \rangle 
  = \sum_{i=1}^k p_i \sum_{j=1}^m \langle \psi_i ,  \phi_j \rangle \langle \phi_j , \psi_i \rangle   
 = \sum_{i=1}^k p_i \langle \psi_i, P_{\lambda} \psi_i \rangle
\end{align*}
Recall that $\langle \psi_i , P_{\lambda} \psi_i \rangle$ is the probability of measuring $\lambda$ in the state $\psi_i$. Hence, $\tr(\rho P_{\lambda})$ is the probability of measuring $\lambda$ for the ensemble of states. Furthermore, $\tr(\rho A)$ is the expected value of the measurement. 

From this computation we can also see, by taking $P_{\lambda}$ to be the identity, that $\tr(\rho) = \sum_{i} p_i = 1$. In fact, since this operator is already diagonal, we can see that $p_i$ are the eigenvalues, so all the eigenvalues are contained in $[0,1]$. In the finite dimensional setting all positive trace-one operators are of the form (\ref{density}) such that all the eigenvalues are positive. We therefore define the (mixed) states of the system to be the elements of $\sB(\sH)$ with positive eigenvalues and trace one, and modify the axioms to take into account this more general notion of state. 
\begin{enumerate}[(i')]
\item In the absence of measurement, the mixed state of the systems evolves as
\begin{align*}
\rho(t) = e^{-itH} \rho e^{itH}. 
\end{align*}
It is not hard to check that this is equivalent to evolving each pure state in the ensemble and then forming a new density matrix. 
\item If $A$ is an observable with eigenvalue $\lambda$, the probability of measuring $\lambda$ in state $\rho$ is 
\begin{align*}
\tr(\rho P_{\lambda}). 
\end{align*}
\item After the measurement, the new state of the system is
\begin{align*}
\frac{P_{\lambda} \rho P_{\lambda}}{\tr(P_{\lambda} \rho P_{\lambda})} = \frac{P_{\lambda} \rho P_{\lambda}}{\tr(\rho P_{\lambda})}. 
\end{align*}
The cyclicity of the trace is used. 
\end{enumerate}

Since the only physical (i.e. measurable) quantities are the traces, we may also choose to view the state as fixed and the observables as dynamical. 
\begin{align*}
A(t) & := e^{itH} A e^{-itH} \\
\tr(\rho(t) A)& = \tr (e^{-itH} \rho e^{itH} A) = \tr (\rho e^{itH} A e^{-itH} ) = \tr (\rho A(t) )
\end{align*}
In physics this is called the Heisenberg picture. From a mathematical point of view, the two descriptions are completely equivalent, but it will be more convenient for us to work with the dynamics on the observables rather than the states. 

\subsection{Entropy}
In statistical mechanics an important object is the equilibrium state. An equilibrium state is, in some sense, the maximally chaotic state of a physical system. To be more precise we will need to introduce the concept of entropy. The entropy is a function on the space of mixed states of a system which measures our lack of knowledge about its specific pure state. 
\\
\\
We begin in the classical setting. The space of pure states is given by a finite set $\Omega = \{ \omega_1, ... ,\omega_{|\Omega|} \}$, while the space of mixed states $\sP(\Omega)$ is given by all the probability measures on $\Omega$. We may identify an element in $P \in \sP(\Omega)$ with an element in $\mathbb{R}^{|\Omega|}$ by associating to $P$ the tuple $(P(\omega_1), ..., P(\omega_{|\Omega|}))$. This induces a topology on $\sP(\Omega)$. 
\begin{defn}
The entropy is the real-valued function $\sS$ defined by
\begin{align*}
\sS :  \bigcup_{\Omega \textrm{ s.t. } |\Omega| < \infty} \sP(\Omega) & \rightarrow \mathbb{R}\\
(P(\omega_1), ..., P(\omega_{|\Omega|})) & \mapsto - \sum_{i=1}^{|\Omega|} P(\omega_i) \log(P(\omega_i)).
\end{align*}
\end{defn}
This is a good way to measure our lack of knowledge about a system because it scales correctly when two mixed states are combined in a probabilistic way to form a new mixed state. 
\begin{prop}
Let $\{\Omega_k\}_{k=1}^m$ be a tuple of finite sets and $P_k$ be a probability measure on $\Omega_k$. Let $(p_1, ..., p_m)$ be a probability measure. Then we can define a new probability measure on $\bigsqcup_{k=1}^m \Omega_k$ by $\oplus_{k=1}^m p_k P_k$. The entropy is the unique function from $\bigcup_{\Omega \textrm{ finite} } \sP(\Omega)$ to $\mathbb{R}$ satisfying
\begin{enumerate}
\item For all $p_k, P_k$, 
\[ \sS(\oplus_{k=1}^m p_k P_k) = \sum_{k=1}^m p_k \sS(P_k) + \sS( (p_1, ..., p_m) ).\]
\item For each $\Omega$, $\sS : \sP(\Omega) \rightarrow \mathbb{R}$ is continuous. 
\end{enumerate}
\end{prop}
\begin{rem}
An intuitive way to think about condition $1. $ in this proposition is to view each system $(\Omega_k, P_k)$ as a black box from which you can reach in and pick the ``true" state of the system according to some probability distribution. Then you can think of the larger system created by putting each one of these boxes into a larger box and assigning them some probabilities $p_1, ..., p_m$. Measuring the pure state of this larger system amounts to first picking one of the boxes, and then picking an object from that box. 
\end{rem}
In thermal physics, one often views a pure state as a full list of the positions and momenta of a large system of particles. However, one can only measure some macroscopic parameters (such as temperature). It is assumed that each specific configuration of particles that could correspond to a given set of observed macroscopic parameters is equally likely. Here, the mixed states correspond to different values of the macroscopic measurements (such as temperature). The pure states correspond to a microscopic configuration. Because of our assumption, the entropy of a given mixed state reduces to $\log(N)$ where $N$ is the number of microscopic states which could correspond to the observed macroscopic state. The highest entropy state is the one which has the greatest number of compatible microstates.  However, supposing no measurements have been made, a priori each microstate is equally likely. The macrostate most likely to be observed is also the one with the greatest number of compatible microstates. Hence, the highest entropy state is the one which is most likely to be observed. This is the statistical formulation of the second law of thermodynamics. 
\\
\\
The definition of entropy can be extended to the quantum case. When working over a finite dimensional Hilbert space $\sH$, the density matrix plays the role of the system $(\Omega, P)$. 
\begin{defn}
The quantum entropy is the functional $\sS: \{\rho \in \sB(\sH) | \tr(\rho)=1, \rho \geq 0 \} \rightarrow \mathbb{R}$ defined by 
\[ \sS(\rho) = - \tr(\rho \log\rho) = - \sum_{i} \lambda_i \log(\lambda_i).\]
where $\rho= \sum_i \lambda_i P_{\lambda_i}$ is the spectral decomposition. Note that $\log(\lambda_i)$ is defined since $\rho$ is positive. If we take $\rho$ to be a density matrix defined by $\sum_i p_i |\psi_i\rangle \langle \psi_i |$, then the quantum entropy is the same as the classical entropy of the probability distribution $(p_1, ..., p_n)$. 
\end{defn}
\subsection{Gibbs State}
Now that we have a way to measure the disorder of a state, we can define the equilibrium states. Recall that in thermodynamics we can interpret the maximal entropy state (i.e. the most chaotic state) as the most likely macroscopic state in which to observe the system, supposing that all microscopic states are equally likely. We can also restrict the microscopic states that are accessible by, for example, fixing the total energy of the system. We are going to construct a state that maximizes the entropy functional among states with a fixed energy. 
\begin{defn}
For a finite dimensional quantum system $\sH$ with Hamiltonian $H$, the Gibbs State at inverse temperature $\beta$ is 
\[ \rho_{\beta}  := \frac{e^{-\beta H}} {\tr(e^{-\beta H})}. \]
Note that this is a positive operator with trace one. 
\end{defn}
Readers familiar with classical statistical mechanics will recognize this as the quantum version of the Gibbs canonical ensemble, which describes the probability distribution of microstates for a system in contact with a thermal reservoir. A microstate with energy $E$ has probability $e^{-\beta E}/ Z$, where $Z$ is a normalization constant.  
\\
\\
If we take $\psi_i$ to be a basis of eigenvectors for the Hamiltonian $H$ with eigenvalues $E_i$, the corresponding quantum state is 
\[ \sum_i \frac{e^{-\beta E_i}}{Z} | \psi_i \rangle \langle \psi_i | = e^{-\beta H} / Z. \]
We will now rigorously state the variational principle satisfied by the Gibbs state. Fix a Hamiltonian $H$ and consider the map 
\[  \beta \mapsto E(\beta) := \tr(H \rho_{\beta}).  \]
This is a function of real numbers that takes the inverse temperature $\beta$ to the expected energy when the system is in the $\beta$-Gibbs state. It is continuous and differentiable with derivative 
\begin{align*}
 \frac{d}{d \beta} ( \tr(H \rho_{\beta})) &= \tr\left(H \left(\frac{-He^{-\beta H}}  {\tr(e^{-\beta H})}  - e^{-\beta H}  \frac{\tr(-He^{-\beta H})}{\tr(e^{-\beta H})^2}\right) \right) \\
 &= \tr(\left( -H^2 + H E(\beta)\right) \rho_{\beta}) \\
 & = -\tr((H-E(\beta))^2\rho_{\beta}) \leq 0
 \end{align*}
which shows that $\beta \mapsto  \tr(H \rho_{\beta})$ is everywhere decreasing. In fact it is strictly decreasing unless $H$ is a constant multiple of the identity. We can then compute the limits
\begin{align*}
\lim_{\beta \rightarrow \mp \infty} \tr(H \rho_{\beta}) &= \lim_{\beta \rightarrow \mp \infty} \sum_i \langle \psi_i , H \frac{e^{-\beta H}}{\tr(e^{-\beta H}) }\psi_i \rangle \\
& = \lim_{\beta \rightarrow \mp \infty} \sum_i \langle \psi_i , E_i \frac{e^{-\beta E_i}}{\tr(e^{-\beta H}) }\psi_i \rangle \\
& = \lim_{\beta \rightarrow \mp \infty} \sum_i E_i \frac{e^{-\beta E_i}}{\sum_j e^{-\beta E_j}} \\
& = E_{\textrm{max}} / E_{\textrm{min}}
\end{align*}
where $E_{\textrm{max}} / E_{\textrm{min}}$ are the largest and smallest eigenvalues. From this we may conclude that for each $E \in [E_{\textrm{min}}, E_{\textrm{max}}]$ there exists some inverse temperature $\beta$ such that $E(\beta) = \tr(H \rho_{\beta}) = E$. Call this number $\beta_E$. We have the following proposition \cite{JOPP}.
\begin{prop}
For $E \in [E_{\min}, E_{\max} ]$ let  $\beta = \beta_E \in \mathbb{R}$ such that $\tr ( \rho_{\beta} H) = E$, as above. Then
 \[ \max \{ S(\rho): \rho \textrm{ a state, }\tr(H \rho) = E \}  = S(\rho_{\beta}).\]
\end{prop}
\begin{proof}
We begin by computing $S(\rho_{\beta})$. Let $H = \sum_i \lambda_i \langle e_i, \cdot \rangle e_i$ be the spectral decomposition of $H$. 
\begin{align*}
S(\rho_{\beta}) & = - \sum_i \frac{e^{-\beta \lambda_i}}{\tr(e^{-\beta H})} \log \left(\frac{e^{-\beta \lambda_i}}{\tr(e^{-\beta H})}\right)\\
& = \beta \sum_i \lambda_i \frac{e^{-\beta \lambda_i}}{\tr(e^{-\beta H})}  + \sum_i \frac{e^{-\beta \lambda_i}}{\tr(e^{-\beta H})} \log(\tr(e^{-\beta H})) \\
&= \beta E + \log(\tr(e^{-\beta H}))
\end{align*}
Now if $\nu$ is some other state such that $\tr(\nu H)=E$, we have the bound 
\begin{align*}
S(\nu) &= S(\nu) - \beta \tr(\nu H) + \beta E \\
& \leq \max_{\rho} \left( S(\rho) - \beta \tr(\rho H) \right) + \beta E.
\end{align*}
The last step is to establish the identity 
\[ \log(\tr(e^A)) = \max_{\rho} ( \tr(\rho A) + S(\rho)) .\]
From this it follows that $S(\nu) \leq \log(\tr(e^{-\beta H})) + \beta E = S(\rho_{\beta})$. To prove the identity, first note that 
\[ \log(\frac{e^A}{\tr(e^A)}) = A - \log(\tr(e^A)) \Id. \]
Therefore, since $\tr(\rho) =1$, 
\begin{align*}
\tr(\rho \log(\frac{e^A}{\tr(e^A)}))  = \tr(\rho A) - \log(\tr(e^A)) \\
\implies \log(\tr(e^A)) - (\tr(\rho A) + S(\rho)) & =  \tr(\rho \log(\rho) - \rho \log(\frac{e^A}{\tr(e^A)})) \\
& \geq \tr(\rho - \frac{e^A}{\tr(e^A)}) = 0 
\end{align*}
by Klein's inequality. Equality holds when $\rho = \frac{e^A}{\tr(e^A)}$.
\end{proof} 
This proposition says that for a given energy $E$, there is a temperature $\beta^{-1}$ such that the $\beta$-Gibbs state maximizes the entropy among all states with energy $E$. It is the most chaotic state at that energy. This can also be interpreted as a stability property due to the second law of thermodynamics. If a system in the Gibbs state is perturbed slightly, it is likely to return to the Gibbs state, since the system will evolve in such a way to increase entropy. This is something stronger than simply saying that the state is invariant under the dynamics. 
\subsection{Trace Class Operators}
The definition of the Gibbs state in the finite dimensional setting relied on using the trace: $\tr(e^{-\beta H})$. When we move to the infinite dimensional setting, we will see that although we can try to extend the definition of this quantity, it may diverge. Therefore, the Gibbs state may not exist. First we need to define the trace in general. To do so we must introduce the concept of the absolute value of an operator. 
\\
\\
For a complex number $z$, we have a unique decomposition into its polar form: $z = |z| e^{i \theta}$, where $|z|$ is a real number and $\theta \in [0, 2\pi)$.  We take $|z|$ to be the definition of the absolute value of $z$. Furthermore we have the identity $ (\overline{z}z)^{1/2} = |z|$. We make use of this identity to develop a similar decomposition for a bounded operator $A$. $A^*$ plays the role of $\overline{z}$. Just as $\overline{z}z$ is always a positive real number, $A^*A$ is self-adjoint and all its expectation values are positive, in the sense that for all $\psi \in \sH$, 
\[ \langle \psi, A^*A \psi \rangle = || A \psi ||^2 \geq 0.\]
Such an operator is called a positive operator. Note that this condition is equivalent to saying that all the eigenvalues are positive in the finite dimensional case. The next proposition says that for positive operators, there is a well-defined (unique) positive square root \cite{BR1}. 
\begin{prop}
Let $A$ be a positive operator. There is a unique positive operator $B$ such that $B^2 = A$. We denote $B$ by $\sqrt{A}$.
\end {prop}
In the finite dimensional case, we can construct this square root easily by using the spectral theorem. If $\sum_i \lambda_i P_{\lambda_i}$ is the spectral decomposition of a positive operator $A$, then each $\lambda_i$ is positive and we can let $B = \sum_i \sqrt{\lambda_i} P_{\lambda_i}$. $B$ is positive and it follows from the fact that the projections are mutually orthogonal that with this definition $B^2 = A$. In general, the construction is more involved, but we will see later that it will also follow from a more general and powerful version of the spectral theorem. 
\\
\\
Finally, we define $|A| := \sqrt{A^*A}$. By the last proposition, $|A|$ is a positive operator, and so all its expectation values are positive. We will use this in the next definition. 
\begin{defn}
Let $A$ be a bounded linear operator over a separable Hilbert space $\sH$ with orthonormal basis $\{ e_k\}$. $A$ is called trace class if 
\[ \sum_k \langle e_k, |A| e_k \rangle < \infty.\]
Note that this sum cannot fail to converge in $[0, \infty]$ since for each $k$, $\langle e_k, |A| e_k \rangle \geq 0$. Furthermore, the sum does not depend on the choice of basis. 
\end{defn}
The trace class operators can be thought of as a non-commutative analogue of $l^1$. They are compact operators, and so can be written in the form 
\[ A = \sum_k \lambda_k \langle \phi_k, \cdot \rangle \phi_k \]
where $A \phi_k = \lambda_k \phi_k$ and the set $\{ \phi_k \}$ is orthonormal. The trace class condition is equivalent to saying that the sequence of complex numbers $\{ \lambda_k \}$ is $l^1$. (See, for example, \cite{Ma}.)
Now, if $A$ is trace class we have the following estimate.
\[ |\sum_k \langle \phi_k, A \phi_k \rangle | \leq \sum_k |\langle \phi_k, A \phi_k \rangle | = 
\sum_k (\langle \phi_k, A \phi_k \rangle \langle A \phi_k, \phi_k\rangle)^{1/2} = \sum_k (\lambda_k\overline{\lambda_k})^{1/2}= \sum_k \langle \phi_k, |A| \phi_k \rangle < \infty\]
In other words, the series which we will define as 
\[ \tr(A) := \sum_k \langle \phi_k, A \phi_k \rangle \]
converges absolutely. From this one can show that the trace is finite and that the trace does not depend on the choice of basis. The trace class operators are those operators for which the trace ``makes sense". It is important to note that in the finite-dimensional case, this definition coincides with the usual definition of trace, and therefore all operators are trace class. 
\\
\\
In general, it is easy to cook up examples for which the Hamiltonian $H$ is a bounded self-adjoint operator, and yet $e^{-\beta H}$ is not trace class. A simple one is to let $H$ be the identity over an infinite dimensional Hilbert space. Therefore, in the infinite setting, the Gibbs state is not always defined, and our formalism is not appropriate for describing the equilibrium behaviour of these systems. 
\subsection{$C^*$-Algebras}
The solution to this problem is to forget about the underlying Hilbert space, and view the algebra of bounded operators as the fundamental object. To do this without referring to a Hilbert space we will need to extract the correct abstract structure from $\sB(\sH)$. First of all, it is a vector space with the pointwise operator addition and scalar multiplication. It has a non-commutative operator multiplication, and an operator adjoint $^*: \sB(\sH) \rightarrow \sB(\sH)$, which satisfies the property $A^{**} = A$. Finally, it has some analytic structure in the form of the operator norm, $|| \cdot ||_{op}$, defined by
\[ ||A||_{op} := \sup_{||\psi|| = 1} || A \psi ||. \]
$\sB(\sH)$ is complete in this norm. This property is inherited from $\sH$. If $A_n$ is a Cauchy sequence in the operator norm, it follows that $A_n \psi$ is Cauchy in $\sH$ for each $\psi$. Since $\sH$ is complete, $A_n \psi$ converges. Call this limit $A \psi$. The map that sends $\psi$ to $A\psi$ is bounded and linear and one can check that $A_n$ converges to $A$ in $\sB(\sH)$. 
\\
\\Furthermore, these structures: multiplication, adjoint and norm, behave well together. 
\begin{prop} Let $A, B \in \sB(\sH)$. 
\begin{enumerate}
\item multiplication + adjoint: $ (AB)^* = B^*A^*$
\item norm + multiplication: $||AB|| \leq ||A|| \cdot ||B||$
\item norm + adjoint: $||A^*A|| = ||A||^2$
\end{enumerate}
\end{prop}
\begin{proof}
\begin{enumerate}[(1)]
\item In $\sB(\sH)$ the adjoint of an operator $A$ is defined as the unique operator $A^*$ such that for all $\phi, \psi \in \sH$,
\[ \langle A^* \phi, \psi \rangle = \langle \phi, A \psi \rangle. \]
We compute $\langle B^*A^* \phi, \psi \rangle = \langle A^* \phi, B \phi \rangle = \langle \phi, AB \psi \rangle$. Hence $(AB)^* = B^*A^*$. 
\item We use the definition of the operator norm to estimate, 
\[ ||AB|| =\sup_{||\psi|| =1} ||AB\psi|| \leq \sup_{||\psi|| =1} ||A|| \cdot ||B\psi|| = ||A|| \cdot ||B||. \]
\item For this proof we will use the construction of the operator adjoint for a general Hilbert space. The construction is based on a version of Riesz's Hilbert space representation theorem, which states that if $u: \sH \times \sH \rightarrow \mathbb{C}$ is a bounded sesquilinear form, then there is a unique $B \in \sB(\sH)$ such that 
\[ u(\phi, \psi) = \langle B \phi, \psi \rangle \]
for all $\phi, \psi \in \sH$. Furthermore, $||B|| \leq ||u||$ where $||u|| = \inf \{ m : |u(\phi, \psi)| \leq m ||\phi|| \cdot ||\psi|| \}$. To construct the adjoint of $A$, we take the bounded sesquilinear form 
\[ u_A(\phi, \psi) = \langle \phi, A \psi \rangle.\]
We then take $A^*$ to be the unique operator such that $\langle A^* \phi, \psi \rangle = u_A(\phi ,\psi)$ from Riesz's theorem. 
\\
\\
Since $||u_A|| = ||A||_{\sB(\sH)}$, we have that $||A^*|| \leq ||A||$. Also, since $A^{**}=A$, $||A||= ||A^{**}|| \leq ||A^*||$ and hence $||A|| = ||A^*||$. Combining this with part (2) gives that $||A^*A|| \leq ||A||^2$. On the other hand, for $\psi \in \sH$,
\[ ||A\psi||^2 = |\langle A\psi, A\psi \rangle | = \langle \psi, A^*A, \psi \rangle \leq ||\psi|| ||A^*A \psi|| \leq ||\psi||^2 ||A^*A||.\]
Hence
\[ ||A||^2 = \sup_{\psi \neq 0} \frac{||A \psi||^2}{||\psi||^2} \leq ||A^*A||. \]
\end{enumerate}
\end{proof} 
Motivated by this discussion, we define a $\Cstar$-algebra as follows. 
\begin{defn}
Let $\sA$ be a vector space over $\mathbb{C}$. 
\\
$\sA$ is an \textbf{algebra} if it is equipped with a multiplication $\times: \sA \times \sA \rightarrow \sA$ satisfying associativity and distributivity, i.e. compatibility with the vector space structure. 
\begin{enumerate}[(i)]
\item $A(BC) = (AB)C$
\item $A(B+C) = AB + AC$
\item $\alpha \beta AB = (\alpha A)(\beta B)$
\end{enumerate}
An algebra is called a \textbf{$^*$-algebra} if it is equipped with an involution, which is a map $^*: \sA \rightarrow \sA$ satisfying $A^{**} = A$, and the involution is compatible with the algebra structure. 
\begin{enumerate}[(i)]
\setcounter{enumi}{3}
\item $(AB)^* = B^*A^*$
\item $(\alpha A+ \beta B)^* = \overline{\alpha}A^* + \overline{\beta} B^*$
\end{enumerate}
An algebra is called a \textbf{Banach algebra} if it is equipped with a norm, is complete with respect to this norm and the norm is compatible with the algebra structure. 
\begin{enumerate}[(i)]
\setcounter{enumi}{5}
\item $\| AB \| \leq \|A \| \|B\|$
\end{enumerate}
If $\sA$ is a Banach algebra and a $^*$-algebra, then it is called a \textbf{Banach-$^*$-algebra} if the norm and involution satisfy the following compatibility condition.
\begin{enumerate}[(i)]
\setcounter{enumi}{6}
\item $\|A^*\| = \|A\|$
\end{enumerate}
Finally, a Banach-$^*$-algebra is a \textbf{$\Cstar$-algebra} if the norm and involution satisfy the $\Cstar$ condition. 
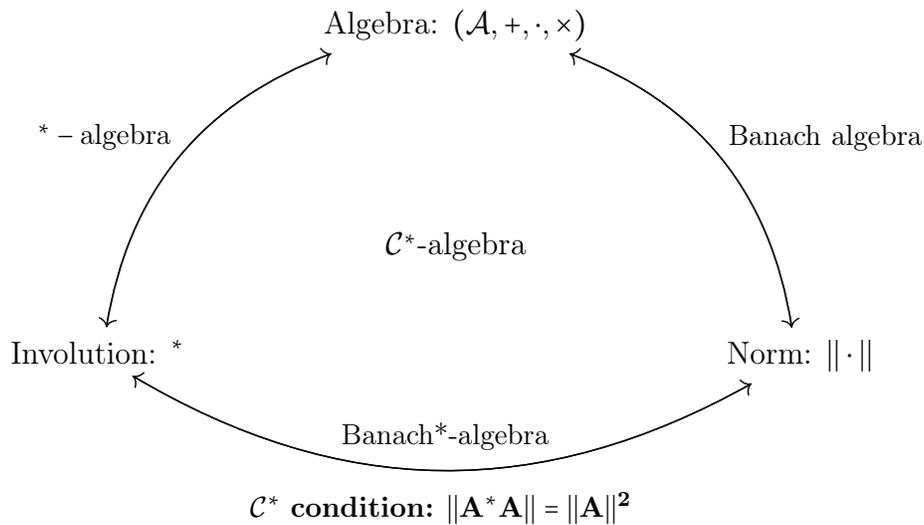
\begin{enumerate}[(i)]
\setcounter{enumi}{7}
\item $\| A^*A \| = \|A \|^2$
\end{enumerate}
\end{defn}

\begin{center}
\begin{tikzcd}[column sep= large]
						&		\textrm{Algebra:  } (\sA, +, \cdot, \times)	\arrow[bend right]{ddddl}[swap]{^*-\textrm{algebra}}	\arrow[bend left]{ddddr}{\textrm{Banach algebra}}		&			\\
						&														&			\\
						& 	&	\\
						&	\Cstar \textrm{-algebra}&	\\
\textrm{Involution:  }^*	\arrow[bend left]{uuuur}	\arrow[bend right]{rr}[ yshift=1ex]{\textrm{ Banach*-algebra} }[swap, yshift=-1ex]{\Cstar \textrm{ \textbf{condition:}  } \mathbf{||A^*A|| = ||A||^2}}    &																&\arrow[bend left]{ll} \arrow[bend right]{uuuul}	\textrm{Norm:  } || \cdot ||		\\
\end{tikzcd}
\end{center}
The $\Cstar$-condition is the linchpin which connects the algebraic and analytic structures. For one thing, it is needed to prove the following structure theorem. We have just shown that every space of bounded linear operators over a Hilbert space is a $\Cstar$-algebra. The structure theorem for $\Cstar$-algebras gives the converse; all $\Cstar$-algebras inherit their structure from such a space of bounded linear operators \cite{BR1}. 
\begin{thm}
Every $\Cstar$-algebra $\sA$ is isomorphic to a norm-closed, self-adjoint (closed under taking adjoints) algebra of bounded operators on a Hilbert space.
\end{thm}
Furthermore we will see later that in a $\Cstar$-algebra, the norm is actually completely determined by the algebraic structure. The next section is devoted to developing an important tool in the analysis of $\Cstar$-algebras: the spectral theory, which will allow us to explore this interplay between algebraic and analytic structure further. 
\section{Spectral Theory for $\Cstar$-algebras}
Having established that a $\Cstar$-algebra is the correct structure to generalize $\sB(\sH)$, we now need to reformulate the axioms of quantum mechanics in this setting. Importantly, we no longer have a Hilbert Space on which our operators act. Before, we said that the set of possible outcomes of a measurement was given by the eigenvalues of a bounded self-adjoint operator. Without a Hilbert space, we do not have a notion of eigenvalues and eigenvectors. However, we can rewrite the eigenvalue-eigenvector equation in finite dimensions as a purely algebraic property of the operator. For $A \in \sB(\sH)$ we have that $\lambda$ is an eigenvalue of $A$ if 
\[ \exists \psi \in \sH \;\;\; \textrm{s.t.}\;\;\;  A \psi = \lambda \psi \iff \exists \psi \in \sH \;\;\; \textrm{s.t.}\;\;\;  (A - \lambda \Id) \psi = 0 \iff (A - \lambda \Id)^{-1} \;\; \textrm{does not exist.}\]
Note that the last equivalence only holds for $\textrm{dim}(\sH) < \infty$. This motivates the following definition. 
\begin{defn}
Let $A$ be an element of a $\Cstar$-algebra $\sA$. The spectrum of $A$ is 
\[ \sigma(A) = \{ \lambda \in \mathbb{C} : (A - \lambda \Id)^{-1} \textrm{does not exist} \}. \]
\end{defn}
We have just argued that in the finite-dimensional case the spectrum of an operator is exactly the set of its eigenvalues. In general, the spectrum can have a richer structure. This is related to the fact that a bounded linear operator on a finite-dimensional vector space is injective if and only if it is surjective, but that this is \textit{not} true in infinite dimensions. The eigenvalue condition can be viewed as the statement that $(A -  \lambda \Id)$ fails to be injective. However, even if it is injective, $(A - \lambda \Id)$ could fail to be invertible because either: $(1)$ it is not bounded below and hence the set-theoretic inverse is not a bounded operator, or $(2)$ the range is not dense. 
\\
\\
Of course, in the $\Cstar$-algebraic setting we do not refer to the domain and range of an element, because the algebra is defined abstractly and not as operators acting on an underlying space. However, keeping in mind the structure theorem for $\Cstar$-algebras, the preceeding discussion should give a feel for what the spectrum will look like. In particular, it could be larger than the set of eigenvalues and may not be discrete. 

\subsection{Spectral Radius}
In a $\Cstar$-algebra, the spectrum is one tool we can use to explore the interplay between the algebraic and analytic structure. A priori, it looks like something that depends only on the algebraic structure, since one just needs to check whether or not certain elements are invertible. However, we will see that considering the spectrum of an element as a set in $\mathbb{C}$ and looking at its size will reveal information about the norm of that element. 

\begin{defn}
Let $A \in \sA$ be an element of a $\Cstar$-algebra. The spectral radius of A is
\[R(A) = \sup \{ |\lambda| : \lambda \in \sigma(A)\} .\]
\end{defn}
In the finite-dimensional setting the spectral radius of an operator is simply the norm of its largest eigenvalue. It is also not hard to see from the definition that the operator norm of a self-adjoint (or more generally normal) operator is also the largest eigenvalue. One uses the finite dimensional spectral theorem to choose an orthonormal basis $\{e_i\}$ of eigenvectors of $A \in \sB(\sH)$ and then expresses an arbitrary norm one element of the Hilbert space as a linear combination of these vectors $\psi = \sum_i c_i e_i$. The coefficients satisfy $\sum_i | c_i |^2 = 1$. Then
\[ ||A \psi ||^2 = || \sum_i c_i \lambda_i e_i ||^2 = \left( \sum_i |c_i \lambda_i |^2 \right)^2.\]
When we sup over all the possible coefficients we see this is maximized when the coefficient corresponding to the largest eigenvalue is $1$ and all others are $0$. This shows that the spectral radius of an operator is equal to the operator norm. In fact, this holds in general as well. We define normal, self-adjoint and unitary operators abstractly using the involution. 
\begin{defn}
An element of a $\Cstar$-algebra $A \in \sA$ is called normal if $A^*A = AA^*$, unitary if $A^*A=AA^* = \Id$ and self-adjoint if $A^* = A$. In particular, unitary and self-adjoint operators are normal. 
\end{defn}
Then we have the following important result relating the $\Cstar$ norm to the spectral radius. 
\begin{prop}
For a normal element of a $\Cstar$-algebra $A \in \sA$, 
\[ R(A) = \|A\| .\]
\end{prop}
\begin{proof}
The main difficulty of this proof is to establish the identity $R(A) = \lim_{n \rightarrow \infty} ||A^n||^{1/n}$. We begin with the upper bound. Fix $\lambda$ such that $|\lambda| > ||A^n||^{1/n}$ for some $n$. Consider the series
\[ \lambda^{-1} \sum_{m \geq 0} \left( \frac{A}{\lambda} \right)^m.\] 
Decomposing each $m \in \mathbb{Z}$ as $m = pn + q$ with $0 \leq q <n$ and using the fact that $\left( \|\frac{A}{\lambda} \|\right)^n < 1$ we can establish that this series is Cauchy and hence converges since the $\Cstar$-algebra is complete. Next we compute, 
\[ ( \lambda \Id - A) \lambda^{-1} \sum_{m \geq 0} \left( \frac{A}{\lambda} \right)^m =  \sum_{m \geq 0} \left( \frac{A}{\lambda} \right)^m- \sum_{m \geq 1} \left( \frac{A}{\lambda}\right)^m = \Id .\] 
Hence the limit of the series is the inverse of $(\lambda \Id - A)$ and $\lambda \nin \sigma(A)$. It follows that $R(A) \leq \|A^n\|^{1/n}$ for all $n$ and hence
\[ R(A) \leq \inf_{n} \|A^n\|^{1/n} \leq \liminf_{n} \|A^n\|^{1/n} .\]
Now we establish the lower bound. Define $r_A := \limsup_{n} \|A^n\|^{1/n}$. Consider the case that $A$ is invertible. Then $ 1 = \|A^n A^{-n}\| \leq \|A^n\|\cdot \|A^{-n}\|$. Taking the limsup we get that $ 1 \leq r_A r_{A^{-1}}$. Hence, $r_A > 0$. The contrapositive of what we have just proven is that if $r_A = 0$, then $A$ is not invertible and hence $ 0 \in \sigma(A)$. Hence $R(A) \geq 0 = r_A$. 
\\
\\
We are left with the case that $r(A) >0$. In the proof of this case, we will use the fact that for a sequence of operators $A_n$ such that $(\Id - A_n)^{-1}$ exists, 
\[ \|A_n\| \rightarrow 0 \iff \| \Id - (\Id - A_n)^{-1} \| \rightarrow 0. \]
To establish this, let $R_n = (\Id - A_n)^{-1}$. Then 
\[ \Id - R_n = \Id - (\Id - A_n)^{-1} = (\Id - A_n)(\Id - A_n)^{-1} - (\Id - A_n)^{-1} = -A_n (\Id - A_n)^{-1} \]
and  
\[ A_n = \Id - R_n^{-1} = R_n R_n^{-1} - R_n^{-1} = (R_n - \Id)R_n^{-1} = -( \Id - R_n)(\Id - (\Id - R_n))^{-1}.\]
From these expansions we can see that $ \|A_n \| \rightarrow 0$ if and only if $\|\Id - R_n\| \rightarrow 0$. 
\\
\\
Now we define $S_{A} = \{ \lambda : |\lambda| \geq r_A \}$. Suppose for contradiction that $S_A \cap \sigma(A) = \phi $. Then for $\omega$ a primitive $n^{th}$ root of unity the operator 
\[ R_n(A, \lambda) = \frac{1}{n} \sum_{k=1}^n \left( \Id - \frac{\omega^kA}{\lambda}\right)^{-1} \]
exists whenever $\lambda \in S_A$. We then calculate
\begin{align*}
& \left( \Id - \frac{A^n}{\lambda^n} \right)  \frac{1}{n} \sum_{k=1}^n \left( \Id - \frac{\omega^kA}{\lambda}\right)^{-1} \\
= & \frac{1}{n}  \left[ \prod_{j=1}^n \left( \Id - \omega^j \frac{(\omega A)}{\lambda} \right)\right]  \sum_{k=1}^n \left( \Id - \frac{\omega^kA}{\lambda}\right)^{-1} \\
& = \frac{1}{n} \sum_{k=1}^n (\Id - \frac{\omega A}{\lambda} ) \cdot \cdot \cdot \widehat{(\Id - \frac{\omega^k A}{\lambda})} \cdot \cdot \cdot (\Id - \frac{\omega^nA}{\lambda}) \\
& = \frac{1}{n} (n \Id) = \Id 
\end{align*}
where we have used the identity $\sum_{i=1}^n \omega^i = 0$ for $n^{th}$ roots of unity. We have established that $R_n(A, \lambda) = (\Id - \frac{A^n}{\lambda^n})^{-1}$. In particular this inverse exists. Now we have the estimate
\begin{align*}
& \left \| \left( \Id - \frac{\omega^kA}{r_A}\right)^{-1} - \left( \Id - \frac{\omega^kA}{\lambda} \right) ^{-1} \right\| \\
& = \left \| \left( \Id - \frac{\omega^kA}{r_A}\right)^{-1} \omega^k A \left( \frac{1}{\lambda} - \frac{1}{r_A} \right)\left( \Id - \frac{\omega^kA}{\lambda} \right) ^{-1}  \right \| \\
& \leq \lvert \lambda - r_A \rvert \| A \| \sup_{\gamma \in S_A} \|(\gamma \Id - A)^{-1} \|^2
\end{align*}
The supremum is finite because we have 
\[ \| (\gamma \Id - A)^{-1} \| \leq \lvert \gamma \rvert^{-1} \sum_{n \geq 0} \frac{\|A^n \|}{\lvert \gamma^n \rvert} = (\lvert \gamma \rvert - \|A\|)^{-1}\]
using the formula for $(\gamma \Id - A)^{-1}$ when $\gamma \geq r_A$ established earlier. We have that the map $\lambda \mapsto (\lvert \lambda \rvert - \|A\|)^{-1}$ is continuous on the complement of $\sigma(A)$ and since we have assumed that $S_A$ has no intersection with $\sigma(A)$, it is continuous on $S_A$. 
\\
\\
Note also that the bound is uniform in $k$, so we may use the formula for $(\Id - \frac{A^n}{\lambda^n})^{-1}$ to obtain that for all $\epsilon > 0$ there exists $\lambda > r_A$ such that
\[ \left\| \left( \Id - \frac{A^n}{r_A^n}\right) ^{-1} - \left(  \Id - \frac{A^n}{\lambda^n}\right)^{-1} \right\| < \epsilon \]
for all $n$. For $\lambda > r_A$ we have that $\|A^n\|^{1/n} / \lvert \lambda \rvert <1$ and so $\|A^n\| / \lvert \lambda \rvert^n \rightarrow 0$ as $n \rightarrow \infty$. It follows that $\|(\Id - A^n/\lambda^n)^{-1} - \Id \| \rightarrow 0$ and hence by the previous statement that $\| ( \Id - A^n/r_A^n)^{-1} - \Id \| \rightarrow 0$. This implies that $\|A^n\|/r_A^n \rightarrow 0 $. But $r_A = \limsup_n \| A^n \|^{1/n}$ by definition, so this last statement is a contradiction. Therefore, there is an element $\lambda$ in the spectrum of $A$ which is in $S_A$ and hence satisfies $| \lambda | > r_A$ and 
\[ R(A) \geq \limsup_n \| A^n \|^{1/n}. \]
The last step is to assume that $A$ is normal and repeatedly apply the $\Cstar$ identity. 
\begin{align*}
\|A^{2^n}\|^2 = \|(A^*)^{2^n} A^{2^n} \| =  \| (A^*A)^{2^n} \| = \|(A^*A)^{2^{n-1}} \|^2 = \cdot \cdot \cdot = ||A||^{2^{n+1}}
\end{align*}
Then we have $R(A) = \lim_{n \rightarrow \infty} \| A^{2^n} \|^{2^{-n}} = \lim_{n \rightarrow \infty} \left( \| A \|^{2^{n+1}} \right)^{2^{-(n+1)}} = \| A \|$. 

\end{proof}
A remarkable consequence of this result is that the algebraic structure of a $\Cstar$-algebra completely determines the analytic structure. 
\begin{cor}
Let $\sA$ be a $*$-algebra. If there is a norm on $\sA$ which makes $\sA$ into a $\Cstar$-algebra, then this norm is the unique norm which makes $\sA$ into a $\Cstar$-algebra. 
\end{cor}
\begin{proof}
If $A$ is normal in a $*$ -algebra and $\| \cdot \|$ is a $\Cstar$-norm, then $\| A \| = R(A)$, which is completely determined by the algebraic structure since the spectrum of the element does not depend on the norm. Then if $A$ is an arbitrary element we use the $\Cstar$-identity to get  
\[ \| A \| = \| A^* A \| ^{1/2} = R(A^*A)^{1/2}\]
since for any $A$, $A^*A$ is normal. The norm of every element is determined by its spectrum. 
\end{proof}
This illustrates to what extent the analytic and algebraic structures are interconnected in a $\Cstar$-algebra. If you have a $\Cstar$ algebra you never need to ask what norm to put on it, the norm can be deduced from the algebraic part of the $\Cstar$-algebra. A similar result holds for morphisms of $\Cstar$-algebras. If a map between two $\Cstar$ algebras respects all the algebraic structure, then it automatically respects the analytic structure as well. 
\begin{cor} \label{morphisms}
Let $\sA$ and $\sB$ be $\Cstar$-algebras and $\pi: \sA \rightarrow \sB$ be a morphism of $^*$-algebras. Then $|| \pi(A) ||_{\sB} \leq ||A||_{\sA}$ for all $A \in \sA$. Hence $\pi$ is a morphism of $\Cstar$-algebras. 
\end{cor}

\begin{proof}
First consider $A$ self-adjoint. If $\lambda \nin \sigma(A)$ then $(\lambda \mathbb{I} - A)^{-1}$ exists. Then $\pi((\lambda \mathbb{I} - A)^{-1}) = (\pi(\lambda \mathbb{I} - A))^{-1} = (\lambda \mathbb{I}_{\pi(\sA)} - \pi(A))^{-1}$ exists. Therefore, $\sigma(\pi(A)) \subset \sigma(A)$ and so we have that the spectral radius of $\pi(A)$ is bounded by the spectral radius of $A$ and hence $||\pi(A)|| \leq ||A||$. 
\newline
\newline
For general $A$, $A^*A$ is self adjoint and so by the previous result,

\[ ||\pi(A)||^2 = ||\pi(A)^*\pi(A)|| = ||\pi(A^*A)|| \leq ||A^*A|| = ||A||^2. \]
$\pi$ is continuous and hence a morphism of $\Cstar$-algebras. 
\end{proof}
Additionally, this result gives us some important information about the spectra of self-adjoint and unitary operators, which we would like to use to represent the observables and time evolution of a physical system in quantum mechanics. Recall that in the finite-dimensional case the observables are time evolved by
\[ A \mapsto U(t)AU(t)^*\]
where $U(t) = e^{itH}$ is a unitary group. Furthermore, we identified the possible results of measurements with the eigenvalues of the operator representing the measurement. Since these are self-adjoint, the eigenvalues are real. This is the reason why we chose the self-adjoint operators to represent observables: because physical experiments should have real-valued outcomes. Unitary operators can also be characterized by their eigenvalues. All the eigenvalues lie on the unit circle. Looking at the spectral decomposition of an operator we have 
\[ UAU^* = U \left( \sum_i \lambda_i P_i \right) U^* = \sum_i \lambda_i UP_iU^*\]
where $UPU^*$ is a projection since $U$ is unitary. Time evolution has no effect on the spectrum, it only moves around the eigenspaces. These spectral properties of unitary and self-adjoint operators continue to hold in the general case, which allows us to continue to use our axioms of quantum mechanics. 
\begin{cor}
If $U$ is a unitary element of a $\Cstar$-algebra, then the spectrum of $U$ is contained in the unit circle in $\mathbb{C}$. If $A$ is a self-adjoint element, then the spectrum of $A$ is contained in $[ -\|A \|, \|A\| ] \subset \mathbb{R}$. 
\end{cor}
\begin{proof}
Let $U$ be unitary. Using the $\Cstar$ condition we get 
\[ \|U\|^2 = \| U^*U\| = \| \Id \| = 1.\]
Hence, since $U$ is normal, $R(U) = ||U|| = 1$ and $\sigma(U)$ is contained in the unit disc. Now suppose $\lambda \in \sigma(U)$. It follows that $\overline{\lambda} \in \sigma(U^*) = \sigma(U^{-1})$ and hence that $(\overline{\lambda}) ^{-1}\in \sigma(U)$. But then $| (\overline{\lambda}) ^{-1} | = |\lambda^{-1}| \leq 1$ which means that $|\lambda| \geq 1$. Hence the spectrum of $U$ is contained in the unit circle. 
\\
\\
Let $A$ be self-adjoint. In particular it is normal, so $R(A) = \| A\|$. Fix $\lambda$ such that $| \lambda^{-1} | > \|A\|$.  $\lambda^{-1}$ is not in the spectrum, and the inverse of $ \Id + i |\lambda| A = -i |\lambda| (i |\lambda^{-1}| \Id - A)$ exists. Consider the element 
\[ B = (\Id - i |\lambda|A)(\Id + i |\lambda| A)^{-1} .\]
We check that $B$ is unitary. 
\begin{align*}
B^*B &= (\Id - i |\lambda| A^*)^{-1} (\Id + i |\lambda| A^*) (\Id - i |\lambda|A) (\Id + i |\lambda| A)^{-1} \\
& = (\Id - i |\lambda| A)^{-1}  (\Id - i |\lambda|A)(\Id + i |\lambda| A) (\Id + i |\lambda| A)^{-1} \\
& = \Id
\end{align*} 
Then we know that the spectrum of $B$ is contained in the unit circle, so for $\alpha \in \mathbb{C}$, whenever 
\[ \left(\frac{1 - i|\lambda| \alpha}{1 + i |\lambda| \alpha} \Id - U\right) \] 
is invertible we must have
\[ \left\lvert \frac{1 - i|\lambda| \alpha}{1 + i |\lambda| \alpha} \right\rvert = 1  \]
which is true exactly when $\alpha$ is real. We can factor 
\begin{align*}
 (\frac{1 - i|\lambda| \alpha}{1 + i |\lambda| \alpha} \Id - U) &= (\Id + i |\lambda |\alpha)^{-1}\Big{(}(1 - i|\lambda | \alpha)(1 + i |\lambda | A) - (1 +i |\lambda| \alpha)(1 - i |\lambda| A) \Big{)} (\Id + i |\lambda | A)^{-1} \\
 & = 2i |\lambda | (\Id + i |\lambda |\alpha)^{-1}(A - \alpha \Id) (\Id + i |\lambda | A)^{-1}
 \end{align*}
Hence, if $\alpha \in \sigma(A)$ then $(\Id \alpha - A)$ is invertible and hence $\alpha$ is real. 
\end{proof}
We have shown that in a $\Cstar$-algebra, a self-adjoint element has real spectrum, so it still makes sense to identify a physical measurement with a self-adjoint element $A$ and the set of possible outcomes with the set $\sigma(A) \subset \mathbb{R}$. What we need now is a way to put a probability measure on this set given some state for the system. Recall that in the finite-dimensional case we used the spectral theorem to decompose 
\[ A = \sum_i \lambda_i P_{\lambda_i} \]
and then for each density matrix $\rho$ we put a probability measure $p_{\rho}$ on the finite set $\sigma(A) = \{ \lambda_i \}$ by 
\[ p_{\rho}(\lambda_i) = \tr(\rho P_{\lambda_i}). \]
To extend this to the $\Cstar$-algebra setting, we will need a more powerful version of the spectral theorem. 
\subsection{Functional Calculus}
The finite-dimensional spectral theorem can be viewed as a way to take functions of matrices. If $p(x) = \sum_{i=0}^{n} a_i x^i$ is a polynomial and $A \in \sB(\sH)$ self-adjoint has spectral decomposition $A = \sum_i \lambda_i P_{\lambda_i}$ then we can define
\[ p(A) = \sum_{i=0}^{n} a_i \left( \sum_j \lambda_j P_{\lambda_j} \right)^i = \sum_{i=0}^{n} a_i \left( \sum_j \lambda_j^i P_{\lambda_j} \right) = \sum_j p(\lambda_j) P_{\lambda_j}. \]
We can then avail ourselves of the Stone-Weierstrass theorem, which says that a continuous function $f: \mathbb{R} \rightarrow \mathbb{C}$ can be approximated uniformly by a sequence of polynomials $\{ p_n\}$ to define 
\[ f(A) = \lim_{n \rightarrow \infty} p_n(A) \]
where the limit is taken in the operator norm. This limit exists because 
\[ \| p_n(A) - p_m(A) \| = \| \sum_{i} (a_{n, i} - a_{m,i}) A^i \| \leq \sum_{i} |a_{n,i} - a_{m,i}| \|A\|^i \leq \| p_n - p_m \|_{\infty} . \]
Or, looking at the spectral form of $p_n(A)$ we can see that $f(A) = \sum_i f(\lambda_i) P_{\lambda_i}$. 
In this section, we develop a tool which will allow us to do the same thing in the $\Cstar$-algebra case. This is known as the functional calculus form of the spectral theorem \cite{RS1}. 
\begin{thm}[Spectral Theorem, Functional Calculus Version]
Let $A$ be a self-adjoint element of a $\Cstar$-algebra $\sA$, and let $\Cstar(A)$ be the sub-algebra of $\sA$ generated by $A$. Then there exists a map $\Phi_A : C(\sigma(A)) \rightarrow \Cstar(A)$, where $C(\sigma(A))$ is the space of continuous functions on $\sigma(A)$ satisfying
\begin{enumerate}
\item $\Phi_A$ is an isomorphism ,
\item $\Phi_A$ is an isometry (in particular it is continuous) with the sup norm on $C(\sigma(A)) $,
\item $\Phi_A(1) = \Id$, where $1$ is the map $ \mathbb{R} \ni x \mapsto 1 \in \mathbb{R}$,
\item $\Phi_A( id) = A$, where $id$ is the map $\mathbb{R} \ni x \mapsto x \in \mathbb{R}$. 
\end{enumerate}
Furthermore, $\sigma(\Phi_A(f)) = f (\sigma(A))$. This is sometimes called the Spectral Mapping Theorem. 
\end{thm}
The element $\Phi_A(f)$ for a continuous function $f$ should be interpreted as ``taking the function"  of the element $A$. We denote it by $f(A)$. 
\\
\\
To prove the theorem we need to do two things. First, we prove a structure theorem for \textit{commutative} $\Cstar$-algebras, which says that every unital commutative $\Cstar$-algebra is isomorphic to a space of continuous, vanishing at infinity functions on a compact Hausdorff space. This is called the Gelfand representation. We will then apply the Gelfand representation to the algebra generated by a self adjoint element, which is commutative, to prove the spectral theorem. The main technical difficulty is in proving the Gelfand representation. 
\\
\\
We proceed by explicitly constructing the Gelfand representation. Let $\sA$ be a $\Cstar$-algebra. 
\begin{defn}
A \textit{character} of $\sA$ is a surjective (algebraic) homomorphism $\gamma : \sA \rightarrow \mathbb{C}$. Denote by $\widehat{\sA}$ the set of all characters of $\sA$.
\end{defn}
A homomorphism of $\Cstar$-algebras is automatically continuous, so $\widehat{\sA} \subset \sA^*$. Our goal will be to identify $\sA$ with $C^0(\widehat{\sA})$, the space of continuous, vanishing at infinity functions on $\widehat{\sA}$, but we first need to equip $\widehat{\sA}$ with the appropriate topology. As a subspace of $\sA^*$ there are many topologies we could put on it, but the one which turns out to give the desired result is the weak$^*$-topology. 
\begin{defn}
Let $X$ be a Banach space. There is a natural inclusion $\iota: X \hookrightarrow X^{**}$ given by 
\[ \iota(A)(\Lambda) = \Lambda(A)\]
for $A \in X$ and $\Lambda \in X^*$. The weakest topology on $X^*$ which makes every element of $\iota(X) \subset X^{**}$ continuous is called the \textit{weak-*-topology} on $X^*$. 
\end{defn}
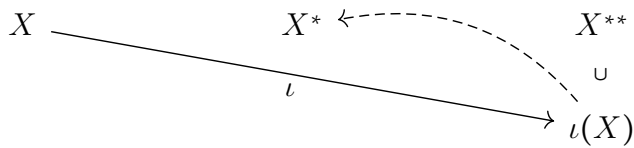
\begin{center}
\begin{tikzcd}[column sep=7em]
		X \arrow{rrd}[swap]{\iota} & X^* &  X^{**} \\
								&	&	 \iota(X) \arrow[draw=none]{u}[sloped, auto=false]{\subset} \arrow[bend right, dashed, sloped]{ul}
\end{tikzcd}
\end{center}
The weak* topology is the topology of pointwise convergence, in the following sense. 
\begin{prop}[MacCluer 5.32 ]
A sequence $\Lambda_n$ in $X^*$ converges in the weak* topology if and only if $\Lambda_n(A)$ converges in $\mathbb{C}$ for all $A \in X$. 
\end{prop}
It follows that  the weak* topology on $X^*$ is Hausdorff. Let $\Lambda \neq \Lambda'$ be two elements of $X^*$. There is some $A \in X$ such that $\Lambda(A) \neq \Lambda'(A)$. These two points in $\mathbb{C}$ can be separated by balls of radius $r$. Let 
\[ U = \{ f \in X^* : f(A) \in B(r, \Lambda(A)) \} \;\;\;\;\;\;\;\;\;\;\; V = \{ f \in X^* : f(A) \in B(r, \Lambda'(A)) \}.\] Since pointwise evaluation is weakly* continuous, $U$ and $V$ are weakly* open and separate $\Lambda$ and $\Lambda'$. Therefore, $ \widehat{\sA}$ is Hausdorff with the weak$^*$ topology. To prove that it is compact we use a well-known result about the weak$^*$ topology
\begin{thm}[Banach-Alaoglu]
Let $X$ be a Banach space. The unit ball of $X^*$ is compact in the weak$^*$ topology. 
\end{thm}
It follows from a spectral argument that $\widehat{\sA}$ is contained in the unit ball of $\sA^*$. We first need a lemma which allows us to characterize the spectrum of $A$ in terms of the characters. 
\begin{lemma}
\label{spectrum Gelfand}
For $A$ in a $\Cstar$-algebra $\sA$, 
\[ \sigma(A) = \{ \gamma(A) : \gamma \in \widehat{\sA} \}.\]
\end{lemma} 
\begin{proof}
The result follows from the fact that there is a correspondence between the set of characters $\widehat{\sA}$ and $M(\sA)$, the set of maximal ideals in $\sA$, given by 
\[ \widehat{\sA} \ni \gamma \leftrightarrow \textrm{ker}(\gamma) \in M(\sA). \]
Suppose $\sI$ is a maximal ideal and let $\gamma: \sA \rightarrow \sA / \sI$ be the quotient map. Since $\sI$ is maximal, $\sA/\sI$ has no proper ideals. Therefore, for any non-zero element $A \in \sA/\sI$, we must have that $A$ is invertible. Otherwise, the ideal generated by $A$ would be a proper ideal. 
\\
\\
Now, fix $A \neq 0$ in $\sA/\sI$. By the formula for the spectral radius $R(A) = \lim_n \|A^n\|^{1/n}$ we know that the spectrum is non-empty. So for some $\lambda$, $(\lambda \Id - A)$ is not invertible. But the only non-invertible element in $\sA/\sI$ is $0$. Hence, $A = \lambda \Id$. This shows that $\sA/\sI = \mathbb{C}$, and hence that $\gamma: \sA \rightarrow \mathbb{C}$ is a character. 
\\
\\
Conversely, suppose $\gamma$ is a character. Then $\mathbb{C} \simeq \sA / \textrm{ker}(\gamma)$. Since $\mathbb{C}$ is a field, $\textrm{ker}(\gamma)$ is a maximal ideal. We have established a bijection between $\widehat{\sA}$ and $M(\sA)$. 
\\
\\
We now fix $A \in \sA$. Suppose $\lambda \in \sigma(A)$. Then $\lambda \Id - A$ is not invertible, and hence it generates a proper ideal of $\sA$, because $\Id$ is not contained in this ideal. We can then take the family of all ideals containing $A$ but not $\Id$ and order them by inclusion. Using Zorn's lemma, we find that $A$ is contained in a proper maximal ideal. Then using the correspondence, there is a character $\gamma$ such that $\lambda \Id - A \in \textrm{ker}(\gamma)$. It follows that $\gamma(A) = \lambda$. 
\\
\\
Conversely, suppose for some character $\gamma$, $\gamma(A) = \lambda$. Then $\gamma(\lambda \Id - A) = 0$, so $\lambda \Id - A \in \ker(\gamma)$, which is a proper maximal ideal by the correspondence, and hence $\lambda \Id - A$ is not invertible, i.e. $\lambda \in \sigma(A)$. 
\end{proof}
A consequence of the preceding lemma is that for any character $\gamma$, 
\[| \gamma(A) | \leq R(A) \leq \| A \|  \]
which implies that $\| \gamma \|_{\sA^*} \leq 1$, and is contained in a weak$^*$ compact set. 
Hence, to show that $\widehat{\sA}$ is weak$^*$ compact, it suffices to show that it is weak$^*$ closed. Let $\gamma_{\alpha}$ be a net in $\widehat{\sA}$ which converges to $\gamma$ in $\sA^*$. We need to show that $\gamma$ is a character. First, $\gamma$ is a homomorphism. 
\begin{align*}
\gamma(AB) &= \lim_{\alpha} \gamma_{\alpha}(AB) = \lim_{\alpha}\gamma_{\alpha}(A)\gamma_{\alpha}(B) = \gamma(A)\gamma(B) \\
\gamma(A^*) &= \lim_{\alpha} \gamma_{\alpha}(A^*) = \lim_{\alpha} \overline{\gamma_{\alpha}(A)} = \overline{(\gamma(A))}
\end{align*}
Since $\gamma$ is a homomorphism, it must be surjective. If $\gamma(\Id) = 0$ then $\gamma$ is trivial since $\gamma(A) = \gamma(\Id)\gamma(A) = 0$. But this contradicts the fact that $\gamma_{\alpha}$ converges to $\gamma$ in the weak$^*$ topology. Otherwise if $\gamma(\Id) = z \neq 0$, then for all $\lambda \in \mathbb{C}$, $\gamma(\lambda z^{-1} \Id) = \lambda$ and hence $\gamma$ is surjective. 
\\
\\
We have shown that $\widehat{\sA}$ with the weak$^*$ topology is a compact Hausdorff space. We will use this topological space to define the Gelfand transform. 
\begin{defn}
Let $\sA$ be a unital, commutative $\Cstar$-algebra. The Gelfand transform is the map $\Gamma$ defined by:
\begin{align*}
\Gamma: \sA 	&\rightarrow C(\widehat{\sA})  	& 			 			&\\
		A 	&\mapsto \widehat{A} 	\;\;\;\;\;\;\;\;\;\;\;\;\;\;\;\;\;\;\;\;\;\;\;\;\;\;\;\;\;\;\;\;\;\;\; \textrm{where}	& \widehat{A} :	 \widehat{\sA} &\rightarrow \mathbb{C} \\
			& 						& 			 \gamma 	& \mapsto \gamma(A) 
\end{align*}
\end{defn}
$\widehat{A}$ is the ``evaluation at $A$" map on the characters, which is continuous by the definition of the weak$^*$ topology. Lemma \ref{spectrum Gelfand} can be reformulated in terms of $\Gamma$ as 
\[ \Gamma(A) (\widehat{\sA}) = \sigma(A) \]
from which it follows that 
\[ \| \Gamma(A) \|_{\infty} = R(A).\]
We can also rewrite the bound $|\gamma(A) | \leq \| A\|$ as the fact that $\| \Gamma (A) \| \leq \| A\|$, so $\Gamma$ is norm decreasing. $\Gamma$ is a homomorphism when we equip $C(\widehat{\sA})$ with pointwise addition and multiplication and with an involution given by complex conjugation. Finally, $\Gamma(\sA)$ clearly separates points in $\widehat{\sA}$ since if $\gamma_1 \neq \gamma_2$ in $\widehat{\sA}$, then for some $A \in \sA$, $\gamma_1(A) \neq \gamma_2(A)$. Hence $\widehat{A}(\gamma_1) \neq \widehat{\sA}(\gamma_2)$. The last ingredient we need to prove the structure theorem for commutative unital $\Cstar$-algebras is the general Stone-Weierstrass theorem. 
\begin{thm}[Stone-Weierstrass ] 
Let $X$ be a compact Hausdorff space and $\sA$ be a self-adjoint ($f \in \sA \implies \overline{f} \in \sA$) sub-algebra of $C(X)$ which contains the constant functions and separates points in $X$. Then $\sA$ is uniformly dense in $C(X)$. 
\end{thm}
\begin{thm}[Structure Theorem for Commutative Unital $\Cstar$-algebras]
Every commutative unital $\Cstar$-algebra $\sA$ is isometrically isomorphic to $C(\widehat{\sA})$. 
\end{thm}
\begin{proof}
Since $\sA$ is commutative, every element is normal. For any $A \in \sA$,
\[ \| \Gamma(A) \|_{\infty} = R(A) = \| A \|. \]
$\Gamma$ is an isometry, and hence it is also injective. $\Gamma(\sA)$ is a sub-algebra of $C(\widehat{\sA})$ which separates points in $\widehat{\sA}$ and contains the constant functions. We just need to show that $\Gamma(\sA)$ is self-adjoint. Suppose $A \in \sA$ is self-adjoint. Then for all $\gamma \in \widehat{\sA}$
\[\Gamma(A)(\gamma) = \widehat{A}(\gamma) = \gamma(A) \in \sigma(A) \subset \mathbb{R}. \]
Hence, as a function on $\widehat{\sA}$, $\Gamma(A) = \overline {\Gamma(A)}$. Now let $A \in \sA$ be arbitrary. It can be decomposed into a linear combination of self-adjoint elements, $A = B + iC$ where $B,C$ are self-adjoint. Then,
\[ \Gamma(A^*) = \Gamma(B - iC) = \Gamma(B) - i \Gamma(C) = \overline{\Gamma(B) + i \Gamma(C)} = \overline{\Gamma(A) }.\]
Hence $\Gamma(\sA)$ is self-adjoint. It follows from Stone-Weierstrass that $\Gamma(\sA)$ is dense in $C(\widehat{A})$, and since $\Gamma$ is an isometry in fact $\Gamma(\sA) = C(\widehat{\sA})$.
\end{proof}
We will now use the structure theorem to prove the spectral theorem. 
\begin{proof}[Proof. (of the Spectral Theorem)]
Fix $A \in \sA$ a normal element of a unital $\Cstar$-algebra. Let $Pol(\Id, A, A^*)$ be the set of all polynomials in the variables $\Id$, $A$, and $A^*$. These three variables commute because of the normality assumption, so $Pol(\Id, A, A^*)$ is commutative. The closure of $Pol(\Id, A, A^*)$ is the smallest sub-algebra containing $A$ and $\Id$, so $\Cstar(A) = \overline{Pol(\Id, A, A^*)}$, and in particular $\Cstar(A)$ is a commutative unital $\Cstar$-algebra. 
Applying the Gelfand's theorem to $\Cstar(A)$ we get that 
\[ \Gamma : \Cstar(A) \rightarrow C(\widehat{\Cstar(A)})\]
is an isometric isomorphism. 
\\
\\
Now let $\sigma(A)$ be the spectrum of $A$ in $\sA$ and $\sigma^*(A)$ be the spectrum of $A$ in $\Cstar(A)$. We have that $\sigma^*(A) \subset \sigma(A)$. By Lemma \ref{spectrum Gelfand} we have that the map 
\[ \widehat{\Cstar(A)} \ni \gamma \mapsto \gamma(A) \in \sigma^*(A)\]
is a surjection. Furthermore it is continuous because $\widehat{\Cstar(A)}$ has the weak$^*$ topology induced by the dual space of $\Cstar(A)$. We will now show that it is injective. Suppose that $\gamma_1(A) = \gamma_2(A)$. Then,
\[ \gamma_1(A^*) = \overline{\gamma_1(A)} = \overline{\gamma_2(A)} = \gamma_2(A^*)\]
and also $\gamma_1(\Id) = 1 = \gamma_2(\Id)$. Therefore, $\gamma_1$ and $\gamma_2$ agree on $Pol(\Id, A, A^*)$ and since they are continuous, they also agree on $\Cstar(A)$. $\gamma \mapsto \gamma(A)$ is a homeomorphism. Using this we define a map
\begin{align*}
\Psi : C(\sigma^*(A)) & \rightarrow C(\widehat{\Cstar(A)}) \\
\Psi(f)(\gamma) &= f(\gamma(A)) \;\;\;\;\;\;\;\;\;\; f \in C(\sigma^*(A)), \gamma \in \widehat{\Cstar(A)}
\end{align*}
which is an isometric isomorphism. Then $\Phi := \Gamma^{-1} \circ \Psi$ is an isometric isomorphism from $C(\sigma^*(A))$ to $\Cstar(A)$. We compute:
\[ \Gamma(A)(\gamma) = \gamma(A) = id(\gamma(A)) = \Psi (id)(\gamma) \]
hence $\Gamma(A) = \Psi(id)$ which implies that $\Phi(id) = A$. Similarly, 
\[ \Gamma(\Id)(\gamma) = \gamma(\Id) = 1 = 1(\gamma(\Id)) = \Psi(1)(\gamma) \]
hence $\Gamma(\Id) = \Psi(1)$ which implies that $\Phi(1) = \Id$. 
\\
\\
The last thing to do to get the desired mapping is to show that $\sigma^*(A) = \sigma(A)$. Let $\lambda \in \sigma^*(A)$ and fix $\epsilon >0$. We can find $f \in C(\sigma^*(A))$ such that $\| f\|_{\infty} = 1$, $0 \leq |f(x)| \leq 1$, $f(\lambda) = 1$ and $f(x) = 0$ whenever $| x - \lambda | > \epsilon$. Then $\Phi(f)$ is an element of $\Cstar(A)$ and we can compute
\[ \| (A - \lambda \Id) \Phi(f) \| = \| \Phi^{-1} ((A - \lambda \Id)\Phi(f)) \|_{\infty} = \| (id - \lambda) f \|_{\infty} \leq \epsilon. \]
We also have that $ \| \Phi(f) \| = \| f \|_{\infty} = 1$, so $\| A - \lambda \Id \| \leq \epsilon$. Hence, $A - \lambda \Id$ cannot be invertible in $\sA$ (if it has a set-theoretic inverse it must be unbounded) and $\lambda \in \sigma(A)$. 
\end{proof}
\begin{proof}[Proof. (of the spectral mapping theorem)]
First let $p(x,y) = \sum_{i,j} c_{ij} x^i y^j$ be a polynomial in two variables with coefficients in $\mathbb{C}$ and $\phi$ be a character. Since $\phi$ is a homomorphism, 
\[ \phi(p(A,A^*)) = \phi( \sum_{i,j} c_{ij} A^i(A^*)^j) =  \sum_{i,j}^n c_{ij} \phi(A)^i \overline{\phi(A)}^j = p (\phi(A), \overline{\phi(A)}). \]
Now let $f$ be a continuous function on $\sigma(A)$. By Stone-Weierstrass, there is a sequence of elements $p_n(A,A^*) \in Pol(\Id, A, A^*)$ converging to $f(A)$. Then since any character $\phi$ is continuous, we have that 
\[\phi(p_n(A,A^*)) \rightarrow \phi(f(A)).\]
 Consider the sequence, $\Phi^{-1}(p_n(A,A^*)) \in C(\sigma(A))$, which exists since $\Phi$ is an isomorphism. We can compute $\Phi^{-1}(p_n(A,A^*))$ concretely using the definition of $\Phi$. First,
\[ \Gamma(p_n(A, A^*)) = \widehat{p_n(A,A^*)} : \gamma \mapsto \gamma(p_n(A,A^*)) = p_n(\gamma(A),\overline{\gamma(A)}).\]
But also, letting $p_n'(x) = p_n(x, \overline{x})$ we have
\[ \Psi(p_n') : \gamma \mapsto  p_n'(\gamma(A)) = p_n(\gamma(A), \overline{\gamma(A)}) .\]
Hence $\Psi(p_n') = \Gamma(p_n(A,A^*))$ which implies that $p_n' = \Phi^{-1}(p_n(A,A^*))$. Since $\Phi^{-1}$ is continuous 
\[p_n' = \Phi^{-1}(p_n(A)) \rightarrow \Phi^{-1}(f(A)) = f \]
uniformly, and so for all $\gamma$ characters, $p_n' (\gamma (A)) \rightarrow f(\gamma(A))$. Hence $ f(\gamma(A)) = \gamma(f(A))$. The spectral mapping theorem follows by applying Lemma \ref{spectrum Gelfand} to $A$ and $f(A)$. 
\end{proof}
We have developed a tool that allows us to take functions of elements of a $\Cstar$ algebra in a way that is compatible with the structure of the algebra of continuous functions, and that is also compatible with the spectral characterization of elements. We now see how this can be used to give the rules for states and measurement in general quantum systems. 
\subsection{Spectral Measures}
We are looking for a way to construct a measure on the spectrum of an element of our $\Cstar$-algebra. One important application of this is to give a probability distribution on the set of outcomes of a measurement in quantum mechanics. However, we will also use this result in other contexts, such as to define the full counting statistics measure which describes the energy fluctuations of the system and the reservoir. The main tool we will use is the Riesz-Markov theorem, which allows us to associate measures to continuous linear functionals on the space of continuous functions. 
\begin{thm}[Riesz-Markov]
Let $X$ be a locally compact Hausdorff space and $\Lambda$ be a continuous linear functional on $C^0(X)$, the space of complex-valued continuous functions on $X$ which vanish at infinity. There is a unique regular complex Borel measure $\mu$ on $X$ satisfying 
\[ \Lambda (f) = \int_X f(x) d\mu(x)  \]
for each $f \in C^0(X)$.
\end{thm}
Now consider a normal element in a $\Cstar$-algebra $A \in \sA$. By using the functional calculus, we may define $f(A) \in \sA$ for each $f \in C^0(\sigma(A))$. Furthermore the map
$f \mapsto f(A)$
is continuous with the uniform topology on $C^0(\sigma(A))$. Therefore, a continuous linear functional $\omega$ on $\sA$ induces a continuous linear functional on $C^0(\sigma(A))$ by 
\[ \Lambda_{A, \omega} (f) : = \omega(f(A)).\]
We can then apply the Riesz-Markov theorem to obtain a measure $\mu_{A,\omega}$ on $\sigma(A)$ satisfying 
\[ \omega(f(A)) = \int_{\sigma(A)} f(x) d\mu_{A, \omega}(x). \]
for all $f \in C^0(\sigma(A))$. $\mu_{A, \omega}$ is called the spectral measure of the pair $(A, \omega)$. In the special case that $\sA = \sB(\sH)$ for some Hilbert space $\sH$, we may consider a linear functional induced by a vector $\psi \in \sH$
\[ \omega_{\psi}(B) = \langle \psi, B \psi \rangle \]
for all $B \in \sA$. The spectral measure of the pair $(A, \omega_{\psi})$ is then denoted by $\mu_{A, \psi}$. 
\\
\\
In this setting, the spectral measure allows us to extend the functional calculus to all bounded Borel functions on the spectrum. For $f$ bounded and Borel on $\sigma(A)$ we define $f(A)$ as a map from $\sH$ to $\sH$ such that 
\[ \langle \psi, f(A) \psi \rangle = \int_{\sigma(A)} f(x) d\mu_{A, \psi} (x) \]
for all $\psi \in \sH$. We may then use the polarization identity to recover $\langle \psi, f(A) \phi \rangle$ for all $\psi, \phi \in \sH$. Then using the Riesz lemma on the sesquilinear form $(\psi, \phi) \mapsto \langle \psi, f(A) \phi \rangle$ we can reconstruct the bounded operator $f(A)$. All the desired properties follow by applying the continuous functional calculus. 
\subsection{Positive Elements and Approximate Identity}
To conclude our discussion of the spectral theory, we will introduce two tools which are useful in practice: the positive elements and approximate identities. \\
\\
The positive elements are a special class of self-adjoint elements. They form a cone inside the $\Cstar$ algebra and can be used to induce an order relation. 
\begin{defn}
Let $A \in \sA$ be an element of a $\Cstar$-algebra. $A$ is called \textit{positive} if $\sigma(A) \subset [0, \infty )$. We denote the set of all positive elements of $\sA$ by $\sA^+$. 
\end{defn}
Since self-adjoint elements have real spectrum, a positive element is automatically self-adjoint. This definition is equivalent to several other conditions. First, we can reformulate it in terms of an analytic property. 
\begin{lemma}
A self-adjoint element $A \in \sA$ is positive if and only if $\| \Id - A \| / \| A \| \leq 1$. 
\end{lemma}
\begin{proof}
Suppose $A$ is positive. Then by the spectral radius, $\sigma(A) \subset [0, \| A \| ]$ and hence by the spectral mapping theorem $\sigma ((\Id - A)/\|A\|) \subset [0, 1]$. This implies that $ \| \Id - A\| /\| A \| \leq 1$, again by the spectral radius. 
\\
\\
Conversely, suppose $ \| \Id - A\| /\| A \| \leq 1$. Then by the spectral radius $\sigma((\Id - A)/\|A \|) \subset [-1,1]$. Then $\sigma(A) \subset [0, 2\|A\|]$ which implies that $A$ is positive. 
\end{proof}
We also have some equivalent algebraic characterizations \cite{BR1}. The last one in particular is useful in practice. 
\begin{lemma}
TFAE:
\begin{enumerate}
\item $A$ is positive,
\item $A = B^2$ for some $B \in \sA$ positive (this $B$ is also unique and we denote it by $\sqrt{A} $),
\item $A = B^*B$ for some $B \in \sA$.
\end{enumerate}

\end{lemma}
We can use the positive elements to put an order relation on the $\Cstar$-algebra by 
\[ A > B \iff A - B \in \sA^+ .\]
The fact that this is an order relation follows from the cone structure of the positive elements: $\sA^+$ is closed under addition and multiplication by positive scalars, and it satisfies $\sA^+ \cap -\sA^+ = \{ 0 \}$. In addition to being an order relation, because of the $\Cstar$-algebra structure $>$ has some additional properties. 

\begin{lemma}
\label{order lemma}
Let $A \geq B \geq 0$. 
\begin{enumerate}
\item $\|A\| \geq \| B \|$
\item  $C^*AC \geq C^*BC \geq 0$ for any $C \in \sA$. 
\item If $\sA$ has an identity and $\lambda > 0$ then $(B + \lambda \Id)^{-1} \geq (A + \lambda \Id)^{-1}.$
\end{enumerate}
\end{lemma}
We will use this order relation to define an approximate identity. 
\begin{defn}
An \textit{approximate identity} of a $\Cstar$-algebra $\sA$ is a net of elements $\{E_{\alpha}\}_{\alpha}$ satisfying 
\begin{enumerate}
\item $\| E_{\alpha} \| \leq 1$,
\item $\alpha \leq \beta \implies E_{\alpha} \leq E_{\beta}$,
\item $\lim_{\alpha} \| E_{\alpha}A - A \| = \lim_{\alpha} \| AE_{\alpha} - A \| = 0$ for all $A \in \sA$. 
\end{enumerate}
\end{defn}
\begin{prop}
Every $\Cstar$-algebra has an approximate identity. 
\end{prop}
\begin{proof}
For a $\Cstar$-algebra $\sA$, let $\Lambda$ be the collection of all finite subsets of $\sA$. Order $\Lambda$ by inclusion. Define the sequence of functions $f_n : [0, \infty) \rightarrow [0,1)$ by
\[ f_n(t) = \frac{nt}{1 + nt}.\]
Suppose $\alpha = \{ A_1, ..., A_n \} \in \Lambda$. Then since $A_1^2 + ... + A_n^2$ is a positive element we can define 
\[ E_{\alpha} := f_n(A_1^2 + ...+ A_n^2).\]
Since $f_n$ is non-negative, by spectral mapping $E_{\alpha}$ is a positive element. Also since $ \| f_n \|_{\infty} \leq 1$ we have that $ \|E_{\alpha} \| \leq 1$. 
Now let $\alpha \leq \beta$ in $\Lambda$ so that $\alpha = \{ A_1, ..., A_m\}$ and $\beta = \{ A_1, ..., A_n \}$ with $m \leq n$. We have that 
\[ \Id + m (A_1^2 + ... + A_m^2) \leq \Id + n (A_1^2 + ... + A_n^2)\]
and by the properties of our order relation it follows that
\[ ( \Id + n (A_1^2 + ... + A_n^2) )^{-1} \leq (\Id + m (A_1^2 + ... + A_m^2))^{-1}. \]
Since $1 - f_n(t) = 1 - \frac{nt}{1 + nt} = \frac{1}{1 + nt}$ it follows that $\Id - E_{\beta} \leq \Id - E_{\alpha}$ and hence $E_{\alpha} \leq E_{\beta}$. 
\\
To prove the last property, fix $A \in \sA$. Then $\{ A \} \in \Lambda$. Fix $\alpha_0 > \{ A \}$ so that $\alpha_0 = \{A_1, ..., A_m \}$ with $A_i = A$ for some $i$. Let $\alpha = \{ A_1, ... A_n \}$ be an arbitrary element such that $\alpha > \alpha_0$, i.e. $n>m$. We have that $A^2 \leq A_1^2 + ... + A_n^2$ and hence by Lemma \ref{order lemma}
\[ (\Id - E_{\alpha})A^2(\Id - E_{\alpha}) \leq (\Id - E_{\alpha})(A_1^2+ ...+A_n^2)(\Id - E_{\alpha}).\]
Define $g_n(t) = (1- f_n(t))t (1 - f_n(t))$. The right hand side of this inequality is $g_n(A_1^2 + ...+A_n^2)$ so we can rewrite is as
\[ (A - AE_{\alpha})^*(A - AE_{\alpha}) \leq g_n(A_1^2 + ...+A_n^2).\]
We can now compute $g_n(t) = \frac{t}{(1+nt)^2}$, which has a maximum at $t = 1/n$. Therefore, $\| g_n \|_{\infty} = 1/4n$ and hence
\[ \| A - AE_{\alpha} \|^2 \leq \frac{1}{4n} \leq \frac{1}{4m}.\]
It follows that $\| A - AE_{\alpha} \| \rightarrow 0$. 
\end{proof}
\section{States and Measurement}
\subsection{Generalized States}
In the finite dimensional setting, the states of a system are given by density matrices. In the more general $\Cstar$-algebraic setting, this needs to be modified as we again run into issues with the convergence of the trace. Given a density matrix $\rho$ over a (finite dimensional) Hilbert space $\sH$, we can construct $\omega_{\rho} $, a linear functional on $\sB(\sH)$. For $A \in \sB(\sH)$,
\[\omega_{\rho} (A) := \tr(\rho A). \]
We have that $\omega_{\rho}$ is positive, in the sense that it maps positive operators to positive numbers. This follows by considering the eigenvalues of $\rho A$, which are all positive since both $\rho$ and $A$ are positive. Furthermore, we may put a norm on the space of linear functionals defined by 
\[||\omega|| = \sup_{||A||=1} |\omega(A)|. \]
This norm can equally well be defined on the space of linear functionals over an abstract $\Cstar$-algebra. Since $\rho$ is trace-one, it follows that $||\omega_{\rho}|| =1$. 
\\
\\
Importantly, these positive linear functionals which are induced by the density matrices are objects which still make sense in the $\Cstar$-algebra setting. Motivated by this discussion, we give the following definition.
\begin{defn}
A linear functional $\omega$ over a $\Cstar$-algebra $\sA$ is \textit{positive} if $\omega(A^*A) \geq 0$ for all $A \in \sA$. \newline
 A positive linear functional $\omega$ with $|| \omega || = 1$ is called a \textit{state}. 
\end{defn} 
Conveniently, the positivity condition automatically gives continuity.  
\begin{prop}
Let $\omega$ be a linear functional over a $\Cstar$-algebra $\sA$. Then if $\omega$ is positive
it is continuous (or equivalently bounded).
\end{prop}
\begin{proof}
Let $\{A_i\}$ be an arbitrary sequence of positive elements with $||A_i|| \leq 1$. Let $\lambda_i$ be an arbitrary sequence of of positive real numbers such that $\sum_i \lambda_i < \infty$. 
\[ || \sum_{i=1}^m \lambda_i A_i - \sum_{i=1}^n \lambda_i A_i || = || \sum_{i=n}^m \lambda_i A_i || \leq \sum_{i=n}^m \lambda_i ||A_i|| \leq \sum_{i=n}^m \lambda_i\]
Since $\sum_i \lambda_i$ converges, $\sum_i \lambda_i A_i$ converges in norm to a positive element $A$, and since each term is positive this convergence is monotone. It follows that
\[ \sum_i \lambda_i \omega(A_i) = \omega(\sum_i \lambda_i A_i) \leq \omega(A) < \infty \]
where we have used the positivity of $\omega$. 
Since $\lambda_i$ was an arbitrary sequence, $\omega(A_i)$ is uniformly bounded and 
\[M = \sup\{ \omega(A) : A\geq 0, ||A|| \leq 1\} < \infty.\]
But any operator can be written as a linear combination of four positive operators using the decomposition 
\[ A = \frac{A+A^*}{2} + i \frac{A - A^*}{2i}\]
where $A_{re} :=  \frac{A+A^*}{2}$ and $A_{im} :=\frac{A - A^*}{2i}$ are self-adjoint and then the further decomposition of each of those terms into 
\[ A_{re/im} = \frac{|A_{re/im}| + A_{re/im}}{2} - \frac{|A_{re/im}| - A_{re/im}}{2}.\]
Hence $||\omega|| \leq 4M < \infty$ and $\omega$ is bounded and hence continuous. 

\end{proof}
A natural question to ask is whether all states defined in this was can be recovered by tracing against an appropriate density matrix. We will see later that this is not the case, and that the states which can be written in such a form, called normal states, have various special properties. However, for a fixed state we can construct something similar by using a specific representation. 
\subsection{Geometry of the States}
It is interesting to look at the geometry of the set of states. First we need a technical lemma. 
\begin{lemma}
For $A,B \in \sA$, and $\omega$ a positive linear functional,
\[|\omega(AB^*)|^2 \leq \omega(A^*A)\omega(B^*B).\]
This is a Cauchy-Schwartz type inequality. Furthermore, if we take an approximate identity $\{ E_{\alpha} \}$, 
\[ || \omega || = \sup_{\alpha} \omega(E_{\alpha}^2). \]
\label{CS lemma}
\end{lemma}
\begin{proof}
To prove the first bound, we consider $\omega((\lambda A+B)^*(\lambda A+B))$, which must be a positive real number for any choice of $\lambda$, by the positivity of $\omega$. It follows by linearity that  
\[ |\lambda|^2 \omega(A^*A) + \overline{\lambda} \omega(A^*B) + \lambda \omega(B^*A) + \omega(B^*B) \geq 0.\]
In particular the above quantity must be real, from which it follows that $\omega(A^*B) = \overline{\omega(B^*A)}$. Writing $\lambda = a + ib$ where $a,b \in \mathbb{R}$ we have that for any choice of $a$ and $b$ the following must hold. 
\[a^2 \omega(A^*A) + 2a \textrm{Re}(\omega(B^*A)) - 2b\textrm{Im}(\omega(B^*A)) + b^2 \omega(A^*A) + \omega(B^*B) \geq 0\]
Calculating the first and second order partial derivatives, we see that the above quantity is positive definite everywhere and has one local minimum at $a = -\textrm{Re}(\omega(B^*A)) / \omega(A^*A)$ and $b = \textrm{Im}(\omega(B^*A)) / \omega(A^*A)$. Therefore, our inequality being satisfied for all $a,b$ is equivalent to the inequality being satisfied at this one point. Substituting these values, the desired bound follows. 
\\
\\
Now we apply the previous inequality to the approximate identity. 
\[ |\omega(AE_{\alpha})|^2\leq \omega(A^*A) \omega(E_{\alpha}^2) \leq M ||A||^2 \omega(E_{\alpha}^2)\]
where as before $M := \sup \{ \omega(A) : A \geq 0, ||A|| \leq 1 \}$. Taking the limit in $\alpha$ and using the continuity of $\omega$ yields
\[ |\omega(A)|^2 \leq M ||A||^2 \sup_{\alpha} \omega(E_{\alpha}^2) \implies ||\omega||^2 \leq M\sup_{\alpha} \omega(E_{\alpha}^2).\]
We also have that for each $\alpha$, $E_{\alpha}^2 < 1$, and hence $\sup_{\alpha} \omega(E_{\alpha}^2) \leq ||\omega||$. From the definition, $M \leq ||\omega||$. It follows that $||\omega|| = \lim_{\alpha} \omega(E_{\alpha}^2)$. 

\end{proof}
An important corollary of this lemma is that for two linear functionals $\omega_1$ and $\omega_2$,  $||\omega_1 + \omega_2|| = ||\omega_1|| + ||\omega_2||$. In particular if $||\omega_{1/2}|| = 1$,
\[ ||\lambda \omega_1 + (1-\lambda) \omega_2 || = \lambda ||\omega_1|| + (1-\lambda) ||\omega_2|| = 1.\]
In other words, the set of states is convex. It is then natural to define the pure states as the extremal points of the set of states, i.e. those which cannot be written as convex combinations of distinct states. This mirrors the classical case when we view the set of mixed states, which are probability measures, as a subset of the hyperplane in $\mathbb{R}^n$.  The set is convex in the euclidean geometry with the extremal points being the pure probability measures. The difference is that in the classical setting there is a unique way to write each mixed state as a linear combination of pure states: the set of states is a simplex. However, in the quantum case we do not have this unique decomposition. 
\\
\\
For example, in a two-level quantum system, given by the Hilbert space $\mathbb{C}^2$, the states are given by $2 \times 2$ matrices over $\mathbb{C}$ with trace-one and positive eigenvalues. Any such density matrix $\rho$ may be decomposed as a linear combination of Pauli matrices and the identity. 
\[ \rho = \mu \Id + \vec{\lambda} \cdot \vec{\sigma}\]
where $\vec{\lambda}= (\lambda_1, \lambda_2, \lambda_3)$ is a vector of coefficients and $\vec{\sigma} = (\sigma_1, \sigma_2, \sigma_3)$ are the Pauli matrices. Since all the Pauli matrices are traceless, we must have that if $\rho$ is trace-one, $\mu=1/2$. We then compute the eigenvalues of $\rho$ in terms of $\vec{\lambda}$ to be $1/2 (1 \pm \| \vec{\lambda} \|)$. Since both eigenvalues must be positive, $\| \vec{\lambda} \| \leq 1$. Therefore the map which associates the trace-one matrix $\rho$ to the vector $\vec{\lambda}$ sends the density matrices to the unit ball. Note also that this map is linear and so preserves the convexity structure. Recall that a state is pure if $\tr(\rho^2) = 1$. Using the eigenvalues we computed, we have that $\tr(\rho^2) = 1/2 (1 + \| \vec{\lambda} \|^2)$. Hence, if $\rho$ is pure, $\| \vec{\lambda} \| = 1$ and is mapped to the unit sphere. 
\\
\begin{align*}
\begin{tikzpicture}
\draw (0,0) -- (0,3);
\draw (0,0) -- (3,0);
\draw (0,0) -- (-1, -2);
\fill[draw=black, fill=gray, fill opacity= 0.2] (0,2)  -- (2,0) -- (-2/3,-4/3) -- (0,2);
\node[above right] (none) at (2,0) {(1,0,0)};
\node[below right] () at (-2/3, -4/3) {(0,0,1)};
\node[above right] () at (0, 2) {(0,1,0)};
\node[circle, draw, inner sep=2pt, fill=black]() at (0,2){};
\node[circle, draw, inner sep=2pt, fill=black]() at (2,0){};
\node[circle, draw, inner sep=2pt, fill=black]() at (-2/3,-4/3){};
\node[] (none) at (1,-2.5) {Mixed states of a};
\node[] (none) at (1,-3.25) {three-level classical system};
\end{tikzpicture}
\;\;\;\;\;\;\;\;\;\;\;\;\;\;\;\;\;\;\;\;\;\;\;&
\begin{tikzpicture}
\filldraw[color=black, fill=gray, fill opacity=0.2, very thick](0,0) circle (1.5);
\draw (0,0) ellipse (1.5 and 0.8);
\draw (0,0) -- (0,3);
\draw (0,0) -- (3,0);
\draw (0,0) -- (-1, -2);
\node[above right] (none) at (1.1,1.1) {$\|\vec{\lambda}\|=1$};
\node[] (none) at (0,-2.5) {Mixed states of a};
\node[] (none) at (0,-3.25) {two-level quantum system};
\end{tikzpicture}
\end{align*}
\\
In both the classical and quantum case the set of states is convex with the extremal points being the pure states. However, in the classical case each state is a unique linear combination of the pure states, while in the quantum case we have lost the uniqueness. There is no canonical way to build the mixed states out of pure states. 
\\
\\
Using again an approximate identity technique and Lemma \ref{CS lemma} we can also obtain some bounds which hold in general for states, which will be useful later. 
\begin{prop}
Let $\omega$ be a positive linear functional over a $\Cstar$-algebra $\sA$. Then for all $A, B \in \sA$. 
\begin{enumerate}
\item $| \omega(A) |^2 \leq \omega(A^*A) ||\omega|| $
\item $|\omega(A^*BA)| \leq \omega(A^*A) ||B|| $
\end{enumerate}
\end{prop}

\subsection{Measurement}
In the finite-dimensional case we had a finite set of outcomes for a measurement: the eigenvalues of a particular observable $A = \sum_i e_i P_i$. If we were in a pure state $\psi$ we obtained a probability measure on this finite set given by $prob(e_i) = \langle \psi, P_i \psi \rangle$. In a mixed state $\rho$ the probability distribution was given by $prob(e_i) = \tr(\rho P_i)$. 
\\
\\
In the general setting we have that our set of possible outcomes for an observable $A \in \sA$ is $\sigma(A) \subset \mathbb{R}$, and our mixed state is given by a positive linear functional $\omega \in \sA^*$. We will use the spectral measure formulation of the spectral theorem to obtain a probability distribution on the set of outcomes. Recall that we have a measure $\mu_{A,\omega}$ which satisfies
\[ \omega(f(A)) = \int_{\sigma(A)} f d\mu_{A, \omega}. \]
In analogy with the finite case, we want to consider our state evaluated on projections associated with certain subsets of $\sigma(A)$. Let $E \subset \sigma(A)$. We can use the Borel functional calculus to define an operator $1_E(A)$ where $1_E$ is the characteristic function of $E$. This generalizes the finite case since if we take $E = \{ e_i \}$ then $1_E(A) = P_i$. This gives that
\[ \omega(1_E(A)) = \int_E d\mu_{A, \omega} = \mu_{A, \omega}(E)\]
which allows us to interpret the spectral measure $\mu_{A, \omega}(E)$ as giving the probability of the measured outcome lying in $E$. In particular, the expected value of a measurement is given by 
\[ \langle A \rangle = \omega(A) = \int_{\sigma(A)} x d\mu_{A, \omega}(x). \]
\section{Dynamics and Equilibrium}
\subsection{Generalized Dynamics}
Recall that in the finite dimensional setting, our dynamics on the observables was given by a Hamiltonian $H$
\[ \tau^t(A) = e^{itH} A e^{-itH}.\]
We can compute the derivative to obtain the differential equation
\[ \frac{d}{dt} \tau^t(A) = iH \tau^t(A) - \tau^t(A) (iH) = i [ H, \tau^t(A)] \] 
with initial condition $\tau^0(A) = A$. A solution to this is 
\[ \tau^t(A) = e^{it[H, \cdot]} A\]
where we view $\delta_H := i[H, \cdot]: A \mapsto i[H,A]$ as a bounded linear operator on the space $\sB(\sH)$, noting that in this setting $\sB(\sH)$ is a finite dimensional Banach space. $\delta_H$ is self-adjoint and is called the generator of $\tau_t$. Importantly, $\delta_H$ is in the $\Cstar$-algebra of operators rather than the underlying Hilbert space, so this is the formulation we will use when moving to the general setting. 
\\
\\
First let us consider what happens if the Hilbert space is infinite dimensional.  We define dynamics using  unitary operators, which preserve the inner product structure. 
\begin{defn}
Let $\sH$ be an infinite dimensional Hilbert space. A dynamics on $\sB(\sH)$ is a map $\tau^t(A) = U(t)AU(t)^*$ where $U(t)$ is a one-paramter strongly-continuous unitary group. In other words,
\begin{enumerate}
\item for each $t$, $U(t)$ is unitary, 
\item for $s,t >0$, $U(t+s) = U(t)U(s)$,
\item $U(0) = \Id$, 
\item for each $A$, $t \mapsto U(t)AU(t)^*$ is continuous. 
\end{enumerate}
\end{defn}
We can apply Stone's theorem to our dynamics, which gives a correspondence between one parameter-strongly continuous unitary groups $U(t)$ and self-adjoint (but not necessarily bounded) operators $A$. 
\[ U(t) \leftrightarrow e^{itA} \]
The exponential is well-defined using a power series if $A$ is bounded or using spectral theory if $A$ is unbounded \cite{RS1}. Therefore, while we can still write our dynamics in the form $\tau^t(A) = e^{itH}Ae^{-itH}$, $H$ may be unbounded and hence not an element of the algebra $\sB(\sH)$. From a physical point of view, this corresponds to an infinitely extended system having infinite energy, but it will pose problems when we move to the $\Cstar$-algebraic formalism, where we do not have access to the Hilbert space.
\\
\\
We can rewrite the definition of a dynamics in terms of the operator algebra by using the fact that unitary operators preserve inner products, and hence conjugation by unitaries preserves operator norm. 
\begin{defn}
A dynamics on a $\Cstar$-algebra $\sA$ is a map 
\[ \mathbb{R} \ni t \mapsto \tau^t \in Aut(\sA) \]
which satisfies 
\begin{enumerate}
\item $\tau^{t+s} = \tau^t \circ \tau^s$,
\item $t \mapsto \tau^t(A)$ is continuous for all $A \in \sA$ (strong continuity).
\end{enumerate}
\end{defn}
For a $\Cstar$-dynamics, which is a one-parameter strongly continuous automorphism group on a Banach space, we define the generator as the possibly unbounded operator $\delta: \mathcal{D}(\delta) \rightarrow \sA$: 
\begin{align*} 
\mathcal{D}(\delta) &= \{ A \in \sA : \lim_{t \rightarrow 0} \frac1t (\tau^t(A) - A ) \textrm{ exists} \}, \\
  \delta (A) &  = \lim_{t \rightarrow 0} \frac1t (\tau^t(A) - A ) \textrm{    for } A \in \mathcal{D}(\delta).
\end{align*}
The following theorem says that the generator does indeed generate the dynamics \cite{Ge}. 
\begin{thm}[Hille-Yosida]
Let $t \mapsto \tau^t$ be a one-parameter strongly continuous group of automorphisms on a Banach space $X$ with generator $\delta$. Then the domain of $\delta$, $\mathcal{D}(\delta)$, is dense and $\delta$ is closed. Furthermore if $A \in \mathcal{D}(\delta)$ then 
\[ \left( t \mapsto \tau^t(A) \right) \in C^0(\mathbb{R}, \mathcal{D}(\delta)) \cap C^1 (\mathbb{R}, X)\]
and $\delta$ satisfies the differential equation 
\[ \frac{d}{dt}\tau^t(A) = \delta \tau^t(A) = \tau^t \delta (A).\]
\end{thm}
For this reason we use the formal notation $\tau^t = e^{t \delta}$. Note that if we have a Hamiltonian dynamics on a finite dimensional Hilbert space given by Hamiltonian $H$, then we may compute the generator $\delta$ in this last more general sense and find that $\delta=\delta_H = i[H, \cdot]$. 
\subsection{KMS States}
Now that we know how to define states and dynamics in general, we need to find a better way of defining the equilibrium states for a given dynamics. First of all, if the Gibbs state exists, it should be an equilibrium state. We also want our equilibrium states to be invariant and to have certain stability properties under small perturbations of the dynamics. First we derive a condition which is equivalent to the Gibbs state in finite dimensions, but which is formulated in terms of linear functionals rather than density matrices. 
\\
\\
We have a Hamiltonian dynamics acting on $\sB(\sH)$ given by $\tau^t(A) = e^{itH}Ae^{-itH}$. For fixed $A,B$ we can extend the function  $\mathbb{R} \ni t \mapsto \omega(A \tau^t(B)) $ to a strip in the complex plane, just by applying the functional calculus to the Hamiltonian. Then the Gibbs state is equivalent to an ``approximate commutation" property \cite{JOPP}. 
\begin{prop}
Let $\sH$ be a finite dimensional Hilbert space. The state $\omega \in \sB(\sH)^*$ is a Gibbs state at inverse temperature $\beta$ if and only if 
\[ \omega(A\tau^{i\beta}(B)) = \omega(BA)\]
for all $A,B \in \sB(\sH)$. This condition is called the \textbf{KMS condition}. 
\end{prop}
\begin{proof}
Suppose $\omega_{\beta}$ is the Gibbs state. 
\[\omega_{\beta}(A\tau^{i\beta}(B)) = \tr(\rho_{\beta} Ae^{-\beta H}Be^{\beta H}) = \tr(\frac{e^{-\beta H}}{\tr(e^{-\beta H})} Ae^{-\beta H}Be^{\beta H}) = \tr(\rho_{\beta} BA) = \omega_{\beta}(BA) \]
The KMS condition follows from the cyclicity of the trace. This implication actually holds as long as the Gibbs state exists, even if $H$ is not finite-dimensional. The traceclass operators are an ideal in the space of bounded operators, and so the above calculation still makes sense. 
\\
\\
Now suppose $\omega$ is a state satisfying the KMS condition. Since we are in the finite dimensional case, we have some density matrix $\rho$ such that $\omega(\cdot) = \tr(\rho \; \cdot)$. The KMS condition can then be written as
\[ \tr(\rho BA) = \tr(\rho Ae^{-\beta H} B e^{\beta H} ) \;\;\;\;\;\; \forall A,B \in \sB(\sH).\]
Let $X =  e^{\beta H} \rho $ and $Y = A e^{-\beta H} $. Then we have 
\[ \tr(XBY) = \tr(XYB) \;\;\;\;\; \forall B, Y \in \sB(\sH).\]
It follows that $[X, B] = 0$ for all $B$ and hence $X = \alpha \Id$ for some scalar $\alpha$. Hence $\rho = \alpha e^{-\beta H}$. Normalizing so that $\tr(\rho)=1$ gives the Gibbs state. 
\end{proof}
\begin{rem}
The KMS condition can be rewritten as 
\[ \omega(\tau^t(B)A) = \omega(A \tau^{t + i\beta}(B)) \;\;\;\;\; \forall t \in \mathbb{R}.\]
In this form we can see it as a boundary condition of the function 
\[ F_{A,B}(z) := \omega(A\tau^z(B))\]
on the strip $0 \leq \textrm{Im}(z) \leq \beta$: $F_{A,B}(t + i\beta) = \omega(\tau^t(B)A)$. 
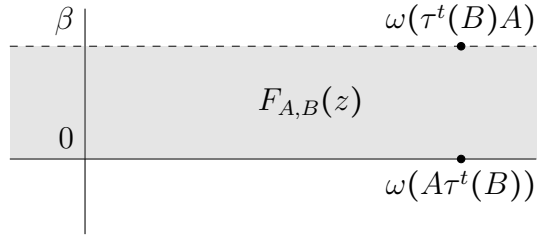
\begin{align*}
\begin{tikzpicture}
\draw (-1,0) -- (6,0);
\draw[dashed] (-1,1.5) -- (6,1.5);
\draw (0, -1) -- (0,2);
\fill[fill=gray, fill opacity = 0.2] (-1,0) -- (6,0) -- (6,1.5) -- (-1,1.5) -- (-1,0);
\node[above left]() at (0,1.5){$\beta$};
\node[above left]() at (0,0){$0$};
\node() at (3,0.75){$F_{A,B}(z)$};
\node[below]() at (5,0){$\omega(A\tau^t(B))$};
\node[above]() at (5,1.5){$\omega(\tau^t(B)A)$};
\node[circle, draw, inner sep=1pt, fill=black]() at (5,0){};
\node[circle, draw, inner sep=1pt, fill=black]() at (5,1.5){};
\end{tikzpicture}
\end{align*}
\end{rem}

Our goal now is to generalize the KMS condition to the operator algebra setting. First we need to do some work to see if, for a general $\Cstar$ dynamics, we can extend the function $t \mapsto \omega(A\tau^t(B))$ analytically to a strip in the complex plane. 
\begin{defn}
Let $t \mapsto \tau^t$ be a strongly continuous, one-parameter group of $*$-automorphisms of a $\Cstar$-algebra $\sA$. An element $A \in \sA$ is analytic for $\tau^t$ if there exists a strip $I_{\lambda} = \{ z \in \mathbb{C} : |\textrm{Im}(z)| < \lambda \}$ and a function $f : I_{\lambda} \rightarrow \sA$ such that 
\begin{enumerate}
\item for $t \in \mathbb{R}$, $f(t) = \tau^t(A)$ and 
\item for each $\omega \in \sA^*$, the function $z \mapsto \omega(f(z))$ is analytic. 
\end{enumerate}
\end{defn}
This ``pointwise analyticity" condition is actually equivalent to a stronger notion of analyticity in the space $\sA$ itself. 
\begin{prop}
If $A \in \sA$ is analytic for $\tau^t$ on the strip $I_{\lambda}$, then for any $z \in I_{\lambda}$ the limit
\[ \lim_{h \rightarrow 0} \frac{f(z+h) - f(z)}{h}\]
exists in norm. 
\end{prop}
\begin{proof}
Let $z \in  I_{\lambda}$ and let $B(z,r)$ be a ball contained in $I_{\lambda}$. Let $C$ be the boundary of this ball. Define $K = B(z, r/2) \subset B(z,r)$. For $x \in K$ and $\omega \in \sA^*$ we have by Cauchy's integral formula
\[ \omega(f(x)) = \frac{1}{2\pi i} \int_{C} \frac{\omega(f(y))}{y-x} dy. \]
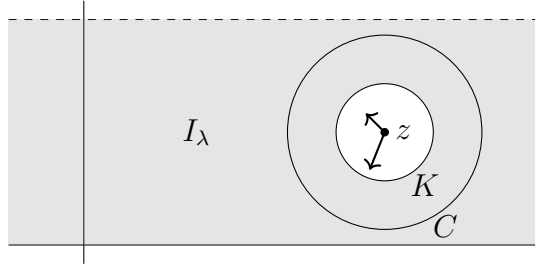
\begin{align*}
\begin{tikzpicture}
\draw (-1,0) -- (6,0);
\draw[dashed] (-1,3) -- (6,3);
\draw (0, -0.25) -- (0,3.25);
\fill[fill=gray, fill opacity = 0.2] (-1,0) -- (6,0) -- (6,3) -- (-1,3) -- (-1,0);
\node() at (1.5,1.5){$I_{\lambda}$};
\node[circle, draw, inner sep=26pt]() at (4,1.5){};
\node[circle, draw, inner sep=13pt, fill=white]() at (4,1.5){};
\node[circle, draw, inner sep=1pt, fill=black]() at (4,1.5){};
\node[right]() at (4,1.5){$z$};
\node[right]() at (4.2,0.8){$K$};
\node[right]() at (4.5,0.25){$C$};
\draw[thick, ->] (4,1.5) -- (3.75,1.75);
\draw[thick, ->] (4,1.5) -- (3.8,1);
\end{tikzpicture}
\end{align*}
Then for $z+h, z+g \in K$ 
\begin{align*}
& \frac{1}{(h-g)} \left( \frac{1}{h} (\omega(f(z+h)) - \omega(f(z))) - \frac{1}{g} (\omega(f(z+g)) - \omega(f(z))) \right) \\
=& \frac{1}{2\pi i (h-g)} \int_{C} \omega(f(y)) \left( \frac{1}{h(y-z-h)} - \frac{1}{g(y-z-g)} - \frac{1}{h(y-z)} +\frac{1}{g(y-z)} \right) dy \\
=&  \frac{1}{2\pi i } \int_{C} \omega(f(y)) (y - z - g)^{-1} (y-z-h)^{-1} (y-z)^{-1} dy \\
\end{align*}
$|y-z-g|$ is bounded below uniformly in $g$ by $r/2$. Hence for fixed $\omega$ the quantity above is uniformly bounded in $g$ and $h$. Consider the family of bounded linear maps indexed on $(h,g)$
\begin{align*}
\Lambda_{h,g} : \sA^* &\rightarrow \mathbb{C} \\
\omega& \mapsto \frac{1}{(h-g)} \left( \frac{1}{h} (\omega(f(z+h)) - \omega(f(z))) - \frac{1}{g} (\omega(f(z+g)) - \omega(f(z))) \right)
\end{align*}
We have shown that for each fixed $\omega$, $\sup_{h,g} | \Lambda_{h,g} (\omega) |< \infty$. Therefore by uniform boundedness, $\sup_{h,g} || \Lambda_{h,g} || < \infty$. It follows that for some $\gamma < \infty$ the following inequality holds for all $h,g$ such that $|z-g| < r/2$ and $|z-h| < r/2$. 
\[  || \frac{1}{h} (f(z+h) - f(z)) - \frac{1}{g} (f(z+g) - f(z)) || \leq \gamma |h-g|\]
Since $\sA$ is complete it follows that the derivative at $z$ exists. 
\end{proof}
Now that we have seen how the analytic elements behave, we will see that we can in fact approximate any element by (entire) analytic elements. To do this, we use the fact that we can define a ``pointwise" integral on $\sA$. 
\begin{prop}
Let $A\in \sA$ and $\mu$ be a Borel measure on $\mathbb{R}$. There exists $B \in \sA$ such that for each $\omega \in \sA^*$ 
\[ \omega(B) = \int \omega(\tau^t(A))d\mu(t).\]
We denote
\[ B = \int \tau^t(A)d\mu(t).\]
\end{prop}
Using this construction we may define a sequence 
\[ A_n = \sqrt{\frac{n}{\pi}}  \int \tau^t(A) e^{-nt^2} dt.\]
To show that $A_n$ is entire analytic we must construct an analytic extension of the function $t \mapsto \tau^t(A_n)$ to the whole complex plane. Define for $z \in \mathbb{C}$
\[ f_n(z) = \sqrt{\frac{n}{\pi}}  \int \tau^t(A) e^{-n(t-z)^2} dt.\]
For $z=s \in \mathbb{R}$ we have
\[ f_n(s) =  \sqrt{\frac{n}{\pi}}  \int \tau^t(A) e^{-n(t-s)^2} dt = \sqrt{\frac{n}{\pi}}  \int \tau^{t+s}(A) e^{-nt^2} dt = \tau^s (\sqrt{\frac{n}{\pi}}  \int \tau^t(A) e^{-nt^2} dt) = \tau^s(A_n).\]
From here, one can show that $f_n(z)$ is an analytic function for each $n$ and hence that each $A_n$ is analytic for $\tau^t$. Finally, one checks that $A_n$ converges to $A$ in norm. 
\\
\\
The norm dense sub-algebra of entire analytic elements (for a given dynamics $\tau^t$) is denoted by $\sA_{\tau}$. To show a set is dense in $\sA$ it is sufficient to show it is dense in $\sA_{\tau}$. We are finally ready to give the general definition of our equilibrium states. 
\begin{defn}
Let $(\sA, \tau)$ be a $\Cstar$-dynamical system. A state $\omega$ is a $(\tau, \beta)$-KMS state if 
\[ \omega(A\tau^{i\beta}(B)) = \omega(BA)\]
for all $A,B$ in some norm-dense, $\tau$-invariant sub-algebra $\sB_{\tau} \subset \sA_{\tau}$. 
\end{defn}

\subsection{Properties of KMS States}
It is not as simple to rewrite the KMS condition as a boundary condition as we have done in the finite-dimensional case. This is because it is not clear that the map $t \mapsto \omega(A\tau^t(B))$ has an analytic extension for arbitrary $A,B \in \sA$, since they may not be analytic elements. However, the next proposition says that this is indeed the case. 
\begin{prop}
Let $(\sA, \tau)$ be a $\Cstar$-dynamical system and $\omega$ a state. Define 
\[ D_{\beta} = \{ z \in \mathbb{C} : 0 < \textrm{Im}(z) < \beta \}.\]
If $\omega$ is a $(\beta, \tau)$-KMS state, then for any $A, B \in \sA$ there exists a complex-valued function $F_{A,B}$ which is bounded and continuous on $\overline{D_{\beta}}$ and analytic on $D_{\beta}$ satisfying the following boundary conditions:
\[ F_{A,B}(t) = \omega(A \tau^t(B))\]
\[ F_{A,B}(t + i\beta) = \omega(\tau^t(B)A)\]\
for all $t \in \mathbb{R}$. 
\end{prop}
This result follows from the density of the analytic elements, and the fact that we already have an entire analytic extension $z \mapsto \tau^z(A)$ satisfying these properties for $A \in \sA_{\tau}$. ([BR 5.3.7] for details) This characterization of KMS states will be useful in practice. 
\\
\\
The next thing to do is check that the KMS states satisfy the most basic property that we desire from an equilibrium state: invariance under the dynamics. 
\begin{lemma}
Let $(\sA, \tau)$ be a $\Cstar$-dynamical system. The set of entire analytic elements $\sA_{\tau}$ is invariant under the dynamics.
\end{lemma}
\begin{proof}
Suppose $A \in \sA_{\tau}$. There is a function $f: \mathbb{C} \rightarrow \sA $ such that $f(t) = \tau^t(A)$ and for each $\omega \in \sA^*$, $z \mapsto \omega(f(z))$ is analytic. Define $g(z) = f(z+s)$ for fixed $s \in \mathbb{R}$. Then 
\[g(t) = f(t+s) = \tau^{t+s}(A) = \tau^t(\tau^s(A)) \] 
for all $t \in \mathbb{R}$. Also for $\omega \in \sA^*$, $z \mapsto \omega(g(z)) = \omega(f(z+s))$ is analytic since it is just a shifted analytic function. Hence $\tau^s(A) \in \sA_{\tau}$ and the analytic elements are invariant under the dynamics. 
\end{proof}
\begin{prop}
Let $\omega$ be a $(\beta, \tau)$-KMS state over $\sA$. Then for all $A \in \sA$, 
\[ \omega(\tau^t(A)) = \omega(A).\] 
\end{prop}
\begin{proof}
Wlog, we assume $\beta = -1$. Let $B \in \sB_{\tau}$, the dense sub-algebra for which the KMS condition holds. Define 
\[ F(z) = \omega(\tau^z(B)).\]
This is an analytic function. We will bound this function on the strip $I = \{ z \in \mathbb{C} : -1 \leq Im(z) \leq 0 \}$. 
\[ | F(z) | \leq \| \tau^z(B) \| = \| \tau^{Rez}(\tau^{Imz}(B)) \| = \| \tau^{Imz}(B) \| \]
Hence 
\[\sup_{z \in I} |F(z) | \leq \sup_{\lambda \in [-1,0]}  \| \tau^{i\lambda}(B) \| .\] 
Since $\lambda \mapsto \| \tau^{i\lambda}(B) \|$ is continuous, this supremum is attained and $F$ is bounded on $I$ by some constant $M$. We now use the KMS condition to show that $F$ is periodic in $i$. 
\[ F(z-i) = \omega(\Id \tau^{-i}(\tau^z(B))) = \omega(\tau^z(B)\Id)=F(z)\]
(If $\sA$ does not have an identity, use the approximate identity.) This implies that $F$ is bounded by $M$ everywhere and hence is constant by Liouville's theorem. 
\end{proof}
Unlike in the finite dimensional case, where there is exactly one equilibrium state (the Gibbs state) for a given dynamics and temperature, in the general setting a KMS state may not exist, and if it does it may not be unique. However, we can study the geometry of the set of KMS states at a fixed temperature. 
\begin{prop}
Let $(\sA, \tau^t)$ be a $\Cstar$-dynamical system, and $\beta \in \mathbb{R}$. The set of $(\beta, \tau)$-KMS states is convex, and in fact it is a simplex.
\end{prop}
A physical interpretation for this result is that the extremal points of the set of KMS states represent different pure phases, with each KMS state having a unique decomposition as a convex combination of these pure phases. To motivate this further and to explore the stability properties of the KMS states under perturbations of the dynamics, we will need to introduce several more tools. 
\section{Representations and Perturbation Theory}
\subsection{The GNS Representation}
Although we have formulated our theory without reference to an underlying Hilbert space in order to extend it to the infinite setting, it is sometimes useful to exploit the structure theorem for $\Cstar$-algebras to work in the more concrete setting of an algebra of bounded operators over a Hilbert space. We refer to this process as working in a representation. 

\begin{defn}
Let $\sA$ be a $\Cstar$-algebra and $\sH$ be a Hilbert space. We say a map $\Pi: \sA \rightarrow \mathcal{B}(\sH)$ is a \textit{representation of $\sA$ in $\sH$} if $\pi$ is a morphism of $\Cstar$- algebras.
\end{defn}

By Corollary \ref{morphisms} we know that a $^*$-algebra morphism between $\Cstar$-algebras is automatically a $\Cstar$-algebra morphism, so we only need to check that the algebraic structure is preserved to find a representation. Furthermore, we know from the structure theorem that a representation always exists. However, not every rep gives an isomorphism between $\sA$ and $\pi(\sA)$. 

\begin{defn}
A representation $\pi$ of $\sA$ over $\sH$ is called \textit{faithful} if $\ker(\pi) = 0$. 
\end{defn}
Faithful reps are particularly nice because they preserve all the structure. The structure theorem tells us that every $\Cstar$-algebra has a faithful rep. The power of using a representation comes into play when we are also considering a state on the $\Cstar$-algebra, and we want to represent the state in the Hilbert space as well. 
\begin{defn}
Given a $\Cstar$-algebra $\sA$ and a representation $(\sH, \pi)$, a state $\omega$ is a \textit{vector state} if 
\[  \omega(A) = (\Omega, \pi(A) \Omega) \]
for some $\Omega \in \sH$. $\Omega$ is a vector representative of $\omega$. 
\end{defn}

The GNS representation is a powerful result that allows any $\Cstar$-algebra with a state to be represented in the bounded operators of some Hilbert space with the state being described by a cyclic vector in this Hilbert space. It says that \textit{all} states can be realised as vector states by constructing the appropriate representation. 
\begin{leftbar}
\begin{prop}[GNS Representation] 
\label{prop GNS}
Given a $\Cstar$-algebra $\sA$ and a state $\omega$, there exists a representation of $\sA$ such that $\omega$ is a vector state and its vector representative is a cyclic vector. This representation is unique up to unitary equivalence, and is called the GNS representation. 
\end{prop} 
\end{leftbar}
\begin{proof} \textit{(Existence)} 

Begin with a $\Cstar$-algebra $\sA$ and a state $\omega$. Let $I_{\omega} = \{ A \in \sA : \omega(A^*A) = 0 \}$. This is a left ideal of $\sA$. (This is not easy to show, so omit for now.) Now let $ \sH= \sA / I_{\omega} $. Define an inner product on $\sH_{\omega}$ by $([A], [B])_{\omega} = \omega(A^*B)$. It is not hard to check that this is well-defined using the fact that $I_{\omega}$ is an ideal. Finally let $\sH_{\omega} = \overline{ \left( \sA / I_{\omega} \right)} $ where the closure is with respect to this inner product. $(\sH_{\omega}, ()_{\omega})$ is a Hilbert space. 

We construct a map 
\begin{align*}
 \pi_{\omega} : \sA & \rightarrow  \sB(\sH_{\omega}) 
\end{align*}
where $\pi_{\omega}(A)$ is defined as follows. First, for all $B \in \sA$ let 
\[ \pi_{\omega}(A)_0 : [B]  \mapsto [A B]\]
This is densely defined on $\sH_{\omega}$. 
We now use Useful Bound 2 to see that $\forall B \in \sA$ 
\[ || \pi_{\omega}(A)_0 [B] ||_{\sH_{\omega}}^2 = ([AB], [AB]) = \omega(B^*A^*AB) \leq ||A||^2 \omega(B^*B) = ||A||^2 || [B]||_{\sH_{\omega}}^2 \]
Hence the closure of $\pi_{\omega}(A)_0$ is in $\sB(\sH_{\omega})$ and we let $\pi_{\omega}(A) = \overline{\pi_{\omega}(A)_0}$. With this definition the map $\pi_{\omega}$ is a morphism of $\Cstar$-algebras. 

Finally let $\Omega_\omega = \mathbb{I}$. For any element $[A]$ of $\sH_{\omega}$, $\pi(A) \Omega_{\omega}$ = $[A \mathbb{I}] = [A]$, so $\Omega_{\omega}$ is cyclic.  Then 

\[ (\Omega_{\omega}, \pi_{\omega}(A) \Omega_{\omega})_{\omega} = ([\mathbb{I}], \pi_{\omega}(A) [\mathbb{I}])_{\omega} = ( [\mathbb{I}], [A])_{\omega} = \omega(\mathbb{I}^* A) = \omega(A)\] 

This construction satisfies the requirements of Proposition \ref{prop GNS}. It is called the \textit{canonical GNS representation}. 
\newline
\newline
 \textit{(Uniqueness)}
 \newline
Let $(\sH_{\omega}, \pi_{\omega}, \Omega_{\omega})$ be the cannonical GNS representation, and $(\sH_{\omega}', \pi_{\omega}', \Omega_{\omega}')$ be another representation satisfying the conditions of proposition \ref{prop GNS}. 

Construct an operator $\mathcal{U}$ as follows. 
\begin{align*}
\mathcal{\tilde{U}} : \textrm{Dom}(\mathcal{\tilde{U}}) \subset \sH_{\omega} & \rightarrow \sH_{\omega}' \\
\Omega_{\omega} & \mapsto \Omega_{\omega}' \\
\pi_{\omega}(A) \Omega_{\omega} & \mapsto \pi_{\omega}'(A) \Omega_{\omega}' \qquad \forall A \in \sA
\end{align*} 
$\mathcal{\tilde{U}}$ is bounded on its domain. 
\begin{align*} 
|| \mathcal{\tilde{U}} (\pi_{\omega}(A) \Omega_{\omega} )||_{\sH_{\omega}'}& = || \pi_{\omega}' (A) \Omega_{\omega} ' ||_{\sH_{\omega}' } \\
& = (\pi_{\omega}' (A) \Omega_{\omega}', \pi_{\omega}' (A) \Omega_{\omega}')_{\sH_{\omega}' } ^2 \\
& = (\Omega_{\omega}', \pi_{\omega}' (A^*A) \Omega_{\omega}' )^2  \\
& = \omega(A^*A)^2 \\
& = || \pi_{\omega} (A) \Omega_{\omega} ||_{\sH_{\omega} }
\end{align*} 

Since $\textrm{Dom} (\mathcal{\tilde{U}}) = \{ \pi_{\omega}(A) \Omega_{\omega} : A \in \sA \}$ is dense in $\sH_{\omega}$ we can extend $\mathcal{\tilde{U}}$ to a bounded linear operator $\mathcal{U}$ on all of $\sH_{\omega}$ by the B.L.T. theorem. Note that since the range of $\mathcal{\tilde{U}}$ is dense in $\sH_{\omega}'$ the range of $\mathcal{U}$ is also dense. Finally we verify that $\mathcal{\tilde{U}}$ and hence $ \mathcal{U}$ preserve the inner product. Let $\xi_A = \pi_{\omega}(A) \Omega_{\omega}$ and $\xi_B = \pi_{\omega}(B) \Omega_{\omega}$ .
\begin{align*}
(\mathcal{U} \xi_A, \mathcal{U} \xi_{B} ) & = (\pi_{\omega}'(A) \Omega_{\omega}', \pi_{\omega}'(B) \Omega_{\omega}' ) \\
& = \omega(A^*B) \\
& = (\xi_A, \xi_B)
\end{align*} 
Let $\xi, \xi' \in \sH_{\omega}$. Then there are sequences $A_n$ and $B_n$ in $\sA$ such that $\pi_{\omega}(A_n) \Omega_{\omega} \rightarrow \xi$ and $\pi_{\omega}(B_n) \Omega_{\omega} \rightarrow \xi'$. 
\begin{align*}
(\mathcal{U} \xi, \mathcal{U} \xi') & = \lim_{n \rightarrow \infty} (\mathcal{U} \pi_{\omega}(A_n) \Omega_{\omega}, \mathcal{U} \pi_{\omega} (B_n) \Omega_{\omega}) \\
& = \lim_{n \rightarrow \infty} ( \pi_{\omega}(A_n) \Omega_{\omega}, \pi_{\omega} (B_n) \Omega_{\omega}) \\
& = (\xi, \xi')
\end{align*} 
The fact that $\mathcal{U}$ is bounded and linear and hence continuous is used in the first and last steps. Hence, since $\mathcal{U}$ preserves the inner product and has dense range, it is unitary. 
\end{proof}

One useful application of the uniqueness part of the result is to take an invariant state $\omega$ for some dynamics $\tau$ (such as a KMS state) and apply the theorem to the cyclic representation $(\sH_{\omega}, \pi_{\omega} \circ \tau^t, \Omega_{\omega})$. Since
\[ (\Omega_{\omega}, \pi_{\omega} \circ \tau^t (A) \Omega_{\omega} ) = \omega(\tau^t(A))=\omega(A)\]
we have that this rep is unitarily equivalent to $(\sH_{\omega}, \pi_{\omega}, \Omega_{\omega})$.  Hence there is some unitary $\sU(t)$ such that 
\begin{align*}
\pi_{\omega}(\tau^t(A)) &= \sU(t) \pi_{\omega}(A) \sU(t)^* \\
\sU(t) \Omega_{\omega} &= \Omega_{\omega}
\end{align*}
i.e. the automorphism $\tau^t$ in unitarily implemented by $\sU(t)$ in the GNS rep of the invariant state. The unitary $\sU(t)$ is not unique, but we will see later that with an additional restriction there is a unique way to choose it, and then $t \mapsto \sU(t)$ will be a one-parameter strongly continuous unitary group. We can apply Stone's theorem to get a self-adjoint (over $\sH_{\omega}$) generator $L$ such that $U(t) = e^{itL}$. 
\\
\\
Additionally, for a given state the GNS representation may not be faithful, in which case some of the structure is lost by working in the representation. However, there is a notion of faithfulness for states. 
\begin{defn}
 A state $\omega$ over $\sA$ is faithful if $\omega(A^*A) = 0 $ implies that $A=0$. 
 \end{defn}
 If a state $\omega$ is faithful, then in the construction of the GNS rep $\sI_{\omega}$ is trivial and hence $\sH_{\omega} = \sA$ and $\pi_{\omega}(A) X = AX$ for all $X\in \sA$. It follows that $\pi_{\omega}$ is injective. The GNS representation of a faithful state is faithful. We also have the following useful result \cite{BR2}.
 \begin{prop}
A KMS state over a $\Cstar$-algebra is automatically faithful. 
 \end{prop}
We may safely work in the GNS rep of an equilibrium state without losing any information. 
\\
\\
Finally, we may sometimes wish to work in the GNS rep of a $\Cstar$-algebra of the form $\sB(\sH)$ where $\sH$ is a finite dimensional Hilbert space. In this case it is often more convenient to consider the standard GNS rep, which is unitarily equivalent but not equal to the canonical GNS rep. It is constructed from $(\sB(\sH), \tau^t = e^{it[H, \cdot]}, \omega = \tr(\rho \cdot))$ by taking $\sH_{\omega} = \sB(\sH)$ with inner product 
\[ (A,B)_{\sH_{\omega}} = \tr(A^*B),\] 
$\pi_{\omega}(A) = L(A)$ where $L(A)$ is the standard left representation 
\[L(A) B = AB \]
for $B \in \sH_{\omega}$, and $\Omega_{\omega} = \rho^{1/2}$. 
\subsection{Normal States}
Having moved to a GNS representation for a given state, one can ask whether other states can also be represented as vectors in this representation. In general, two arbitrary states cannot be represented in the same GNS Hilbert space, and those that can exhibit certain behaviour. 
\begin{defn}
Let $\sA$ be a $\Cstar$ algebra and $\omega, \eta$ two states. The state $\eta$ is said to be $\omega$-normal if there is a positive trace-one operator $\rho_{\eta}$ in $\sB(\sH_{\omega})$ such that
\[ \eta(A) = \tr(\rho_{\eta}\pi_{\omega}(A)).\]
In other words, $\eta$ can be written as a trace state in the $\omega$-GNS rep. 
\end{defn}
\begin{prop}
If $\eta$ is $\omega$-normal, then there exists a vector $\Omega_{\eta} \in \sH_{\omega}$ such that
\[ \eta(A) = (\Omega_{\eta}, \pi_{\omega}(A) \Omega_{\eta}).\]
\end{prop}
This should be thought of as the non-commutative analog of a measure $\eta$ being absolutely continuous with respect to another measure $\omega$. In that situation integrals against $\eta$ can be computed as integrals against $\omega$ (by weighting by the Radon-Nikodym derivative). Similarly here, the value of $\eta$ on an operator can be computed in the $\omega$-GNS rep, by taking an inner product against an appropriate vector. We will see later that there is an object built out of $\Omega_{\eta}$ which plays the role of the Radon-Nikodym derivative. 
\\
\\
We denote by $\mathcal{N}_{\omega}$ the set of all states which are $\omega$-normal. This set can be thought of physically as the states that are ``not too different" to the state $\omega$. We can now compare states by comparing these sets of relatively normal states. 
\begin{defn}
Let $\omega$ and $\eta$ be two states. If $\mathcal{N}_{\omega} = \mathcal{N}_{\eta}$ then we say they are \textit{quasi-equivalent}. If $\mathcal{N}_{\omega} \cap \mathcal{N}_{\eta} = \phi$ then we say they are \textit{disjoint}. 
\end{defn}
This allows us to refine our description of the set of KMS states. 
\begin{prop}
If $\omega_1$ and $\omega_2$ are two extremal states in the set of $\tau, \beta$-KMS states for some fixed $\tau$ and $\beta$ then they are either equal or disjoint. 
\end{prop}
Physically, this corresponds to pure phases being ``far apart" in the sense that one cannot represent two different pure phases both as vectors states in \textit{any} GNS representation. The set of normal states is also an important concept used in Araki's perturbation theory of KMS states, which is the subject of the next section. 

\subsection{Araki Perturbation Theory}
\label{Araki}
In Araki's perturbation theory of KMS states we start with a $\Cstar$-dynamical system $(\sA, \tau^t)$ and a fixed temperature $\beta^{-1}$. We assume that $\omega$ is some fixed $\tau, \beta$-KMS state. We then perturb the dynamics by a bounded self-adjoint operator and study the new equilibrium states. The following proposition describes the type of perturbation we are using. 
\begin{prop}
Let  $(\sA, \tau^t)$  be a $\Cstar$-dynamical system with generator $\delta$. Let $V \in \sA$ be self adjoint. Define $\delta_V$ by
\begin{align*}
\mathcal{D}(\delta_V)& = \sA \\
\delta_V (A) &= i[V,A] .
\end{align*}
Then $\delta + \delta_V$ generates a one-parameter group of $^*$-automorphisms, $\tau_V^t$. This dynamics has the perturbative expansion
\[ \tau_V^t(A) = \tau^t(A) + \sum_{n \geq 1} i^n \int_0^t \int_0^{t_1} ... \int_0^{t_{n-1}} [\tau^{t_n}(V) , [ ...[\tau^{t_1}(V), \tau^t(A)]]] dt_n ... dt_2 dt_1 .\]
\end{prop}
If we move to the GNS rep, where our unperturbed dynamics is unitarily implemented, we also have a unitary implementation for the perturbed dynamics. 
\begin{prop}
\[\pi_{\omega} (\tau_V^t(A) ) = \Gamma_V^t   \pi_{\omega}(\tau^t(A))  {\Gamma_V^{t}}^*\]
where
\[ \Gamma_V^t = \Id + \sum_{n \geq 1} i^n \int_0^t \int_0^{t_1} ... \int_0^{t_{n-1}} \pi_{\omega} \left( \tau^{t_1} (V) \right) ... \pi_{\omega} \left( \tau^{t_n} (V) \right) dt_1...dt_n .\]
\end{prop}

Now, under this perturbation we want to see what we can say about the new KMS states \cite{BR2}. 
\begin{leftbar}
\begin{thm}[Araki Perturbation]
\label{araki}
Let $(\sA, \tau)$ be a $\Cstar$-dynamical system, $\omega$ a $(\tau, \beta)$-KMS state and $V$ a self adjoint element of $\sA$. Furthermore, let $(\sH_{\omega}, \pi_{\omega}, \Omega)$ be the GNS representation of $\omega$ and $U(t) = e^{itL}$ be a unitary implementation of $\tau^t$ in this representation. 
\begin{enumerate}[(1)]
\item There is a unique $\tau_V, \beta$-KMS state which is $\omega$-normal. Call it $\omega_V$. 
\item The states $\omega$ and $\omega_V$ are quasi-equivalent. 
\item If $\{ V_{\alpha} \}$ is a net of perturbations converging to $0$ in norm then
\[ \lim_{\alpha} \omega_{V_{\alpha}} = \omega.\]
\end{enumerate}
Furthermore, we have the following explicit formulas.
\begin{itemize}
\item In the $\omega$-GNS representation, the perturbed dynamics $\tau_V$ is given by 
\[ \pi_{\omega} (\tau_V^t(A)) = e^{it(L + \pi_{\omega}(V))} \pi_{\omega}(A) e^{-it(L + \pi_{\omega}(V))}.\]
\item $\Omega$ is in the domain of $e^{-\beta(L+ \pi_{\omega}(V))/2}$ and the vector representative of $\omega_V$ in $\sH_{\omega}$ is 
\[\Omega_V = \frac{e^{-\beta(L+ \pi_{\omega}(V))/2}\Omega}{\| e^{-\beta(L+ \pi_{\omega}(V))/2}\Omega \|}.\]
The function 
\[z \mapsto e^{-z (L +  \pi(V))} \Omega_0\]
is analytic in the strip $0 < \textit{Re} z < \beta / 2$ and bounded and continuous on its closure. 
\end{itemize}
\end{thm}
\end{leftbar}
\begin{proof}[Proof. (of formulas)]
Let $\sU_V^t = e^{it(L + \pi_{\omega}(V))}$ and $X^t = \sU_V^t \sU^{-t}$. We compute
\begin{align*}
\frac{dX^t}{dt} &= e^{it(L + \pi_{\omega}(V))} (i(L + \pi_{\omega}(V))) e^{-itL} +e^{it(L + \pi_{\omega}(V))} (-iL)e^{-itL}  \\
& = i \sU_V^t \pi_{\omega}(V) \sU^{-t} \\
& = i X^t \pi_{\omega}(\tau^t(V)) \\
\textrm{and} & \\
X^0 & = \Id .
\end{align*}
Hence $X^t$ is the unique solution of 
\[ X^t = \Id + i \int_0^t X^s \pi_{\omega}(\tau^s(V)) ds.\]
Iterating this we obtain that $X^t = \Gamma_V^t$ and hence $\Gamma_V^t \sU^{t} = \sU_V^t$. To conclude we have that
\[ \pi_{\omega}(\tau_V^t(A))= \Gamma_V^t \pi_{\omega}(\tau^t(A)) {\Gamma_V^t }^* = \Gamma_V^t \sU^t \pi_{\omega}((A)) {\sU^t}^* {\Gamma_V^t }^*  = \sU_V^t \pi_{\omega}(A) {\sU_V^t}^*.\]
Next we will assume that $V$ is an analytic element for $\tau$ in order to define 
\[ \Gamma_V^{is} = \Id + \sum_{n \geq 1} (-1)^n \int_0^s \int_0^{s_1} ... \int_0^{s_{n-1}} \pi_{\omega} \left( \tau^{is_n}(V) \right)...\pi_{\omega} \left( \tau^{is_1}(V) \right) ds_1ds_2...ds_n.\]
This converges uniformly and we have that $\Gamma_V^{i \beta/2} \Omega = e^{-\beta(L + \pi_{\omega})/2} \Omega$, noting that $L \Omega = 0$. 
\end{proof}
Using this, we can also look at the structure of the set of KMS states under the perturbation, and see that it is preserved. 
\begin{cor}
The map from the set of $(\tau, \beta)$-KMS states to $(\tau_V, \beta)$-KMS states given by 
\[ \omega \mapsto \omega_V \]
is an isomorphism and maps extremal points to extremal points. 
\end{cor}
This means that the physical phase structure remains the same under this type of bounded perturbation. 
\section{Modular Theory}

\subsection{von Neumann Algebras}
Modular theory provides a very useful set of tools for doing computations in the GNS representation. To define the necessary structure to develop the modular theory, we need a special class of $\Cstar$-algebras called von Neumann algebras. These are concrete algebras of operators over Hilbert spaces. 
\begin{defn}
Let $\sH$ be a Hilbert space and let $\sM$ be a subset of $\sB(\sH)$. The \textit{commutant} of $\sM$, denoted $\sM'$, is the set of all bounded linear operators on $\sH$ which commute with every element of $\sM$. $\sM$ is a \textit{Von Neumann Algebra}  if 
\[ \sM '' = \sM .  \] 
\end{defn}
This algebraic definition is equivalent to an analytic property. 
\begin{prop}[Bicommutant Theorem ] 
Let $\sM$ be a non-degenerate $^*$-algebra of operators on a Hilbert space $\sH$. $\sM$ is a von Neumann algebra if and only if $\sM$ is closed in the weak operator topology, i.e. the smallest topology which makes all the linear functionals of the form 
\[ A \mapsto \langle x, Ay \rangle\]
continuous for all $x,y \in \sH$. 
\end{prop}
Note that the weak operator topology is weaker than the norm topology, which makes every von Neumann algebra automatically a $\Cstar$ algebra. 
\\
\\
In our setting, we have algebras of the form $\pi_{\omega}(\sA) \subset \sB(\sH_{\omega})$, where $\sA$ is a $\Cstar$ algebra, arising from GNS representations. These may not be von Neumann algebras, but there is a canonical way to enlarge them to get a von Neumann algebra. 
\begin{prop}
Let $\sM$ be a sub-algebra. $\sM^{(n)} = \sM^{(n+2)}$ for $n \geq 1$
\end{prop}
\begin{proof}
$ \sM ' $ is the set of $B \in \sB(\sH)$ such that $[A, B] = 0 \forall A \in \sM$. So if $A \in \sM$ then $[A, B] = 0 \forall B \in \sM' $ which implies that  $ A \in \sM ''$. Therefore we always have the inclusion $\sM \subset \sM''$. It follows from this that $(\sM'')' = \sM ''' \subset \sM ' $. As well, $\sM' \subset (\sM')'' = \sM'''$. By induction, $\sM^{(n)} = \sM^{(n+2)}$ for $n \geq 1$. 
\end{proof}

Given \textit{any} sub algebra $\sM \subset \sB(\sH)$ we can turn it into a von Neumann algebra by taking the bicommutant. This has the effect of ``adding in some elements'' and guarantees the von Neumann algebra property since $\sM ^{(4)} = \sM '' $. For this reason $\sM ''$ is called the \textit{enveloping von Neumann algebra}. The enveloping von Neumann algebra is also the weak closure, making it the minimal algebra with this property. 
\\
\\
Following \cite{BR1} Chapter 2, our next step is to define certain operators on $\sH$. We will do this by defining them on a domain of the form $\sM \Omega$ or $\sM' \Omega$ for some vector $\Omega$ in the Hilbert space. $\Omega$ must have certain properties. 
\begin{defn}
Let $\sM$ be a von Neumann algebra on $\sH$ and let $\Omega \in \sH$.
\begin{itemize}
\item $\Omega$ is \textit{separating} if for all $A \in \sM$, $A\Omega = 0 \implies A=0$. 
\item $\Omega$ is \textit{cyclic} if $\sM \Omega = \{ A \Omega : A \in \sM \}$ is dense in $\sH$. 
\end{itemize}
\end{defn}

\begin{prop}
$\Omega$ is cyclic for $\sM$ if and only if $\Omega$ is separating for $\sM'$. 
\end{prop}
In particular, if a vector is cyclic and separating for $\sM$ it is also cyclic and separating for $\sM'$. Now we go back to the GNS setting $(\sH_{\omega}, \pi_{\omega}, \Omega_{\omega})$, and consider the von Neumann algebra $\pi_{\omega}(\sA)''$. If $\Omega_{\omega}$ is cyclic and separating for $\pi_{\omega}(\sA)''$ then we say the state $\omega$ on $\sA$ is \textit{modular}, and we will be able to use it to build the modular structure described in the next section. Furthermore, when we move to the GNS rep of a state, we will be able to extend the state to be a linear functional on the enveloping von Neumann algebra by using the vector representative. This extended state is always a trace state, and we can use this form to compute the action of various operators explicitly. 
\begin{defn}
A state $\omega$ on a Von Neumann algebra $\sM$ is called \textit{normal} if there is a density matrix $\rho_{\omega}$ such that $\omega(A) = \textrm{Tr}(\rho_{\omega}A)$ for all $A \in \sM$. 
\end{defn}
\begin{prop}
If $\omega$ is a state on a $\Cstar$-algebra $\sO$ then it has a normal extension to the enveloping Von Neumann algebra given by 
\[ \widehat{\omega}(A) = (\Omega_{\omega}, A \Omega_{\omega}) \qquad \forall A \in \sM.\]
\end{prop}
So, given a state $\omega$ on $\sA$, we can extend it to a state on the enveloping von Neumann algebra $\widehat{\omega}$ and this state is a trace state with density matrix $\rho_{\omega} =| \Omega_{\omega} \rangle \langle \Omega_{\omega} |$. 
\subsection{Modular Operators} 
\label{modular}
We start with a pair $(\sM, \Omega)$ where $\sM$ is a von Neumann algebra and $\Omega$ is a cyclic and separating vector, and define an operator which implements the operator adjoint on the space $\sM \Omega$. Note that these definitions rely crucially on having the underlying Hilbert space structure, which is not manifestly present for $\Cstar$-algebras. We must go to the GNS representation. 
\newline
\newline
Let $\sM$ be a Von Neumann algebra with separating and cyclic vector $\Omega$. By the previous proposition,  $\Omega$ is also separating and cyclic for $\sM'$. We define two (possibly unbounded) operators on $\sH$
\begin{align*}
S_0 A \Omega &= A^* \Omega \qquad A \in \sM \\
F_0 A \Omega & = A^* \Omega \qquad A \in \sM' 
\end{align*}
These are both densely defined, since $\Omega$ is cyclic. Also note that these are anti-linear operators. 

\begin{prop}
$S_0$ and $F_0$ are closable and $\overline{F_0} = S_0^*$, $\overline{S_0} = F_0^*$. 
\end{prop}
\begin{proof}
Suppose $A \in \sM$ and $A' \in \sM'$. 
\[ (A' \Omega, S_0 A \Omega) = (A' \Omega, A^* \Omega) = (A \Omega, A'^* \Omega) = (A \Omega, F_0 A' \Omega) \]
Hence $\mathcal{D}(F_0) \subset \mathcal{D}(S_0^*)$, and since $\mathcal{D}(F_0) = \sM' \Omega$, which is dense, $\mathcal{D}(S_0^*)$ is also dense. It follows that $S_0$ is closable. (Theorem 23, Spectral Theory notes) It is more involved to show that in fact $\overline{S_0} = F_0^*$. (Details in BR 2.5.9)
\end{proof}
We define $S = \overline{S_0}$ and $F = \overline{F_0}$. Then we take $\Delta^{1/2}$ to be the unique positive self-adjoint operator and $J$ to be the unique anti-unitary operator in the polar decomposition
\[ S = J \Delta^{1/2}.\]
$\Delta$ is called the modular operator and $J$ is the modular conjugation associated to $(\sM, \Omega)$. $J$ and $\Delta$ may be unbounded operators. However, the functional calculus is still valid by the unbounded version of the spectral theorem (see for example \cite{Ja}). 
There are many identities involving these operators which are useful for computational purposes. 
\begin{leftbar}
\begin{prop}[Useful Modular Identities] 
\[ \Delta = FS \qquad \Delta^{-1} = SF \qquad F = J \Delta^{-1/2} \qquad J = J^* \qquad J^2 = \Id \qquad \Delta^{-1/2} = J \Delta^{1/2} J \]
\end{prop}
\end{leftbar}
\begin{proof}
We know that $J^* = J^{-1}$ since $J$ is anti-unitary. 
\[FS = S^*S = \Delta^{1/2}J^*J\Delta^{1/2} = \Delta\]
Also, from the definition of $S$, $S^{-1} = S$. 
\[ S = S^{-1} \implies J\Delta^{1/2} = \Delta^{-1/2}J^* \implies J \Delta^{1/2} J = \Delta^{-1/2} \]
Then,
\[ SF = SS^* = J \Delta^{1/2} \Delta^{1/2} J^* = (J \Delta^{1/2} J^*)^2 = \Delta^{-1}.\]
We use uniqueness of the polar decomposition and 
\[ J\Delta^{1/2} = \Delta^{-1/2}J^* \implies J^2 \Delta^{1/2} = J \Delta^{-1/2} J \]
to conclude that $J^2=1$, since both $\Delta^{1/2}$ and $J \Delta^{-1/2} J^*$ are positive operators. Finally,
\[J^2 = \Id = JJ^* \implies J = J^*.\]
\end{proof}
We can also compute how these operators act explicitly when we are working in the enveloping von Neumann algebra of a GNS representation. Suppose we have a modular state $\omega$ which is given by a density matrix $\rho_{\omega} = | \Omega_{\omega} \rangle \langle \Omega_{\omega} |$ in its GNS rep. Let $\sM = \pi_{\omega}(\sA)''$ where $\pi_{\omega}$ is the canonical GNS representation of a modular state from the proof of Theorem \ref{prop GNS}. Hence $\Omega_{\omega}$ is cyclic and separating and we can define the associated modular operator and conjugation. The star operator acts as 
\begin{align*}
S : \sH_{\omega} & \rightarrow \sH_{\omega} \\
[A \Id ] & \mapsto [A^* \Id ] .
\end{align*}
We can then verify that the polar decomposition is:
\begin{align*}
J :  [A] & \mapsto [\rho_{\omega}^{1/2} A^* \rho_{\omega}^{-1/2}] &A &\in \sA \\
 \Delta^{1/2} :  [A] &\mapsto [\rho_{\omega}^{1/2} A \rho_{\omega}^{-1/2}] &A &\in \sA \\
  \Delta :  [A] &\mapsto [\rho_{\omega} A \rho_{\omega}^{-1}] &A &\in \sA 
\end{align*}
by checking that $J$ is anti-unitary and $\Delta^{1/2}$ is self-adjoint. 
\begin{equation*}
\begin{aligned}
 & (J \pi_{\omega}(A) \Omega,  J \pi_{\omega}(B) \Omega)  \\
 &  = ([\rho_{\omega}^{1/2}A^*\rho_{\omega}^{-1/2}], [\rho_{\omega}^{1/2}B^*\rho_{\omega}^{-1/2}]) \\
 &= \tr(\rho_{\omega} \rho_{\omega}^{-1/2} A \rho_{\omega}^{1/2} \rho_{\omega}^{1/2} B^* \rho_{\omega}^{-1/2}) \\
 & = \tr(\rho_{\omega} B^*A) \\
 & = (\pi_{\omega}(B) \Omega, \pi_{\omega}(A) \Omega) \\
 \end{aligned}
 \;\;\;\;\;\;\;\;\;\;
 \begin{aligned}
 &( \Delta^{1/2} \pi_{\omega}(A) \Omega, \pi_{\omega}(B) \Omega) \\
 & = ([\rho_{\omega}^{1/2}A \rho_{\omega}^{-1/2}], [B]) \\
 & = \tr(\rho_{\omega}\rho_{\omega}^{-1/2}A^*\rho_{\omega}^{1/2}B) \\
 & = \tr(\rho_{\omega}A^* \rho_{\omega}^{1/2}B\rho_{\omega}^{-1/2})\\
 & = (\pi_{\omega}(A) \Omega, \Delta^{1/2} \pi_{\omega}(B)\Omega)
 \end{aligned}
 \end{equation*}
 \subsection{Natural Cone}
The natural cone is very special subset of the underlying Hilbert Space $\sH$, which is associated to a pair $(\sM, \Omega)$.  It has many nice geometric properties, such as being pointed, convex, self-dual and spanning the Hilbert space. (Details in BR1 2.5.26 and 2.5.28.) In addition, it gives us a way to represent states on $\sM$ uniquely by a vector. We shall also see that by changing our point of view between different vectors in the cone, we do not change the cone itself and some of the modular structure is invariant as well. This property is called universality of the cone. 
\begin{defn}
The \textit{natural positive cone} $\sP$ associated with $(\sM, \Omega)$ is the closure of the set 
\[ \{ AJAJ \Omega : A \in \sM \} \]
\end{defn}

The next proposition shows how restricting our attention to the natural cone gives us a way to choose unique vector representatives of states.
\begin{defn}
For each $\xi \in \sP$ define the normal positive form $\omega_{\xi} \in \sM_{*, +}$ by 
\[ \omega_{\xi}(A) = (\xi, A \xi)_{\sH}\]
\end{defn}

\begin{prop} 
For any $\omega \in \sM_{*,+}$ there is a unique $\xi \in \sP$ such that $\omega = \omega_{\xi}$.
\end{prop}
Finally, we have a result which allows us to consider another vector $\xi$ in the cone of $(\sM, \Omega)$, and see that the cone remains invariant. 
\begin{prop}[Universality of the Cone ] 
\begin{enumerate}
\item[]
\item if $\xi \in \sP$, then $\xi$ is cyclic if and only if $\xi$ is separating
\item if $\xi \in \sP$ is cyclic (and hence separating) for $\sM$ then $J_{\xi}$ and $\sP_{\xi}$ associated to $(\sM, \xi)$ satisfy 
\[ J_{\xi} = J \qquad \qquad \sP_{\xi} = \sP \]
\end{enumerate} 
\end{prop}
We can compute the cone in the canonical GNS rep. 
\begin{align*}
\sP_{\omega} & = \overline{\{ AJAJ \Omega_{\omega} : A \in\sM  \}} \\
& =  \overline{\{ AJA [\Id] : A \in\sM  \}}  \\
& =  \overline{\{ A \rho_{\omega}^{1/2} A^* \rho_{\omega}^{-1/2} : A \in\sM  \}} \\
& = \overline{\{ BB^*\rho_{\omega}^{-1/2}  : B \in\sM^+  \}} = \overline{ \sM^+ \rho_{\omega}^{-1/2} }
\end{align*} 

\subsection{Relative Modular Operator}
We may now generalise the modular structure that was developed in section \ref{modular}. Given a Von Neumann algebra $\sM$ and two cyclic and separating vectors $\Omega$ and $\Omega'$ we may define the operator on $\sH$
\[ S^0_{\Omega' | \Omega} : A \Omega \mapsto A^* \Omega' \qquad \forall A \in \sM \]
If $\Omega' = \Omega$ we recover the operator $S_0$.  Let $S_{\Omega ' | \Omega} = \overline{S^0_{\Omega ' | \Omega}}$. Again we take the polar decomposition into a unique positive self adjoint operator $\Delta_{\Omega' | \Omega} ^{1/2}$ and anti-unitary operator $ J_{\Omega' | \Omega}$. 
\[ S_{\Omega' | \Omega}  = J_{\Omega'| \Omega} \Delta_{\Omega' | \Omega} ^{1/2}\]
$\Delta_{\Omega' | \Omega}$ is called the relative modular operator associated to $(\sM,\Omega, \Omega')$. 
\\
\\
If $\omega$ is a modular state and $\eta$ is $\omega$-normal on a $\Cstar$-algebra $\sA$, then we can move to the $\omega$-GNS rep. $\omega$ is represented by a cyclic and separating vector $\Omega_{\omega}$. Furthermore, there is a unique vector rep $\Omega_{\eta}$ for $\eta$ in the natural cone $\sP_{\omega}$,  and by the universality of the cone we have that if $\Omega_{\eta}$ is cyclic, it is automatically also separating. This allows us to define the relative modular operator $\Delta_{\Omega_{\eta}|\Omega_{\omega}}$. We can explicitly compute the polar decomposition in the canonical GNS rep. As before, $\rho_{\omega} = | \Omega_{\omega} \rangle\langle \Omega_{\omega}|$. The method is exactly the same as in the simpler case $\eta = \omega$. 
\begin{align*}
J_{\Omega_{\eta} | \Omega_{\omega} }:  [A] & \mapsto [\rho_{\omega}^{1/2} A^* \rho_{\omega}^{-1/2}] &A &\in \sA \\
 \Delta^{1/2}_{\Omega_{\eta} | \Omega_{\omega} } :  [A] &\mapsto [\rho_{\eta}^{1/2} A \rho_{\omega}^{-1/2}] &A &\in \sA \\
  \Delta_{\Omega_{\eta} | \Omega_{\omega} } :  [A] &\mapsto [\rho_{\eta} A \rho_{\omega}^{-1}] &A &\in \sA 
\end{align*}
Notice that $J_{\Omega_{\eta} | \Omega_{\omega}}  = J_{\Omega_{\omega}}$. A priori, we do not know that $J_{\Omega'  | \Omega} = J_{\Omega} = J_{\Omega'}$ more generally. However, this turns out to be the case as long as $ \Omega' \in \sP_{\Omega} $. 

\begin{prop}
Let $\sM$ be a VNA and $\Omega$ be a cyclic and separating vector. Let $\Omega'$ be a cyclic and separating vector in $\sP_{\Omega}$. Then $J_{\Omega'  | \Omega} = J_{\Omega}$.
\end{prop}
\begin{proof}
Consider the state $\omega$ on $\sM$ defined by 
\[ \omega(A) = (\Omega, A \Omega)_{\sH} \qquad \forall A \in \sM \]
Note that $\omega$ is a normal state with $\rho_{\omega} = | \Omega \rangle \langle \Omega |$. The triple $(\sH, id : \sM \rightarrow \sB(\sH), \Omega)$ is a GNS representation for $\omega$. Hence by the uniqueness of the GNS representation it is unitarily equivalent to the canonical GNS representation $(\sH_{\omega} = \overline{\sM / I_{\omega}} , \pi_{\omega}, \Omega_{\omega})$. Following the uniqueness proof, one can construct the unitary operator 
\begin{align*}
 U : \sH & \rightarrow \sH_{\omega} \\
 A \Omega & \mapsto [ A ] \qquad \forall A \in \sM 
 \end{align*}
Now we consider the two star operators 
\begin{align*}
S_{\Omega} : \sH 	& \rightarrow \sH   		& S_{\Omega_{\omega}} : 	 \sH_{\omega} 						& \rightarrow \sH_{\omega} \\
A \Omega 		& \mapsto A^* \Omega 	& 						\pi_{\omega}(A) \Omega_{\omega} = [A] 	& \mapsto  \pi_{\omega}(A)^* \Omega_{\omega} = [A^*] . \end{align*}
It is not hard to check that $ S_{\Omega} = U^* S_{\Omega_{\omega}} U$ and hence by the uniqueness of the polar decomposition $ J_{\Omega} = U^* J_{\Omega_{\omega}} U$. It is also easy to check that $ U \sP_{\Omega} = \sP_{\Omega_{\omega}}$. If $\Omega' \in \sP_{\Omega}$ then for some $A \in \sM$, $\Omega' = A J_{\Omega} A J_{\Omega} \Omega$. 
\begin{align*}
U \Omega ' & = U A J_{\Omega} A J_{\Omega} \Omega \\
& = U A U^*U J_{\Omega} U^* U A U^* U J_{\Omega} U^* U \Omega \\
& = \pi_{\omega}(A) J_{\Omega_{\omega}} \pi_{\omega}(A) \Omega_{\omega} \in \sP_{\Omega_{\omega}} 
\end{align*}
Now we take $\sN$ in the natural cone $\sP_{\Omega}$. Again we consider the state induced by $\sN$ defined by $\eta(A) = (\sN, A \sN)_{\sH}$. We know that $U\sN$ is in the cone $\sP_{\Omega_{\omega}}$. A calculation shows that it is the unique vector rep of $\eta$ in this cone. Let $B_n$ be a sequence in $\sM$ such that $ \lim_{n \rightarrow \infty} B_n \Omega = \sN$. 
\begin{align*}
(U \sN, \pi_{\omega}(A) U \sN)_{\sH_{\omega}} & = \lim_{n \rightarrow \infty} (UB_n \Omega, \pi_{\omega}(A) U B_n \Omega)_{\sH_{\omega}}  = \lim_{n \rightarrow \infty}  ([B], [AB])_{\sH_{\omega}} \\
& =  \lim_{n \rightarrow \infty}  \omega(B^*AB) =  \lim_{n \rightarrow \infty} (\Omega, B^*AB \Omega)_{\sH} = (\sN, A\sN)_{\sH} \\
&= \eta(A) 
\end{align*}
Denote $\Omega_{\eta} = U \sN$. It is again easy to check that $S_{\sN | \Omega} = U^* S_{\Omega_{\eta} | \Omega_{\omega}} U$ and hence that $ J_{\sN | \Omega} = U^* J_{\Omega_{\eta}| \Omega_{\omega} }U$. Hence it is sufficient to prove that $J_{\Omega_{\eta} | \Omega_{\omega}} = J_{\Omega_{\omega}}$ for $\omega$ and $\eta$ two normal states on $\sM$ with $\eta$ normal relative to $\omega$ in the canonical GNS representation of $\omega$, which we have already done. 
\end{proof}
In this representation we can also see that the relative modular operator is the object which plays the role of the non-commutative Radon-Nikodym derivative. Classically, if we have two states $\nu$ and $\mu$ with $\nu << \mu$ and an observable $f$ then
\[ \nu(f) = \int f d\nu = \int f \frac{d\nu}{d\mu}d\mu = \mu(f\frac{d\nu}{d\mu}).\]
In the quantum case, if a state $\eta$ is $\omega$-normal then we can define $\Delta_{\eta|\omega}:= \Delta_{\Omega_{\omega}|\Omega_{\eta}}$ and 
 \begin{align*}
( \Omega_{\omega} , \Delta_{\eta|\omega} (\pi_{\omega}(A) \Omega_{\omega} ) ) & = ( [\Id], [\rho_{\eta} A \rho_{\omega}^{-1}]) \\
 & = \textrm{tr}(\rho_{\omega} \rho_{\eta} A \rho_{\omega}^{-1}) \\
 & = \textrm{tr}(\rho_{\eta}A) \\
 & = \eta(A) = (\Omega_{\eta}, \pi_{\omega}(A) \Omega_{\eta} ).
 \end{align*}
\subsection{Modular Dynamics} 
The modular structure is related to the thermodynamic properties of systems described by operator algebras in an interesting way. This will be explored in this section. We consider a von Neumann algebra $\sM$ with cyclic and separating vector $\Omega$. Let $J$ and $\Delta$ be the modular conjugation and modular operator associated to $(\sM, \Omega)$. The following result is an extremely important and useful tool. 
\begin{leftbar}
\begin{thm}[Tomita-Takesaki]
\[ J \sM J = \sM ' \qquad \textrm{and} \qquad \Delta^{it} \sM \Delta^{-it} = \sM \qquad \forall t \in \mathbb{R}\]
\end{thm}
\end{leftbar}
It follows from Tomita-Takesaki that the map 
\[ \mathbb{R} \ni  t \mapsto \sigma_t^{\omega} \in Aut{(\sM)} \qquad \textrm{where} \qquad \sigma_t^{\omega}(A) = \Delta^{-it/ \beta} A \Delta^{it/ \beta} \]
is a $\sigma$-weakly continuous one parameter automorphism group. It is called the \textit{modular dynamics}. 
\newline
\newline
This map can be extended in an appropriate way to the strip $| Im(t) | \leq 1/2 $ for the algebra of analytic elements.
\\
\\
If we start with a $\Cstar$-algebra and a modular faithful state $\omega$, we can use the appropriate modular operator to construct a dynamics for which $\omega$ is a $\beta$-KMS state. Let $(\sH_{\omega}, \pi_{\omega}, \Omega)$ be the GNS representation of $\omega$. We have that $\pi_{\omega}$ is injective and $\Omega$ is cyclic and separating. Let $\Delta$ be the modular operator associated to $(\pi_{\omega}(\sA)'' , \Omega)$. 
\begin{align*} 
\omega(\pi_{\omega}^{-1}(A \sigma_{\omega}^{i\beta}(B))) & = ( \Omega, A \Delta B \Delta^{-1} \Omega )  = ( \Omega, A \Delta B SF \Omega )  = (\Omega, A \Delta B \Omega) \\
& = ( \Omega, A \Delta^{1/2} J J \Delta^{1/2} B \Omega )  = ( \Delta^{1/2} A^* \Omega , J B^* \Omega )  = (B^* \Omega, J \Delta^{1/2} A^* \Omega) \\
& = (\Omega, BA \Omega) \\
& = \omega(BA)
\end{align*}
$\Omega$ is a $\tau, \beta$-KMS state for $\tau^t(A) = \pi_{\omega}^{-1}(\sigma_{\omega}^t(A))$. This means that if we know $\omega$ is a KMS state for a dynamics $\tau$, we have that $\omega$ is faithful. Furthermore, we have the following proposition. 
\begin{prop}
If $(\sA, \tau)$ is a $\Cstar$-dynamical system and $\omega$ is a $\tau$-KMS state, then $\omega$ is modular. 
\end{prop}
Hence, when we move to the $\omega$-GNS representation we have 
\[ \tau^t(A) = \Delta^{-it/\beta}\pi_{\omega}(A)\Delta^{it/\beta}.\]
This is a unitary implementation of the dynamics. It is special, however, because it is the unique unitary implementation which satisfies the property
\[ \Delta^{it/\beta}\sP = \sP\]
so that the natural cone is preserved, if we see vectors as states and view the time evolution in the Schr\"odinger picture. 

\subsection{Perturbed Dynamics}
If $\omega$ is a KMS state for a dynamics $\tau$, the generator $L$ for the unitary group $\sU(t) = \Delta^{-it/\beta}=e^{itL}$ is called the \textit{standard Liouvillean} for $(\sA, \tau, \omega)$. It is the unique self-adjoint operator satisfying 
\[ e^{itL}\pi_{\omega}(A) e^{-itL} = \pi_{\omega}(\tau^t(A)), \;\;\;\;\;\;\;\;\;\; L \Omega_{\omega} = 0, \;\;\;\;\;\;\;\;\;\; e^{-itL}\sP = \sP\]
i.e. it generates an implementation of the dynamics in the GNS rep and preserves the natural cone. From the developments in the last section we can calculate the Liouvillean in terms of the modular operator
\[L_{\omega} = -\frac{1}{\beta} \log \Delta_{\omega} .\]
Now suppose we perturb the dyamics $\tau$ by some bounded self-adjoint operator $V$. We would like to compute the standard Liouvillean for this perturbed dynamics \cite{Pi}. 
\begin{leftbar}
\begin{prop}
Let $(\sA, \tau)$ be a $\Cstar$-dynamical system and $\omega$ a $\tau, \beta$-KMS state. Let $V \in \sA$ be a self-adjoint element and $\tau_{V}$ be the associated perturbed dynamics. The standard Liouvillean of $(\sA, \tau_V, \omega_V)$ is 
\[L_V = L + \pi_{\omega}(V) - J\pi_{\omega}VJ.\]
\end{prop}
\end{leftbar}
\begin{proof}
Both $L + \pi_{\omega}(V)$ and $J\pi_{\omega}(V)J$ are self-adjoint. We can use Trotter's formula to write 
\[ e^{it(L + \pi_{\omega}(V) - J\pi_{\omega}(V)J))}= s-\lim_{n \rightarrow \infty} \left(e^{it(L + \pi_{\omega}(V))/n}e^{it(-J\pi_{\omega}(V)J)/n}\right)^n.\]
Furthermore, using the fact that $J^2=\Id$ and by Tomita-Takesaki
\[e^{-itJ\pi_{\omega}(V)J}=Je^{it\pi_{\omega}(V)}J \in \sM'.\]
Now using Prop. \ref{araki} we can compute 
\begin{align*}
\pi_{\omega}(\tau_V^t(A)) &= e^{it(L+ \pi_{\omega}(V))}\pi_{\omega}(A)e^{-it(L+ \pi_{\omega}(V))} \\
& = s-\lim_{n \rightarrow \infty} \left(e^{it(L+ \pi_{\omega}(V))/n} \right)^n \pi_{\omega}(A) \left(e^{-it(L+ \pi_{\omega}(V))/n} \right)^n \\
& = s-\lim_{n \rightarrow \infty} \left(e^{it(L+ \pi_{\omega}(V))/n} e^{it(-J\pi_{\omega}(V)J)/n}\right)^n  \pi_{\omega}(A) \left( e^{-it(-J\pi_{\omega}(V)J)/n}e^{-it(L+ \pi_{\omega}(V))/n} \right)^n  \\
& = e^{it(L+ \pi_{\omega}(V) - J\pi_{\omega}(V)J)}\pi_{\omega}(A)e^{-it(L+ \pi_{\omega}(V) - J \pi_{\omega}(V)J)} .
\end{align*}
We have that $e^{it\pi_{\omega}(V)} \in \sM$ and $e^{-itJ\pi_{\omega}(V)J} \in \sM'$, hence they commute. Therefore, 
\[ e^{it(\pi_{\omega}(V) - J\pi_{\omega}(V)J)}= e^{it\pi_{\omega}(V)}e^{-itJ\pi_{\omega}(V)J}=e^{it\pi_{\omega}(V)}Je^{it\pi_{\omega}(V)}J .\]
It follows from the definition of the natural cone that $e^{it(\pi_{\omega}(V) - J\pi_{\omega}(V)J)} \sP \subset \sP$. 
Furthermore, we already have that $e^{itL}\sP = \sP$. By Trotter we have that 
\[e^{itL_V} = s-\lim_{n \rightarrow \infty} \left( e^{itL/n}e^{it(\pi_{\omega}(V) - J\pi_{\omega}(V)J)/n} \right)^n.\]
For each fixed $n$, $\left( e^{itL/n}e^{it(\pi_{\omega}(V) - J\pi_{\omega}(V)J)/n} \right)^n \sP \subset \sP$, and since $\sP$ is closed $e^{itL_V}\sP \subset \sP$. 
\end{proof}

\section{Return to Equilibrium Systems}
\subsection{Formal Setup and Notation}
The systems we are interested in studying consist of a small (finite dimensional) quantum system coupled to a quantum thermodynamic reservoir at some fixed inverse temperature $\beta$. We model the small system with a finite dimensional Hilbert space $\sH_S$, a Hamiltonian dynamics given by $H_S \in \sB(\sH_S)$ and an arbitrary initial state $\omega_S$ given by a density matrix $\rho_S$. We denote by $\delta_S := i [H_S, \cdot]$ the generator of the dynamics. We will also consider the Gibbs state at inverse temperature $\beta$ which we denote by $\rho_{\beta} := \frac{e^{-\beta H_S}}{\tr(e^{-\beta H_S})}$. 
\[ \textrm{Small system: }\;\;\;\;\;\;\;\;\;\; \sO_S = \sB(\sH_S) \;\;\;\;\;\;\;\;\;\; \tau_S^t = e^{it\delta_S} \;\;\;\;\;\;\;\;\;\; \omega_S = \tr(\rho_S \cdot )\]
The reservoir is modelled by a $\Cstar$-dynamical system $(\sO_R, \tau_R^t)$. We denote by $\delta_R$ the generator of $\tau_R$. Furthermore we assume that there is a $\tau_R, \beta$-KMS state $\omega_R$ and that the reservoir is in this state.  
\[ \textrm{Reservoir: }\;\;\;\;\;\;\;\;\;\; \sO_R  \;\;\;\;\;\;\;\;\;\; \tau_R^t = e^{it\delta_R} \;\;\;\;\;\;\;\;\;\; \omega_R \]
The joint system is then described by 
\[ \sO = \sO_S \otimes \sO_R, \;\;\;\;\;\;\;\;\;\; \tau_0^t = \tau_S^t \otimes \tau_R^t, \;\;\;\;\;\;\;\;\;\; \omega = \omega_S \otimes \omega_R.\]
The generator of $\tau_0^t$ is $\delta_0 = \delta_S \otimes 0 + 0 \otimes \delta_R$. Furthermore, the state $\omega_{eq}=\omega_{\beta} \otimes \omega_R$ is a $\tau_0, \beta$-KMS state. We now couple the systems by taking a self-adjoint element $V \in \sO$. We define the coupled dynamics $\tau_{\lambda}^t$ as the one generated by 
\[ \delta_{\lambda} = \delta_0 + i\lambda[V, \cdot]. \]
Here $\lambda$ is a parameter which controls the strength of the coupling. This perturbation is of the form described in section \ref{Araki} and we can apply Araki perturbation theory to obtain:
\begin{enumerate}[(i)]
\item For each $\lambda$, there is a unique $\omega_{eq}$-normal $\tau_{\lambda}, \beta$-KMS state. We denote it by $\omega_{\lambda}$.
\item  $\omega_{\lambda}$ and $\omega_{eq}$ are quasi-equivalent. 
\item $\lim_{\lambda \rightarrow 0} \| \omega_{\lambda} - \omega_{eq} \| = 0$. 
\end{enumerate}
We work in the GNS representation of the uncoupled equilibrium state $\omega_{eq}$, given by the triple $(\sH, \pi, \Omega_{eq})$. This can be decomposed as $\sH = \sH_S \otimes \sH_R$, $\pi = \pi_S \otimes \pi_R$ and $\Omega_{eq} = \Omega_S \otimes \Omega_R$ where $(\sH_S, \pi_S, \Omega_S)$ is the GNS rep of $(\sO_S, \omega_{\beta})$ and  $(\sH_R, \pi_R, \Omega_R)$ is the GNS rep of $(\sO_R, \omega_{R})$. We denote the enveloping von Neumann algebra by $\sM = \pi(\sO)'' = \pi_S(\sO_S)'' \otimes \pi_R(\sO_R)''$. The standard Liouvillean of the uncoupled dynamics $\tau_0$ is 
\begin{align}
 L_0 = L_S + L_R = [\pi(H_S), \cdot] - \frac{1}{\beta} \log \Delta_{\omega_R}.
 \end{align}
Let $(J, \Delta, \sP)$ be the modular structure associated to $(\sM, \Omega_{eq})$. The standard Liouvillean of the coupled dynamics $\tau_{\lambda}$ is 
\begin{align}
 L_{\lambda} = L_0 + \lambda\pi(V) - \lambda J\pi(V) J. 
 \end{align}
From Araki's perturbation theory we also know that the perturbed KMS state $\omega_{\lambda}$ has vector representative 
\begin{align}
 \Omega_{\lambda} = \frac{e^{-\beta(L_0 + \lambda\pi(V))/2}\Omega_{eq}}{\|e^{-\beta(L_0 + \lambda\pi(V))/2}\Omega_{eq}\|}
 \end{align}
and that $\Omega_{\lambda} \in \sP$. The initial state $\omega$ has vector representative $\Omega = \rho_S^{1/2} \otimes \omega_R$. 
\subsection{Return to Equilibrium}
In particular we are interested in a subset of the systems described in the previous section which exhibit a specific type of behaviour. Physically, we want the overall system to equilibrate at the temperature of the reservoir after a sufficiently long time. We will now develop what this means formally. 
\begin{defn}
A dynamical system $(\sA, \tau, \omega)$ is called \textit{mixing} if for every state $\eta$ which is $\omega$-normal and every $A \in \sA$,
\[ \lim_{t \rightarrow \infty} \eta(\tau^t(A)) = \omega(A).\]
\end{defn}
This means that, in the large-time limit, all the states that were ``similar" to $\omega$ initially behave like $\omega$ in terms of measurements of any observable. This condition is equivalent to another condition which is used in classical ergodic theory. 
\begin{prop}
A dynamical system $(\sA, \tau, \omega)$, with $\omega$ a $\tau$-invariant state, is mixing if and only if 
\[\lim_{t \rightarrow \infty} \omega(A\tau^t(B))= \omega(A) \omega(B) \]
for all $A,B \in \sA$. 
\end{prop}
\begin{proof}
Suppose the system is mixing. Define a state $\eta_A(B) = \frac{\omega(AB)}{\omega(A)}$. In the $\omega$ GNS rep we have
\[ \eta_A(B) = \frac{\omega(AB)}{\omega(A)} = \tr\left(\frac{\rho_{\omega} \pi(A)}{\tr(\rho_{\omega}\pi(A))} \pi(B)\right) . \]
$\eta_A$ is $\omega$ normal so by assumption
\[\lim_{t \rightarrow \infty}\eta_A(\tau^t(B))= \lim_{t \rightarrow \infty} \frac{\omega(A\tau^t(B))}{\omega(A)}= \omega(B) . \]
Conversely, suppose $\eta$ is an $\omega$-normal state. Then it has a vector rep $\Omega_{\eta}=\lim_{n\rightarrow \infty}[B_n]$ in the $\omega$-GNS rep. 
\[ \eta(\tau^t(A)) = \lim_{n \rightarrow \infty}([B_n],[\tau^t(A)B_n]) = \lim_{n \rightarrow \infty}\omega(B_n^*\tau^t(A)B_n) . \]
Then taking $t \rightarrow \infty$ we get that 
\[ \lim_{t \rightarrow \infty} \eta(\tau^t(A)) = \lim_{n \rightarrow \infty} \omega(B_n^*B_n)\omega(A) = \eta(\Id) \omega(A) = \omega(A). \]
\end{proof}
The classical analog of this condition is that for a measure $\mu$, observables $f,g$ and a time evolution $T$
\[ \lim_{n \rightarrow \infty} \int f\circ T^n g d\mu = \int f d\mu \int g d\mu . \]
This in turn is equivalent to the condition that for measurable sets $A,B$
\[ \lim_{n \rightarrow \infty } \mu (T^nA \cap B) = \mu(A) \mu(B).\]
This has a nice physical interpretation. Imagine a container of water. Suppose at some time a region represented by the set A is occupied by ink. After mixing the water $n$ times the proportion of ink in any region corresponding to a set B is 
\[ \frac{\mu(T^nA \cap B)}{\mu(B)}.\]
If the mixing condition holds, then after many mixings the proportion of ink in region B is simply the proportion of ink in the water overall. The ink and water are evenly distributed. 
\\
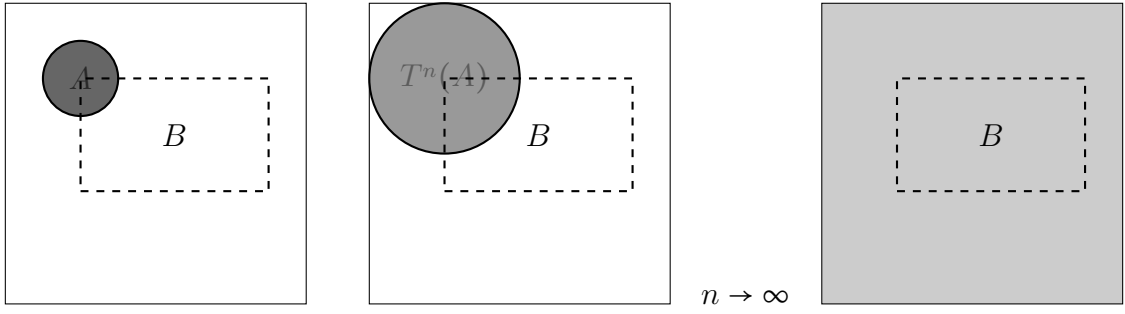
\begin{align*}
\begin{tikzpicture}
\draw (-2,-2) rectangle (2,2) ;
\draw[thick, dashed] ( -1, -0.5) rectangle node {$B$} (1.5, 1) ;
\filldraw[fill=black, fill opacity=0.6, draw=black, thick] (-1,1) circle (0.5) node {$A$};
\end{tikzpicture} 
\;\;\;\;\;\;\;\; 
\begin{tikzpicture}
\draw (-2,-2) rectangle (2,2) ;
\draw[thick, dashed] ( -1, -0.5) rectangle node {$B$} (1.5, 1) ;
\filldraw[fill=black, fill opacity=0.4, draw=black, thick] (-1,1) circle (1) node {$T^n(A)$};
\end{tikzpicture}
\;\;\;\; n \rightarrow \infty \;\;\;\;
\begin{tikzpicture}
\filldraw[fill = black, fill opacity =0.2] (-2,-2) rectangle (2,2) ;
\draw[thick, dashed] ( -1, -0.5) rectangle node {$B$} (1.5, 1) ;
\end{tikzpicture} 
\end{align*}
\\
Returning to our setting, the dynamical condition we will put the system is the following. 
\begin{leftbar}
\begin{ass}
\label{mixing}
There exists $\lambda_0 >0$ such that for all $\lambda$ satisfying $0 < | \lambda | < \lambda_0$ the dynamical system $(\sO, \tau_{\lambda}, \omega_{\lambda})$ is mixing. 
\end{ass}
\end{leftbar}
In other words, for sufficiently small coupling energy, after a long time the coupled system will behave as though it is in the coupled equilibrium state. This is the return-to-equilibrium property. 
\\
\\
An important remark is that in order for the mixing assumption to be satisfied, we \textit{must} have an infinitely extended reservoir. Suppose the reservoir was confined and hence $\sO_R = \sB(\sH_R)$. Now any two states are relatively normal because every GNS representation is the same. In particular two Gibbs states at different temperatures are relatively normal. However, they are both invariant under the dynamics, so clearly if one starts in a Gibbs state at some temperature $\beta' \neq \beta$ it will never end up in a Gibbs state at temperature $\beta$. The mixing condition is violated. The $\Cstar$-algebra formalism is necessary to study return-to-equilibrium systems. 
\\
\\
We can characterize whether dynamical systems are mixing in terms of the spectral properties of the standard Liouvillean. This is called quantum Koopmanism \cite{Pi}. 
\begin{leftbar}
\begin{thm}[Quantum Koopmanism]
\label{Koopman}
A dynamical system $(\sA, \tau, \omega)$ is mixing if and only if 
\begin{align}
 w-\lim _{t \rightarrow \infty} e^{itL_{\omega}} = | \Omega_{\omega} \rangle \langle \Omega_{\omega} |. 
 \end{align}
\end{thm}
\end{leftbar}
\begin{proof}
Suppose $\eta$ is an $\omega$ normal state. Then
\begin{align*} 
\eta(\tau^t(A)) &= (\Omega_{\eta}, e^{itL_{\omega}} \pi_{\omega}(A) e^{-itL_{\omega}} \Omega_{\eta} )\\
& = (\Omega_{\omega}, \Delta_{\eta|\omega}  e^{itL_{\omega}} \pi_{\omega}(A) e^{-itL_{\omega}} \Omega_{\omega} ) \\
& = (\Omega_{\omega}, \Delta_{\eta|\omega} e^{itL_{\omega}} \pi_{\omega}(A) \Omega_{\omega} ).
\end{align*}
The mixing condition $\lim_{t \rightarrow \infty} \eta(\tau^t(A)) = \omega(A)$ is equivalent to 
\[ \lim_{t \rightarrow \infty} (\Omega_{\omega}, \Delta_{\eta|\omega} e^{itL_{\omega}} \pi_{\omega}(A) \Omega_{\omega}) = (\Omega_{\omega}, \pi_{\omega}(A) \Omega_{\omega} ) = (\Omega_{\omega} , \Delta_{\eta|\omega} \Omega_{\omega}) (\Omega_{\omega}, \pi_{\omega}(A) \Omega_{\omega}) \]
noting that $ (\Omega_{\omega} , \Delta_{\eta|\omega} \Omega_{\omega}) = (\Omega_{\eta}, \Omega_{\eta}) = 1$. Since this must hold for all $\eta$ which are $\omega$-normal, the mixing condition is equivalent to 
\[ w-\lim_{t \rightarrow \infty} e^{itL_{\omega}} = |\Omega_{\omega} \rangle \langle \Omega_{\omega} |. \]
\end{proof}
\subsection{Energy Conservation: Averages}
Now that we have formally established the systems under study, we can begin to investigate their behaviour in the double limit $t \rightarrow \infty$ and $\lambda \rightarrow 0$. We will use a two-time measurement protocol. We make one measurement of $H_S$, the system energy, at time $0$ and another at time $t$ and take the difference of the expected value of these two measurements. We do the same for the reservoir, although since it has infinite energy we will first need to rewrite the change in expected energy in terms of an energy flux observable. The first law says that after the system returns to equilibrium ($t \rightarrow \infty$) and we decouple the system and reservoir ($\lambda \rightarrow 0$) the change in energy of the system is equal to the negative change of energy of the reservoir. Note that even classically energy conservation is only expected to hold after the systems have been decoupled because when $\lambda >0$ there is some coupling energy in the joint system which will not be taken into account in this measurement protocol. 
\\
\\
For the small system, the quantity of interest is 
\[ \Delta Q_S(\lambda, t) = \omega(\tau_{\lambda}^t(H_S)) - \omega(H_S)\]
the change in expected value of energy between time $0$ and $t$. Recall that $\omega$ is the initial state of the joint system and reservoir. By the mixing hypothesis, for sufficiently small $\lambda$ we have
\[ \Delta Q_S(\lambda) := \lim_{t \rightarrow \infty} \Delta Q_S(\lambda, t) = \omega_{\lambda}(H_S) - \omega(H_S).\]
Araki's perturbation theory then gives 
\[ \Delta Q_S := \lim_{\lambda \rightarrow 0} \Delta Q_S(\lambda) = \omega_{eq}(H_S) - \omega(H_S).\]
For the reservoir we do not have a bounded Hamiltonian $H_R$, so we need to rewrite this quantity in terms of a flux observable that will make sense even in the infinitely extended case. To motivate the definition of this flux observable, suppose the reservoir is confined and is given by a finite dimensional Hilbert space. In this case the generator of the dynamics $\tau_R$ is $\delta_R = i[H_R, \cdot]$ and the generator for the coupled dynamics $\tau_{\lambda}$ is $\delta_{\lambda} = i[H_{\lambda}, \cdot]$ where $H_{\lambda} = H_S + H_R + \lambda V$. 
\begin{align*}
\Delta Q_R(\lambda, t) = \omega(H_R) - \omega(\tau_{\lambda}^t(H_R)) & = \int_0^t \frac{d}{ds} \left( -\omega(\tau_{\lambda}^s(H_R)) \right) ds \\
& = \int_0^t \omega(\tau_{\lambda}^s(-\delta_{\lambda}(H_R)))ds \\
& = \int_0^t \omega(\tau_{\lambda}^s(\lambda \delta_R(V)))ds
\end{align*}
This last line follows because $\delta_{\lambda}(H_R) = -i[H_{\lambda}, H_R] = -i [\lambda V, H_R] = i\lambda [H_R, V] = \lambda \delta_R(V)$. Therefore we make the following assumption
\begin{leftbar}
\begin{ass}
\label{regularity}
The perturbation $V$ is in the domain of $\delta_R$. 
\end{ass}
\end{leftbar}
We then define the energy flux to be $\phi_R = \lambda \delta_R(V)$, and 
\[\Delta Q_R(\lambda, t) := \int_0^t \omega(\tau_{\lambda}^s(\phi_R)) ds.\]
By exactly the same argument we can define $\phi_S = \lambda \delta_S(V)$ and write $\Delta Q_S(\lambda, t) := \int_0^t \omega(\tau_{\lambda}^s(\phi_S)) ds$. It follows that
\begin{align*}
 \Delta Q_R(\lambda, t) - \Delta Q_S(\lambda, t) &= \int_0^t \omega(\tau_{\lambda}^s(\phi_R - \phi_S)) ds \\
 & = \int_0^t \omega(\tau_{\lambda}^s(\lambda \delta_R(V) + \lambda \delta_S(V))) ds \\
 & = \int_0^t \omega(\tau_{\lambda}^s(\lambda \delta_{\lambda}(V)))ds \\
 & = \lambda \omega(\tau_{\lambda}^t(V) - V)
 \end{align*}
 and hence by mixing 
 \[ \Delta Q_R (\lambda) := \lim_{t \rightarrow \infty} \Delta Q_R(\lambda, t) = \Delta Q_S + \lambda \left( \omega_{\lambda}(V) - \omega(V) \right).\]
 By Araki perturbation theory the term in the brackets converges as $\lambda \rightarrow 0$ so
  \[\Delta Q_R := \lim_{\lambda \rightarrow 0} \Delta Q_R(\lambda) = \Delta Q_S. \]
  This result is well-known and is a formulation of the first law of thermodynamics in terms of averages. 
  \subsection{Full Counting Statistics}
  Using the same two-time measurement protocol, we can take into account at each step the full probability distribution of outcomes for a measurement of energy. This allows us to define a measure which encodes the full statistics of the energy fluctuations. We begin with the small system. 
  Let $H_S = \sum_i \lambda_i P_i$ be the spectral decomposition of the Hamiltonian. If we make a measurement of energy at time $0$ the result is the eigenvalue $\lambda_i$ with probability 
  \[ \omega_S(P_i) = \tr(\rho_{\omega_S} P_i).\]
  Given that an energy $\lambda_i$ was observed the new state of the system is
  \[ \tr\left(\frac{P_i \rho_{\omega_S}P_i }{\tr(\rho_{\omega_S} P_i)} \; \cdot \; \right).\]
  Evolve the system under the coupled dynamics for time $t$ and then make another measurement of $H_S$. Eigenvalue $\lambda_j$ is observed with probablity 
  \[ \left(\tr\left(\frac{P_i \rho_{\omega_S}P_i }{\tr(\rho_{\omega_S} P_i)} \; \cdot \; \right) \otimes \omega_R \right) (\tau_{\lambda}^t(P_j)).\]
  The joint probability distribution of the two measurements gives a distribution for the change in energy 
  \begin{align}
  \label{system FCS}
  \mathbb{P}_{S, \lambda, t} (x) = \sum_{\lambda_j - \lambda_i = x} \left(\tr\left(P_i \rho_{\omega_S}P_i  \; \cdot \; \right) \otimes \omega_R \right) (\tau_{\lambda}^t(P_j)).\end{align}
  This measure is called the system energy full counting statistics. It is supported on the finite set $H_S - H_S$ and is hence an atomic probability measure. We can compute the limiting measure. First by mixing 
  \[\mathbb{P}_{S, \lambda} (x) = \lim_{t \rightarrow \infty}\mathbb{P}_{S, \lambda, t} (x) = \sum_{\lambda_j - \lambda_i = x} \omega_{\lambda}(P_j) \tr(\rho_{\omega_S} P_i)\]
  and then by Araki perturbation theory 
  \[ \mathbb{P}_S(x) = \lim_{\lambda \rightarrow 0} \mathbb{P}_{S, \lambda}(x) = \sum_{\lambda_j - \lambda_i} \omega_{eq}(P_j) \omega(P_i). \]
  Computing the expectation of this measure yields 
  \begin{align*}
  \langle x \rangle_{\mathbb{P}_S} & = \sum_{\lambda_i, \lambda_j \in \sigma(H_S)} (\lambda_j - \lambda_i) \omega_{eq}(P_j) \omega(P_i) \\
  & = \omega_{eq}(H_S) - \omega(H_S) = \Delta Q_{S}.
  \end{align*}
  In the limit, the expectation value of this measure is equal to the change in expectation value between the two measurements. 
  \\
  \\
  The definition of $\mathbb{P}_{S, \lambda, t}$ relies heavily on the finite dimensional structure of $\sH_S$. To define the full counting statistics for the reservoir we will need to find a way to rewrite it in terms of objects which survive the thermodynamic limit. To do this, we will take advantage of the modular structure. Let $\eta = \Id \otimes \omega_R$ which is an $\omega_{eq}$ normal state with vector representative $\Id \otimes \Omega_R$ in the $\omega_{eq}$-GNS rep. 
  \begin{defn}
  $\mathbb{P}_{R, \lambda, t}$ is the spectral measure of the operator defined on $\sH_{\omega_{eq}}$ 
  \[ \frac{1}{\beta} \log(\Delta_{\eta \circ \tau_{\lambda}^{-t}| \eta})\]
  and the vector $\Omega \in \sH_{\omega_{eq}}$ (which is the vector representative of the initial state $\omega$). 
  \end{defn}
  To motivate this definition, suppose that the reservoir is confined so that $\sO_R = \sB(\sH_R)$, $\tau_R$ is given by a Hamiltonian $H_R$ and $\omega_R = \tr(\rho_{\omega_R} \cdot )$ is a Gibbs state. Furthermore suppose we are in the standard GNS representation of $(\sO_R, \omega_R)$. Then we can compute explicitly
  \[ \Omega_{eq}= \rho_{\omega_{\beta}}^{1/2} \otimes \rho_{\omega_R}^{1/2}, \;\;\;\;\;\;\;\;\;\;  \Omega = \rho_{\omega_{S}}^{1/2} \otimes \rho_{\omega_R}^{1/2}, \;\;\;\;\; \;\;\;\;\;\Omega_{\eta} = \Id \otimes \rho_{\omega_R}^{1/2} , \;\;\;\;\;\;\;\;\;\; \Omega_{\eta \circ \tau_{\lambda}^{-t}} = e^{itH_{\lambda} }(\Id \otimes \rho_{\omega_R}^{1/2} )e^{-itH_{\lambda}},\]
  where $\rho_{\omega_R} = e^{-\beta H_R} /\tr(e^{-\beta H_R})$. Furthermore, the action of the modular operator is 
  \[ \Delta_{\eta \circ \tau_{\lambda}^{-t} | \eta} X = \Omega_{\eta \circ \tau_{\lambda}^{-t}} X \Omega_{\eta}^{-1}\]
  for $X  \in \sO$. An explicit computation of the characteristic function, using all these ingredients, shows that $\mathbb{P}_{R, \lambda, t}$ is of the same form as equation (\ref{system FCS}) in the case of a confined reservoir. 
  \begin{align*}
 & \int_{\mathbb{R}} e^{i\alpha x} d\mathbb{P}_{R, \lambda, t} (x) \\
  &= ( \Omega, \Delta^{i\alpha/\beta}_{\eta \circ \tau_{\lambda}^{-t} | \eta} \Omega )_{\sH_{\omega_{eq}}} \\
  & = (\rho_{\omega}^{1/2}, \rho_{\eta \circ \tau_{\lambda}^{-t}}^{i\alpha/\beta} \rho_{\omega}^{1/2} \rho_{\eta}^{-i\alpha/\beta}) \\
  & = \tr((\Id \otimes e^{-i\alpha H_R}/Z^{i\alpha/\beta}) e^{-itH_{\lambda}}(\rho_{\omega_S}^{1/2} \otimes \rho_{\omega_R}^{1/2})(\Id \otimes e^{i\alpha H_R}/Z^{-i\alpha/\beta})(\rho_{\omega_S}^{1/2} \otimes \rho_{\omega_R}^{1/2})e^{itH_{\lambda}}) \\
  & = \tr((\Id \otimes e^{-i\alpha H_R})e^{-itH_{\lambda}} (\Id \otimes e^{i\alpha H_R})(\rho_{\omega_S} \otimes \rho_{\omega_R}) e^{itH_{\lambda}}) \\
  & = \sum_{e, e' \in \sigma(H_R)} e^{i\alpha(e-e')} \tr((\Id \otimes P_{e'}) e^{-itH_{\lambda}} (\Id \otimes P_e)(\rho_{\omega_S} \otimes \rho_{\omega_R})e^{itH_{\lambda}}) \\
  & = \sum_{e, e' \in \sigma(H_R)} e^{i\alpha(e' - e)} (\omega_S \otimes \tr(P_e \rho_{\omega_R} P_e \cdot)) (\tau_{\lambda}^t (P_{e'}))
  \end{align*}
  This is exactly what would have been obtained if the reservoir had been treated like the small system, with the sign changed since we are measuring the negative change in energy of the reservoir. However, when formulated in terms of the modular operator the definition of the measure still makes sense in the infinitely extended setting. In this setting, the reservoir energy full counting statistics could be look very different from $\mathbb{P}_{S, \lambda, t}$. In concrete examples $\mathbb{P}_{R, \lambda, t}$ is typically absolutely continuous and supported on all of $\mathbb{R}$. 
  \\
  \\
  A result of \cite{JP6}, the proof of which is involved but uses the tools of modular theory and Araki perturbation theory, is the following. 
  \begin{thm}
 \[ \log \Delta_{\eta \circ \tau^{-t}_V| \eta} = \log \Delta_{\eta}  + \beta \int_0^t \tau_{\lambda}^s(\phi_R) ds \]
  \end{thm}
  Since $\log \Delta_{\eta} \Omega = 0$ it follows that 
  \begin{align}
  \label{average}
   \langle x \rangle_{\mathbb{P}_{R, \lambda, t}} = \int x d\mathbb{P}_{R, \lambda, t}(x) = (\Omega, \frac{1}{\beta} \log \Delta_{\eta \circ \tau^{-t}_V| \eta} \Omega) = \int_0^t  \omega (\tau_{\lambda}^s(\phi_R) )ds  = \Delta Q_R.
   \end{align}
  Hence if we are able to prove convergence of the FCSs $\mathbb{P}_{S, \lambda, t}$ and $\mathbb{P}_{R, \lambda, t}$, the convergence in terms of averages follows as a consequence and we will have a  much stronger version of the first law. 
  \section{Strong Energy Conservation}
  \subsection{Main Result}
  The main result is that the full counting statistics do indeed converge as measures in the large time and small coupling energy limits. 
  \mdfsetup{linewidth=3pt, }
  \begin{mdframed}
  \begin{thm}
  \label{main}
  Suppose Assumptions (\ref{mixing}) and (\ref{regularity}) hold and let $0 < |\lambda| < \lambda_0$. The weak limits 
  \begin{align*}
  \mathbb{P}_{R, \lambda} = \lim_{t \rightarrow \infty} \mathbb{P}_{R, \lambda, t}, \;\;\;\;\;\;\;\;\;\;\;\;\;\;\;\;\;\;\;\; \mathbb{P}_{R} = \lim_{\lambda \rightarrow 0} \mathbb{P}_{R, \lambda}
  \end{align*}
  exist and 
  \[\mathbb{P}_{R} = \mathbb{P}_{S}. \]
  \end{thm}
  \end{mdframed}
  It is a well-known fact about weak convergence of probability measures (see for example \cite{Bi} Thm. 26.3) that 
  \[ \mu_n \rightarrow \mu \iff \chi_n (\gamma) \rightarrow \chi (\gamma) \;\;\;\;\; \forall \gamma \in \mathbb{R}\]
  where $\chi(\gamma) = \int_{\mathbb{R}} e^{i\gamma x} d\mu(x)$ is the characteristic function of $\mu$. The proof will be centred on the object 
  \[ \sF_{\lambda,t}(\alpha) = \int_{\mathbb{R}} e^{\alpha \beta x} d\mathbb{P}_{R,\lambda,t} = (\Omega, \Delta^{\alpha}_{\eta \circ \tau_{\lambda}^{-t}|\eta} \Omega). \] 
  The characteristic function of $\mathbb{P}_{R,\lambda,t}$ is $\sF_{\lambda,t}(i\gamma/\beta)$ and hence Theorem \ref{main} is equivalent to the following three statements. The limit  \begin{align} \sF_{\lambda}(i\gamma/\beta) = \lim_{t \rightarrow \infty} \sF_{\lambda, t}(i\gamma/\beta)\end{align}
  exists for all $\gamma \in \mathbb{R}$, 
the map 
  \begin{align} \gamma \mapsto \sF_{\lambda}(i\gamma/\beta) \end{align}
  is a continuous function, and 
\begin{align} \lim_{\lambda \rightarrow 0} \sF_{\lambda}(i\gamma/\beta) = \omega_{\beta}(e^{i\gamma H_S}) \omega_S(e^{-i\gamma H_S}) \end{align}
  where the RHS is the characteristic function of $\mathbb{P}_{S, \lambda, t}$. 
  \[ \int_{\mathbb{R}} e^{i\gamma x} d\mathbb{P}_{S, \lambda, t}(x) = \sum_{e,e' \in \sigma(H_S)} \omega_{\beta}(e^{i\gamma e'}P_{e'})\omega_S(e^{-i\gamma e}P_e) = \omega_{\beta}(e^{i\gamma H_S}) \omega_S(e^{-i\gamma H_S}) \]
  
  \subsection{Proofs}
  The first step in the proof is to compute $\sF_{\lambda, t}(\alpha)$ on the line $i\mathbb{R} + 1/2$. To do so we will use the following result from \cite{JP3}. 
  \begin{align}
  \label{cocycle}
  \Delta_{\eta\circ\tau_\lambda^{-t}|\eta}
=e^{\i t(L_0+\lambda\pi(V))}\Delta_\eta e^{-\i t(L_0+\lambda\pi(V))}
=\Gamma_\lambda(t)\Delta_\eta\Gamma_\lambda^*(t),
  \end{align}
  where $\Gamma_\lambda(t):= e^{\i t(L_0+\lambda\pi(V))} e^{-\i tL_0}$ is unitary. $\Gamma_{\lambda}(t)$ satisfies the Cauchy problem 
  \[ \partial_t \Gamma_{\lambda}(t) = i \lambda \Gamma_{\lambda}(t)\pi(\tau_0^t(V)), \;\;\;\;\;\;\;\;\;\; \Gamma_{\lambda}(0) = \Id.\]
Hence for any $B \in \sM'$, $[B, \pi(\tau_0^t(V))]=0$ and 
\[ \partial_t [B, \Gamma_{\lambda}(t)] = i \lambda [B, \Gamma_{\lambda}(t)] \pi(\tau_0^t(V)), \;\;\;\;\;\;\;\;\;\; [B, \Gamma_{\lambda}(0)] = 0.\]
By the uniqueness of solutions to the Cauchy problem, $[B, \Gamma_{\lambda}(t)] = 0$ and hence 
\begin{align}
\label{unique}
\Gamma_{\lambda}(t) \in \sM'' = \sM.
\end{align}
\begin{leftbar}
\begin{lemma} 
\label{lemma1}
Set 
\[ \widehat{L_{\lambda}} = L_R + \pi(H_S + \lambda V), \;\;\;\;\;\;\;\;\;\; \widehat{\Omega} = \pi(\rho_S^{1/2} \otimes \Id) \Omega.\]
For all $\lambda, t, s \in \mathbb{R}$, one has
\[\sF_{\lambda, t}\left(is + \frac12\right) = \left( e^{i\beta s \widehat{L_{\lambda}}}\widehat{\Omega}, e^{itL_{\lambda}}e^{i\beta s \widehat{L_{\lambda}}}\Omega_{\eta}\right). \]
\end{lemma}
\end{leftbar}
\begin{proof}
Let $R = \pi(\rho_S^{1/2} \otimes \Id) \in \sM$. Using Equation ($\ref{cocycle}$) as well as the facts that $R \Omega_{\eta} = \rho_S^{1/2} \otimes \Omega_R = \Omega$ and $ \Delta_{\eta}^{1/2} A \Omega_{\eta} = J A^* \Omega_{\eta}$ we have 
\begin{align*}
\Delta_{\eta \circ \tau_{\lambda}^{-t} | \eta}^{1/2} \Omega = \Gamma_{\lambda}(t) \Delta_{\eta}^{is} \Delta_{\eta}^{1/2} \Gamma_{\lambda}(t)^* R \Omega_{\eta} = \Gamma_{\lambda}(t) \Delta_{\eta}^{is} J R \Gamma_{\lambda}(t) \Omega_{\eta}. 
\end{align*}
Since $J \Delta_{\eta}^{is} = \Delta_{\eta}^{is} J$ and $R$ commutes with $\Delta_{\eta}^{is} = e^{-is \beta L_R}$,
\begin{align*}
\Gamma_{\lambda}(t) \Delta_{\eta}^{is} J R \Gamma_{\lambda}(t) \Omega_{\eta} = \Gamma_{\lambda}(t) J R\Delta_{\eta}^{is}\Gamma_{\lambda}(t) \Omega_{\eta} = \Gamma_{\lambda}(t) (J RJ) J \Delta_{\eta}^{is}\Gamma_{\lambda}(t) \Omega_{\eta}.   
\end{align*}
$JRJ \in \sM'$ and $\Gamma_{\lambda}(t) \in \sM$. This and the identity $\Delta_{\eta}^{-is}J \Omega_{\eta} = \Omega_{\eta}$ yield,
\begin{align*}
\Gamma_{\lambda}(t) (J RJ) J \Delta_{\eta}^{is}\Gamma_{\lambda}(t) \Omega_{\eta} = (JRJ) \Gamma_{\lambda}(t)  \Delta_{\eta}^{is}\Gamma_{\lambda}(t) \Delta_{\eta}^{-is}J \Omega_{\eta} = JRJ \Gamma_{\lambda}(t) \widetilde{\Gamma_{\lambda}}(t)  \Omega_{\eta},
\end{align*}
where $\widetilde{\Gamma_{\lambda}}(t)  =  J \Delta_{\eta}^{is} \Gamma_{\lambda}(t) \Delta_{\eta}^{-is} J$. All of this combined with $JRJ \Omega = \widehat{\Omega}$ gives
\begin{align}
\label{mid lemma}
\sF_{\lambda, t}\left(is + \frac12\right) = \left( \Omega, \Delta_{\eta \circ \tau_{\lambda}^{-t} | \eta}^{is + 1/2} \Omega \right) = \left( \widehat{\Omega}, \Gamma_{\lambda}(t) \widetilde{\Gamma_{\lambda}}(t) \Omega_{\eta} \right). 
\end{align}
For any $A \in \sB(\sH)$ self adjoint, define $A_t := e^{itL_R}Ae^{-itL_R}$ and let $U_A^t$ be the solution to the Cauchy problem 
\begin{align}
\partial_t U_A^t = iU_A^t A_t, \;\;\;\;\;\;\;\;\;\; U_A^0 = \Id. 
\end{align}
One can check that $U_A^t = e^{it(L_R + A)}e^{-itL_R}$. By a similar uniqueness argument to that which was used to obtain (\ref{unique}), one shows that if $A \in \sM$ then $U_A^t \in \sM$. Now if $A, B \in \sM$, we can argue by uniqueness of the solution again.
\begin{align*}
\partial_t(U_A^t J U_B^t J) & = \partial_t U_A^t J U_B^t J + U_A^tJ\partial_tU_B^t J \\
& = (iU_A^t A_t)J U_B^t J + U_A^tJ (iU_B^t B_t) J \\
& = iU_A^t J U_B^t J A_t - iU_A^tJ U_B^t J J B_t J \\
& = iU_A^t J U_B^t J (A_t - J B_t J) = iU_A^t J U_B^t J (A - JBJ)_t
\end{align*}
where the last equality follows since $JB_tJ = J \Delta^{it/\beta}B\Delta^{-it/\beta}J = \Delta^{it/\beta}JBJ\Delta^{-it/\beta} = (JBJ)_t$. Hence,
\begin{align}
U_{A - JBJ}^t = U_A^t J U_B^t J. 
\end{align}
Similarly, one shows that 
\begin{align}
U_{A_{-i\beta s }}^t = \Delta_{\eta}^{is} U_A^t \Delta_{\eta}^{-is}.
\end{align}
Set $M = \pi(H_S + \lambda V)$. Since $L_0 = L_R + L_S = L_R + \pi(H_S) - J\pi(H_S)J$, we have that $L_0 + \lambda \pi(V) = L_R + M - J\pi(H_S)J$. It follows from the identities we have just developed that
\begin{align}
\Gamma_{\lambda}(t) = U_M^t {U_{\pi(H_S)}^t}^* ,\;\;\;\;\;\;\;\;\;\; \widetilde{\Gamma_{\lambda}}(t) = (J U_{M_{-i\beta s}}^tJ)(J U_{\pi(H_S)^t}J)^*,
\end{align}
and hence
\begin{align*}
\Gamma_{\lambda}(t) \widetilde{\Gamma_{\lambda}}(t) = (U_M^tJU_{M_{-i\beta s}}^t J)(U_{\pi(H_S)}^tJU_{\pi(H_S)}^tJ)^* =U_M^tJU_{M_{-i\beta s}}^t J e^{itL_R}e^{-tL_0} .
\end{align*}
From the fact that $L_{\lambda} = L_R + M - JMJ $ and $\widehat{L}_{\lambda} = L_R + M$ we deduce 
\begin{align*}
e^{itL_{\lambda}}e^{i\beta s \widehat{L}_{\lambda}} & = U_M^t J U_M^t J e^{itL_{R}}U_M^{\beta s}e^{i\beta s L_R} \\
& = U_M^t \left( J U_M^t J \right) \left( e^{itL_{R}}U_M^{\beta s}e^{-itL_R} \right) e^{itL_{R}}e^{i\beta s L_R} \\
& = U_M^t  \left( e^{itL_{R}}U_M^{\beta s}e^{-itL_R} \right)  e^{i\beta s L_R} e^{-i\beta s L_R} \left( J U_M^t J \right) e^{i\beta s L_R} e^{itL_{R}}\\
& = e^{it\widehat{L}_{\lambda}} e^{i\beta s \widehat{L}_{\lambda}}e^{-itL_R}(JU_{M_{-i\beta s}}^t J) e^{itL_R} \\ 
& = e^{i\beta s \widehat{L}_{\lambda}}U_M^t JU_{M_{-i\beta s}}^t J e^{itL_R} \\ 
& = e^{i\beta s \widehat{L}_{\lambda}}\Gamma_{\lambda}(t) \widetilde{\Gamma_{\lambda}}(t) e^{itL_0}.
\end{align*}
Combining this with (\ref{mid lemma}) and the fact that $e^{itL_0}\Omega_{\eta} = \Omega_{\eta}$ gives the result. 
\end{proof}
\begin{rem}
The proof of this lemma did not use assumptions (\ref{mixing}) or (\ref{regularity}). 
\end{rem}
The next step of the proof is to study the behaviour of $\sF_{\lambda, t}(\alpha)$ on the strip $\mathfrak{S}([0,1])$ where 
\[ \mathfrak{S}(S) := \{ z \in \mathbb{C} : \textrm{Re}z \in S \}.\]
\begin{leftbar}
\begin{lemma}
\label{lemma2}
(i)  For all $\lambda, t \in \mathbb{R}$, the map $\alpha \mapsto \sF_{\lambda , t}(\alpha)$ is analytic on the interior of $\mathfrak{S}([0,1])$ and bounded and continuous on $\mathfrak{S}([0,1])$. Moreover, for each fixed $\alpha \in \mathfrak{S}([0,1])$,
\[ \sup_{t \in \mathbb{R} } | \sF_{\lambda,t} (\alpha) | \leq 1 + (\textrm{dim}(\sH_S) -1) \textrm{Re}(\alpha).\]
(ii) For all $\lambda \in \mathbb{R}$ and $\delta \in ]0,1[$, one has the bound 
\[ \sup_{t\in \mathbb{R}, \alpha \in \mathfrak{S}(]0,\delta[)} | \partial_{\alpha} \sF_{\lambda, t}(\alpha) | < \infty. \]
\end{lemma}
\end{leftbar}
\begin{proof}
From (\ref{cocycle}) we have 
\begin{align*}
\sF_{\lambda,t}(1) =\| \Delta_{\eta\circ\tau_{\lambda}^{-t}|\eta}^{1/2}\Omega\|^2 =\|\Delta_\eta^{1/2} \Gamma_\lambda^*(t)\pi(\rho_S^{1/2} \otimes\Id)\Omega_\eta\|^2.
\end{align*}
Using the fact that $J$ is anti-unitary and (\ref{unique}) yields 
\begin{align*}
0 \leq \sF_{\lambda,t}(1) =\|J\Delta_\eta^{1/2} \Gamma_\lambda^*(t)\pi(\rho_S^{1/2} \otimes\Id)\Omega_\eta\|^2
=\|\pi(\rho_S^{1/2}\otimes\Id)\Gamma_\lambda(t)\Omega_\eta\|^2\leq \|\Omega_\eta\|^2=\dim\sH_S.
\end{align*}
From this inequality and convexity of the exponential function we get for $\alpha \in [0,1]$. 
\[ \sF_{\lambda, t}(\alpha) \leq (1-\alpha) \sF_{\lambda, t}(0) + \alpha  \sF_{\lambda, t}(1) \leq 1 + (\textrm{dim}(\sH_S) - 1)\alpha.\]
The bound in part {\it (i)} follows. Write $\sF_{\lambda, t}(\alpha)$ as the sum of two Laplace transforms
\[ \sF_{\lambda, t}^{\pm}(\alpha) = \int_{\mathbb{R}^{\pm}} e^{\alpha \beta x} d\mathbb{P}_{R, \lambda, t}(x). \]
Laplace transforms are analytic wherever they converge absolutely. Applying the bound we have just derived, $\sF_{\lambda, t}^{+}(\alpha)$ is analytic on $\mathfrak{S}(]-\infty, 1[)$ and $\sF_{\lambda, t}^{-}(\alpha)$ is analytic on $\mathfrak{S}(]0, \infty[)$. Hence $\sF_{\lambda, t}(\alpha)$ is analytic on $\mathfrak{S}(]0,1[)$. For $\alpha, \alpha' \in \mathfrak{S}(]-\infty,1])$,
\[ | \sF^+_{\lambda, t}(\alpha) - \sF^+_{\lambda, t}(\alpha') | \leq \int_{\mathbb{R}^+} | e^{\alpha \beta x} - e^{\alpha' \beta s}| d \mathbb{P}_{R, \lambda, t}(x). \]
Hence, $\sF_{\lambda, t}^+$ is continuous on $\mathfrak{S}(]-\infty,1])$. Similarly, $\sF_{\lambda, t}^-$ is continuous on $\mathfrak{S}([0,\infty[)$, and hence $\sF_{\lambda, t}$ is continuous on $\mathfrak{S}([0,1])$. 
\\
\\
By (\ref{average}) one has 
\[ \Delta {Q_R} (\lambda, t) = - \int_{\mathbb{R}^-} |x| d\mathbb{P}_{R,\lambda, t}(x) +  \int_{\mathbb{R}^+} x d\mathbb{P}_{R,\lambda, t}(x). \]
Combining this with the bound in part \textit{(i)} and the inequality $e^x > x$ for $x \in \mathbb{R}$ yields
\[ \int_{\mathbb{R}^+} x d\mathbb{P}_{R,\lambda, t}(x) \leq \frac{1}{\beta} \int_{\mathbb{R}^+} e^{\beta x} d\mathbb{P}_{R,\lambda, t}(x) \leq \frac{1}{\beta} \sF_{\lambda, t}(1) \leq  \frac{1}{\beta} \textrm{dim}(\sH_S)\]
and hence
\[ \int_{\mathbb{R}^-} |x| d\mathbb{P}_{R, \lambda, t} (x) \leq \frac{1}{\beta} \textrm{dim}(\sH_S) - \Delta {Q_R} (\lambda, t) . \]
We use the property of the Laplace transform $\mathcal{L}(f(x)) '  = - \mathcal{L}(xf(x))$ and \textit{(i)} to get that for $\alpha \in \mathfrak{S}([0,1])$, 
\[ \partial_{\alpha} \sF_{\lambda,t} (\alpha) = \partial_{\alpha} \int_{\mathbb{R}} e^{\alpha \beta x } d\mathbb{P}_{R, \lambda, t} (x) = \beta  \int_{\mathbb{R}} x e^{\alpha \beta x} d\mathbb{P}_{R, \lambda, t}(x) . \]
The estimate $xe^{ax} \leq (1-a)^{-1} e^x$ is valid for $a < 1$ and $x\in \mathbb{R}$. Combining this with the previous estimates, for $\alpha \in \mathfrak{S}([0,1])$, 
\begin{align*}
| \partial_{\alpha} \sF_{\lambda, t}(\alpha) | & \leq - \beta \int_{\mathbb{R}^-} |x| d\mathbb{P}_{R,\lambda, t}(x) +   \int_{\mathbb{R}} \beta x e^{\textrm{Re}\alpha \beta x} d\mathbb{P}_{R, \lambda, t}(x) \\
& \leq \textrm{dim}(\sH_S) - \beta \Delta Q_{R}(\lambda, t) + (1 - \textrm{Re}(\alpha))^{-1} \sF_{\lambda, t}(1) \\
& \leq (1 + (1 - \textrm{Re}(\alpha))^{-1}) \textrm{dim}(\sH_S) - \beta \Delta Q_{R}(\lambda, t).
\end{align*}
Finally, we can bound the last term by writing 
\[
\Delta Q_R(\lambda, t) = \lambda \omega(\tau_{\lambda}^t(V) - V) + \Delta Q_S (\lambda, t) = \omega(\tau_{\lambda}^t (H_S + \lambda V)) + \omega(H_S + \lambda V).
\]
Since $\| \omega \| = 1$ and $\tau_{\lambda}^t$ is an automorphism, $| \Delta Q_R(\lambda, t) | \leq 2 \| H_S + \lambda V \|$. 
\end{proof}
Equipped with these two lemmas, we can investigate the limiting behaviour of $\sF_{\lambda, t}(\alpha)$. The argument is based on Vitali's convergence theorem (see Appendix B of \cite{JOPP}). 
\begin{leftbar}
\begin{thm}[Vitali Convergence]
Let $D_{\epsilon}$ be an open disc in $\mathbb{C}$ of radius $\epsilon$ centred at $z_0$ and let $F_t : D_{\epsilon} \rightarrow \mathbb{C}$ for $t \in \mathbb{R}^+$ be analytic functions such that 
\[ \sup_{z \in D_{\epsilon}, t>0} | F_t(z) | < \infty. \]
Suppose the limit
\[ \lim_{t \rightarrow \infty} F_t(z) = F(z)\]
exists for all $z \in D_{\epsilon} \cap S$ where $S$ is a line through $z_0$. Then the limit $F(z)$ exists for all $z \in D_{\epsilon}$ and $z \mapsto F(z)$ is analytic on $D_{\epsilon}$. Moreover, as $t \rightarrow \infty$, all derivatives of $F_t$ converge uniformly on compact subsets of $D_{\epsilon}$ to the corresponding derivatives of $F$. 
\end{thm}
\end{leftbar}
We first take the large time limit. 
\begin{leftbar}
\begin{prop}
\label{tlimit}
Suppose that $0 < |\lambda| < \lambda_0$. Then the limit 
\[\sF_{\lambda}(\alpha) = \lim_{t \rightarrow \infty} \sF_{\lambda, t}(\alpha)\]
exists for $\alpha \in \mathfrak{S}([0,1[)$. The function $\alpha \mapsto \sF_{\lambda}(\alpha) $ is analytic on $\mathfrak{S}(]0,1[)$ and continuous on $\mathfrak{S}([0,1[)$. Moreover, for $s \in \mathbb{R}$
\[ \sF_{\lambda}\left(\frac{1}{2} +is \right) = ( \widehat{\Omega}, e^{i \beta s\widehat{L}_{\lambda}} \Omega_{\lambda} ) ( \Omega_{\lambda} , e^{-i \beta s\widehat{L}_{\lambda}}\Omega_{\eta} ).  \]
\end{prop}
\end{leftbar}
\begin{proof}
Lemma \ref{lemma1}  and Theorem \ref{Koopman} yield for $0 < |\lambda | < \lambda_0$ and $s \in \mathbb{R}$, 
\[ \lim_{t \rightarrow \infty} \sF_{\lambda, t} \left(\frac{1}{2} +is \right) = ( \widehat{\Omega}, e^{i \beta s\widehat{L}_{\lambda}} \Omega_{\lambda} ) ( \Omega_{\lambda} , e^{-i \beta s\widehat{L}_{\lambda}}\Omega_{\eta} ). \]
Lemma \ref{lemma2} \textit{(i)} and Vitali's convergence theorem imply that the limit $\sF_{\lambda}(\alpha)$ exists uniformly on compact subsets of $\mathfrak{S}(]0,1[)$ and that $\alpha \mapsto \sF_{\lambda}(\alpha)$ is analytic on $\mathfrak{S}(]0,1[)$. 
\\
\\
Let $K \subset \mathfrak{S}([0,1[)$ be compact. There exists $\delta \in ]0,1[$ and $k>0$ such that $K \subset \widehat{K} := [0, \delta] + i [-k,k]$. Define
\[ C_{\delta} = \sup_{t\in \mathbb{R}, \alpha \in \mathfrak{S}(]0,\delta[)} |\partial_{\alpha} \sF_{\lambda, t}(\alpha)| . \]
Lemma \ref{lemma2} \textit{(ii)} states that $C_{\delta} < \infty$. For sufficiently small $\epsilon$, $r = \epsilon / 12 C_{\delta} < \delta$, and uniform convergence of $\sF_{\lambda, t}(\alpha)$ on compact subsets gives that there exists $T >0$ such that $| \sF_{\lambda,t}(\alpha)  - \sF_{\lambda}(\alpha) | < \epsilon / 3$ for all $t> T$ and $\alpha \in \widehat{K}_r = \{ \alpha \in \widehat{K} : \textrm{Re} \alpha \geq r\}$. For any $\alpha \in \widehat{K}$ there exists $\alpha' \in \widehat{K}_r$ such that $| \alpha - \alpha' | < 2r$. Then $| \sF_{\lambda,t}(\alpha) - \sF_{\lambda, t}(\alpha')| < 2rC_{\delta} = \epsilon/6$. One has for any $s,t >T$ and $\alpha \in \widehat{K}$, 
\begin{align*}
| \sF_{\lambda, t}(\alpha) - \sF_{\lambda, s}(\alpha) | \leq | \sF_{\lambda, t}(\alpha) - \sF_{\lambda, t}(\alpha')| &+ | \sF_{\lambda,t}(\alpha') - \sF_{\lambda}(\alpha') | \\
& +  |\sF_{\lambda}(\alpha') - \sF_{\lambda, s}(\alpha')| + |\sF_{\lambda, s}(\alpha') - \sF_{\lambda, s}(\alpha) | < \epsilon. 
\end{align*}
Hence $\sF_{\lambda, t}$ converges uniformly on $K$ as $t \rightarrow \infty$ and the limiting function $\sF_{\lambda}$ is continuous. 
\end{proof}
We now take the small coupling-energy limit, $\lambda \rightarrow 0$. 
\begin{leftbar}
\begin{prop}
\label{llimit}
For $\gamma \in \mathbb{R}$, 
\[ \sF(i\gamma/\beta) = \lim_{\lambda \rightarrow 0} \sF_{\lambda}(i\gamma/\beta) = \omega_S(e^{-i\gamma H_S})\omega_{\beta}(^{i\gamma H_S}) . \]
\end{prop}
\end{leftbar}
\begin{proof}
or $ s \in \mathbb{R}$, set $\widehat{\Omega}_{\lambda} = e^{-\beta(L_0 + \lambda \pi(V))/2 } \Omega_{eq}$ and 
\[ \mathcal{G}_{\lambda}^{(1)}(s)= (\widehat{\Omega}, e^{-i\beta s \widehat{L}_{\lambda} } \widehat{\Omega}_{\lambda}), \;\;\;\;\;\;\;\;\;\; \mathcal{G}_{\lambda}^{(2)}(s)= (\widehat{\Omega}_{\lambda}, e^{i\beta s \widehat{L}_{\lambda} } \Omega_{\eta}).\]
We write $\widehat{L}_{\lambda} = L_0 + \lambda \pi(V) - J \pi(H_S)J $ and notice that since $L_0 +  \lambda \pi(V) = L_R + \pi(H_S + \lambda V) - J \pi(H_S) J$, $L_0 + \lambda \pi(V)$ commutes with $ J\pi(H_S) J$. Hence 
\[  \mathcal{G}_{\lambda}^{(1)}(s)= (\widehat{\Omega}, e^{-i\beta s J \pi(H_S) J  } e^{-i\beta s (L_0 + \lambda \pi(V))} \widehat{\Omega}_{\lambda}) . \]
By Thm. \ref{Araki}, the map 
\[\mathbb{R} \ni s \mapsto  \mathcal{G}_{\lambda}^{(1)}(s) = (\widehat{\Omega}, e^{-i\beta s J \pi(H_S) J  } e^{-i\beta (s-i/2)  (L_0 + \lambda \pi(V))} \Omega_{eq}). \]
has an analytic continuation to the strip $0< \textrm{Im}s < \frac{1}{2}$ which is bounded and continuous on the closure. 
Thus, for $\gamma \in \mathbb{R}$, 
\[ \mathcal{G}_{\lambda}^{(1)} \left( \frac{\gamma}{\beta} + \frac{1}{2} i  \right)= (\widehat{\Omega}, e^{ (\beta/2 - i \gamma) J \pi(H_S) J } e^{-i\gamma (L_0 + \lambda \pi(V))} \Omega_{eq})\] 
and hence 
\begin{align}
\label{G1}
 \lim_{\lambda \rightarrow 0} \mathcal{G}_{\lambda}^{(1)} \left( \frac{\gamma}{\beta} + \frac{1}{2} i  \right)= (\widehat{\Omega}, Je^{ (\beta/2 + i \gamma)  \pi(H_S)  } J \Omega_{eq}) = \omega_{\beta}(  \rho_{\beta}^{-1/2} \rho_S e^{ (\beta/2 - i \gamma)  \pi(H_S)  }) = Z^{-1/2} \omega_S(e^{-i\gamma H_S})\
 \end{align}
where $ Z = \tr(e^{-\beta H_S})$. 
\\
\\
Since $\Omega_{\eta} = Z^{1/2}J \pi(e^{\beta/2})J\Omega_{eq}$, one has (with the notation of Lemma \ref{lemma1}), 
\begin{align*}
e^{i\beta s\widehat{L}_\lambda}\Omega_\eta & = U_M^{\beta s}e^{i\beta sL_R}\Omega_\eta \\
& =U_M^{\beta s}\Omega_\eta \\
& =Z^{1/2} U_M^{\beta s}J\pi(e^{\frac\beta2H_A})J\Omega_{eq} \\
& =Z^{1/2}J\pi(e^{\frac\beta2H_S})JU_M^{\beta s}\Omega_{eq}\\
&=Z^{1/2}J\pi(e^{\frac\beta2H_A})J e^{i\beta s\widehat{L}_\lambda}e^{-i\beta sL_R}\Omega_{eq} \\
& =Z^{1/2}J\pi(e^{\frac\beta2H_S})Je^{i\beta s\widehat{L}_\lambda}\Omega_{eq},
\end{align*}
and hence
\begin{align*}
e^{i\beta s\widehat{L}_\lambda}\Omega_\eta 
&=Z^{1/2}J\pi(e^{\frac\beta2H_S})Je^{i\beta sJ\pi(H_S)J}e^{i\beta s(L_0+\lambda\pi(V))}\Omega_{eq}\\
&=Z^{1/2}e^{i\beta(s-i/2)J\pi(H_S)J}e^{i\beta s(L_0+\lambda\pi(V))}\Omega_{eq}.
\end{align*}
Araki's perturbation theory (Thm \ref{Araki}) implies that the function 
\[ s \mapsto \mathcal{G}_{\lambda}^{(2)}(s) = Z^{1/2} ( \widehat{\Omega}_{\lambda} , e^{i\beta(s-i/2)J\pi(H_S)J}e^{i\beta s(L_0+\lambda\pi(V))}\Omega_{eq} ) \]
has an analytic extension to the strip $0<\textrm{Im} s < \frac12$ which is bounded and continuous on the closure. For $\gamma \in \mathbb{R}$, 
\[ \mathcal{G}_{\lambda}^{(2)}\left(\frac{\gamma}{\beta} + \frac12 i \right) = Z^{1/2} ( \widehat{\Omega}_{\lambda} ,  J\pi(e^{-i\gamma H_S})Je^{i\gamma(L_0 + \lambda \pi(V))}\widehat{\Omega}_{\lambda}). \]
Recalling that $\lim_{\lambda \rightarrow 0} \widehat{\Omega}_{\lambda} = \Omega_{eq}$,
\begin{align}
\label{G2}
\lim_{\lambda \rightarrow 0} \mathcal{G}_{\lambda}^{(2)}\left(\frac{\gamma}{\beta} + \frac12 i \right) = Z^{1/2} (\Omega_{eq}, J\pi(e^{-i\gamma H_S})J \Omega_{eq}) = Z^{1/2} \omega_{\beta}(e^{i \gamma H_S}). 
\end{align}
Since for $s \in \mathbb{R}$,
\[ \sF_{\lambda}\left( \frac12 + is\right) = \frac{1}{\| \widehat{\Omega}_{\lambda} \|} \mathcal{G}_{\lambda}^{(1)} (s) \mathcal{G}_{\lambda}^{(2)} (s) \]
analytic continuation to then line $s = \frac{\gamma}{\beta} + \frac12 i$ combined with Equations (\ref{G1}) and (\ref{G2}) gives the result.  
\end{proof}
\begin{proof}[Proof (of Thm \ref{main})]
By Proposition \ref{tlimit}, the map 
\[ \mathbb{R} \ni \alpha \mapsto \sF_{\lambda}(i\alpha)\]
is continuous. Hence there exists a unique Borel probability measure $\mathbb{P}_{R, \lambda}$ on $\mathbb{R}$ satisfying 
\[ \int_{\mathbb{R}} e^{i\gamma x} d\mathbb{P}_{R,\lambda}(x) = \sF_{\lambda}\left( i\gamma/\beta \right) \]
for all $\gamma \in \mathbb{R}$. Since 
\[ \int_{\mathbb{R}} e^{i\gamma x} d\mathbb{P}_{R,\lambda, t }(x) = \sF_{\lambda, t}\left( i\gamma/\beta \right) \]
taking the limit $t \rightarrow \infty$ and Proposition \ref{tlimit} give that 
\[ \lim_{t \rightarrow \infty} \mathbb{P}_{R, \lambda, t} = \mathbb{P}_{R, \lambda}. \]
Then by Proposition \ref{llimit}
\[ \sF(i\gamma/\beta) = \lim_{\lambda \rightarrow 0} \int_{\mathbb{R}} e^{i \gamma x} d\mathbb{P}_{R,\lambda}(x)  = \int_{\mathbb{R}} e^{i \gamma x} d\mathbb{P}_{S}(x)\]
which implies that $\mathbb{P}_R = \mathbb{P}_S$. 
\end{proof}
If we make an additional technical assumption, the proof of Theorem \ref{main} simplifies. 
\begin{leftbar}
\begin{ass}
\label{analytic}
Let $V$ be analytic for the dynamics $\tau_R$. i.e. the map
\[ \mathbb{R} \ni t \mapsto \tau_R^t(V) \in \sO\]
extends to an entire analytic function. 
\end{ass}
\end{leftbar}
Additionally, in this case we can prove convergence of all the moments. 
  \mdfsetup{linewidth=3pt, }
  \begin{mdframed}
  \begin{thm}
  \label{main2}
  Suppose Assumptions (\ref{mixing}), (\ref{regularity}) and (\ref{analytic}) hold and let $0 < |\lambda| < \lambda_0$. The weak limits 
  \begin{align*}
  \mathbb{P}_{R, \lambda} = \lim_{t \rightarrow \infty} \mathbb{P}_{R, \lambda, t}, \;\;\;\;\;\;\;\;\;\;\;\;\;\;\;\;\;\;\;\; \mathbb{P}_{R} = \lim_{\lambda \rightarrow 0} \mathbb{P}_{R, \lambda}
  \end{align*}
  exist and 
  \[\mathbb{P}_{R} = \mathbb{P}_{S}. \]
  Moreover, all the moments of $\mathbb{P}_{R, \lambda,t} / \mathbb{P}_{R, \lambda}$ converge to the corresponding moments of $\mathbb{P}_{R, \lambda} / \mathbb{P}_{R}$. 
  \end{thm}
  \end{mdframed}
  Recall from the proof of Lemma \ref{lemma1} that if $A \in \sM$, then $U_A^t \in \sM$. Therefore, 
  \[ U_M^t (\lambda) = e^{i \beta t (L_R + M) } e^{- i \beta t L_R}= e^{i \beta t \widehat{L}_{\lambda} } e^{- i \beta t L_R} \in \sM\]
   since $M = \pi(H_S) + \lambda \pi(V) \in \sM$. 
  \begin{leftbar}
  \begin{lemma}
  \label{dyson}
  Assume (\ref{analytic}). The map 
  \[ \mathbb{R}^2 \ni (t, \lambda) \mapsto U_M^t(\lambda) \in \sM\]
  extends to an entire analytic function on $\mathbb{C}^2$. 
  \end{lemma}
  \end{leftbar}
  \begin{proof}
  By the Dyson expansion (\cite{JOPP} Chapter 2), one obtains
  \[
\pi^{-1}(U_M^s(\lambda))=\sum_{n\geq 0}(i \beta s)^n\int_{0\leq t_n\leq\cdots\leq t_1\leq 1}\tau_R^{\beta s t_n}(H_{S} +\lambda V)\cdots 
\tau_R^{\beta s t_1}(H_{S} +\lambda V)d t_1\cdots d t_n,
\]
and the result follows immediately from the analyticity assumption. 
  \end{proof}
  Using only Lemma \ref{lemma1} and Lemma \ref{dyson} we can now proof Thm \ref{main2}.
  \begin{proof}[Proof (of Theorem \ref{main2}).]
  First note that 
  \[ 
e^{i\beta s {\widehat L}_\lambda}{\widehat \Omega}=U_M^s(\lambda) {\widehat \Omega}, \;\;\;\;\;\;\;\;\;\;
e^{i\beta s {\widehat L}_\lambda}\Omega_\eta=U_M^s(\lambda) \Omega_\eta.
\]
Lemma \ref{lemma1}, Lemma \ref{dyson} and analytic continuation from $s \in \mathbb{R}$ to $s \in \mathbb{C}$ give that for $\lambda \in \mathbb{R}$, $s \in \mathbb{C}$, 
\[ ( \Omega , \Delta_{\eta\circ \tau_{\lambda}^{-t}|\eta}^{\frac{1}{2}+\i s}\Omega ) = (
U_M^{\bar s}(\lambda) {\widehat \Omega}, e^{-i t L_\lambda}U_M^s(\lambda) \Omega_\eta ).\]
In particular, for $K \subset \mathbb{R} \times \mathbb{C}$ compact 
\begin{align}
\label{vbound}
\sup_{t\in \mathbb{R}, (\lambda, s)\in K}| ( \Omega|\Delta_{\eta\circ \tau_{\lambda}^{-t}|\eta}^{\frac{1}{2}+\i s}\Omega )|<\infty.
\end{align}
Now using Assumption (\ref{mixing}) and Thm. \ref{Koopman}, we compute 
\[ \lim_{t\rightarrow \infty} ( \Omega|\Delta_{\eta\circ \tau_{\lambda}^{-t}|\eta}^{\frac{1}{2}+\i s}\Omega ) 
=( U_M^{\bar s} (\lambda) \widehat \Omega, \Omega_\lambda) ( \Omega_\lambda, U_M^s( \lambda)\Omega_\eta).\]
Setting $s = \frac{\gamma}{\beta} + \frac{i}{2}$ yields that as $t \rightarrow \infty$, $\mathbb{P}_{R, \lambda, t}$ converges weakly to the unique Borel probability measure $\mathbb{P}_{R, \lambda}$ with characteristic function 
\[ \int_{\mathbb{R}} e^{i \gamma x}d {\mathbb P}_{R, \lambda}(x)= ( U_M^{\frac{\gamma}{\beta} -\frac{\i}{2}} (\lambda)\widehat \Omega, \Omega_\lambda) ( \Omega_\lambda, U_M^{\frac{\gamma}{\beta} +\frac{\i}{2}}( \lambda)\Omega_\eta ).\]
Furthermore, Vitali's convergence theorem and  Equation \ref{vbound} imply that the moments of $\mathbb{P}_{R, \lambda, t}$ converge to the corresponding moments of $\mathbb{P}_{R, \lambda}$ as $ t \rightarrow \infty$. 
\\
\\
For $s \in \mathbb{R}$ we can now compute 
\begin{align*}
\lim_{\lambda \rightarrow 0 } \Omega_{\lambda} & = \rho_{\beta}^{1/2} \otimes \Omega_R ,\\
\lim_{\lambda \rightarrow 0 } U_M^s(\lambda) \widehat{\Omega} & = \rho_{\beta}^{-is} \rho_S \otimes \Omega_R, \\
\lim_{\lambda \rightarrow 0 } U_M^s(\lambda)  \Omega_{\eta} & = \rho_{\beta}^{-is} \otimes \Omega_R ,
\end{align*}
and hence
\[ \lim_{\lambda \rightarrow 0 }  (U_M^{\bar s}(\lambda) {\widehat \Omega}, e^{-i t L_\lambda}U_M^s(\lambda) \Omega_\eta ) 
= \tr(\rho_S \rho_{\beta} ^{\frac{1}{2} + is})\tr(\rho_{\beta} ^{\frac{1}{2} - is}).\]
Again, by Vitali's convergence theorem and Equation \ref{vbound} we can extend this identity to all $s \in \mathbb{C}$. Letting $s = \frac{\gamma}{\beta} + \frac{i}{2}$ yields that $\mathbb{P}_{R, \lambda}$ converges weakly to the unique Borel probability measure $\mathbb{P}_R$ with characteristic function 
\[ \int_{\mathbb{R}}e^{i \gamma x} d\mathbb{P}_R(x) = \tr(\rho_S \rho_{\beta}^{i \gamma / \beta}) \tr(\rho_{\beta} \rho_{\beta}^{-i\gamma/\beta}) = \omega_S(e^{-i\gamma H_S}) \omega_{\beta}(e^{i \gamma H_S}) =  \int_{\mathbb{R}}e^{i \gamma x} d\mathbb{P}_S(x).\]
Hence $\mathbb{P}_R = \mathbb{P}_S$ and the convergence of the moments follows from Vitali's convergence theorem. 
  \end{proof}
 
 \section{Conclusion}
Thermodynamics is a theory which describes the bulk behaviour of bodies of matter which are composed of many smaller constituent particles. It is based on the existence of equilibrium states of these systems: states which are characterized by a few macroscopic parameters, and which are stable in time. Four laws give constraints on the macroscopic parameters which allow one to calculate the new equilibrium states when systems are combined or external forces are applied. Classical statistical mechanics uses probabilistic methods applied to large systems of particles whose microscopic dynamics are given by the laws of classical mechanics to derive the laws of thermodynamics. Quantum statistical mechanics aims to do the same thing, but when the microscopic dynamics are given by the laws of quantum mechanics. In the quantum case, there is the usual statistical uncertainty about the exact microscopic state of the system, but in addition each measurement of even a given microscopic state is probabilistic. Furthermore, the non-commutative structure means that the probability distribution of measurement outcomes depends on previous measurements. 
\\
\\
We focus on the first law of thermodynamics. It states that the internal energy of an isolated system (one which does no work and to which no heat is supplied) is conserved. We consider a specific set-up in which a small system is coupled to a thermal reservoir. Furthermore we assume that the system equilibrates after sufficiently large time to the reservoir temperature. The small system is modelled by the usual Hamiltonian formulation of quantum mechanics, but the reservoir is idealized as infinitely extended and modelled by a $\Cstar$-dynamical system. In this setting, one way to formulate the first law is in terms of the change in expected value of energy when measurements are made at time $0$ and $t$ of the system and reservoir: $\Delta Q_{S, \lambda, t}$ and $\Delta Q_{R, \lambda, t}$. The first law then states that 
\[ \lim_{\lambda \rightarrow 0} \lim_{t \rightarrow \infty} \Delta Q_{S, \lambda, t} =  \lim_{\lambda \rightarrow 0}  \lim_{t \rightarrow \infty} \Delta Q_{R, \lambda, t}.\]
This is well-known to hold. We can also consider probability measures on $\mathbb{R}$ which completely describe the energy fluctuations, i.e. the full counting statistics. The system full counting statistics is defined by applying the basic axioms of quantum mechanics to a two-time measurement protocol. The reservoir full counting statistics is defined by rewriting the finite dimensional FCS as the spectral measure of a relative modular operator, which is an object that still makes sense in the thermodynamic limit. This is one of the novel ideas in our approach. We then prove weak convergence as $t \rightarrow \infty$ and $\lambda \rightarrow 0$ of the measures themselves, which is a much stronger statement than simply proving convergence of the averages. This is done under very general assumptions: a dynamical assumption which guarantees the return to equilibrium behaviour, and a regularity assumption which is needed in order for $\Delta Q_{R, \lambda,t}$ to even be defined. If we add an extra technical analyticity assumption, the proof simplifies and we also obtain convergence of all the moments of the FCS measures. 
\\
\\
Despite this generality, there are still some physical systems that cannot be described by the class of models under study. Specifically, there are some system-reservoir couplings which are not given by a bounded perturbation of the form needed to apply Araki's perturbation theory of KMS states. To accommodate these, one needs to use a $W^*$-dynamical system and an affiliated unbounded perturbation. The modified perturbation theory of Derezi\'nski, Jak\v si\'c and Pillet \cite{DJP} is expected to allow the same method of proof to be used in this case. 

\end{document}